\documentclass{lmcs} 
\pdfoutput=1

\usepackage{lastpage}
\lmcsdoi{20}{4}{3}
\lmcsheading{}{\pageref{LastPage}}{}{}%
{Dec.~27,~2023}{Oct.~04,~2024}{}

\usepackage[utf8]{inputenc}

\usepackage[T1]{fontenc}


\usepackage{hyperref}

\usepackage{amsmath}
\usepackage{stmaryrd}

\usepackage{isabelle, isabellesym}

\usepackage{tikz-cd}
\usepackage{tikz}
\usepackage{subcaption}
\usepackage{adjustbox}

\usepackage{tikzit}

\tikzstyle{new style 0}=[fill=white, draw=black, shape=rectangle]

\tikzstyle{right arrow}=[->, shorten >=10pt, shorten <=10pt]
\tikzstyle{new edge style 0}=[-, shorten <=10pt]
\tikzstyle{new edge style 1}=[->, shorten >=10pt]

\newcommand{\DefineSnippet}[2]{%
  \expandafter\newcommand\csname snippet--#1\endcsname{%
  \begin{quote}
  \vspace{2pt}
      \begin{isabelle}
        \scalebox{1.0}{\begin{minipage}{\textwidth}\textcolor{blue}{#2}\end{minipage}}
      \end{isabelle}
      \vspace{2pt}
    \end{quote}}
    
    }
\newcommand{\Snippet}[1]{%
        \ifcsname snippet--#1\endcsname{\csname snippet--#1\endcsname}%
  \else+++++++ERROR: Snippet ``#1 not defined+++++++ \fi}

\DefineSnippet{pullback}{
\isacommand{locale}\isamarkupfalse%
\ pullback{\isacharunderscore}{\kern0pt}diagram\ {\isacharequal}{\kern0pt}\isanewline
\ \ b{\isacharcolon}{\kern0pt}\ morphism\ A\ B\ b\ {\isacharplus}{\kern0pt}\isanewline
\ \ c{\isacharcolon}{\kern0pt}\ morphism\ A\ C\ c\ {\isacharplus}{\kern0pt}\isanewline
\ \ f{\isacharcolon}{\kern0pt}\ morphism\ B\ D\ f\ {\isacharplus}{\kern0pt}\isanewline
\ \ g{\isacharcolon}{\kern0pt}\ morphism\ C\ D\ g\ \isakeyword{for}\ A\ B\ C\ D\ b\ c\ f\ g\ {\isacharplus}{\kern0pt}\isanewline
\isakeyword{assumes}\isanewline
\ \ node{\isacharunderscore}{\kern0pt}commutativity{\isacharcolon}{\kern0pt}\ {\isacartoucheopen}{\isasymAnd}v{\isachardot}{\kern0pt}\ v\ {\isasymin}\ V\isactrlbsub A\isactrlesub \ {\isasymLongrightarrow}\ \isactrlbsub f\ {\isasymcirc}\isactrlsub {\isasymrightarrow}\ b\isactrlesub \isactrlsub V\ v\ {\isacharequal}{\kern0pt}\ \isactrlbsub g\ {\isasymcirc}\isactrlsub {\isasymrightarrow}\ c\isactrlesub \isactrlsub V\ v{\isacartoucheclose}\ \isakeyword{and}\isanewline
\ \ edge{\isacharunderscore}{\kern0pt}commutativity{\isacharcolon}{\kern0pt}\ {\isacartoucheopen}{\isasymAnd}e{\isachardot}{\kern0pt}\ e\ {\isasymin}\ E\isactrlbsub A\isactrlesub \ {\isasymLongrightarrow}\ \isactrlbsub f\ {\isasymcirc}\isactrlsub {\isasymrightarrow}\ b\isactrlesub \isactrlsub E\ e\ {\isacharequal}{\kern0pt}\ \isactrlbsub g\ {\isasymcirc}\isactrlsub {\isasymrightarrow}\ c\isactrlesub \isactrlsub E\ e{\isacartoucheclose}\ \isakeyword{and}\isanewline
\ \ universal{\isacharunderscore}{\kern0pt}property{\isacharcolon}{\kern0pt}\ {\isacartoucheopen}{\isasymlbrakk}\isanewline
\ \ \ \ graph\ {\isacharparenleft}{\kern0pt}A{\isacharprime}{\kern0pt}\ {\isacharcolon}{\kern0pt}{\isacharcolon}{\kern0pt}\ {\isacharparenleft}{\kern0pt}{\isacharprime}{\kern0pt}c{\isacharcomma}{\kern0pt}{\isacharprime}{\kern0pt}d{\isacharparenright}{\kern0pt}\ ngraph{\isacharparenright}{\kern0pt}{\isacharsemicolon}{\kern0pt}\ \isanewline
\ \ \ \ morphism\ A{\isacharprime}{\kern0pt}\ C\ c{\isacharprime}{\kern0pt}{\isacharsemicolon}{\kern0pt}\ \isanewline
\ \ \ \ morphism\ A{\isacharprime}{\kern0pt}\ B\ b{\isacharprime}{\kern0pt}{\isacharsemicolon}{\kern0pt}\isanewline
\ \ \ \ \ {\isasymAnd}v{\isachardot}{\kern0pt}\ v\ {\isasymin}\ V\isactrlbsub A{\isacharprime}{\kern0pt}\isactrlesub \ {\isasymLongrightarrow}\ \isactrlbsub f\ {\isasymcirc}\isactrlsub {\isasymrightarrow}\ b{\isacharprime}{\kern0pt}\isactrlesub \isactrlsub V\ v\ {\isacharequal}{\kern0pt}\ \isactrlbsub g\ {\isasymcirc}\isactrlsub {\isasymrightarrow}\ c{\isacharprime}{\kern0pt}\isactrlesub \isactrlsub V\ v{\isacharsemicolon}{\kern0pt}\isanewline
\ \ \ \ \ {\isasymAnd}e{\isachardot}{\kern0pt}\ e\ {\isasymin}\ E\isactrlbsub A{\isacharprime}{\kern0pt}\isactrlesub \ {\isasymLongrightarrow}\ \isactrlbsub f\ {\isasymcirc}\isactrlsub {\isasymrightarrow}\ b{\isacharprime}{\kern0pt}\isactrlesub \isactrlsub E\ e\ {\isacharequal}{\kern0pt}\ \isactrlbsub g\ {\isasymcirc}\isactrlsub {\isasymrightarrow}\ c{\isacharprime}{\kern0pt}\isactrlesub \isactrlsub E\ e{\isasymrbrakk}\ \isanewline
\ \ \ \ {\isasymLongrightarrow}\ Ex{\isadigit{1}}M\ {\isacharparenleft}{\kern0pt}{\isasymlambda}u{\isachardot}{\kern0pt}\ morphism\ A{\isacharprime}{\kern0pt}\ A\ u\ {\isasymand}\isanewline
\ \ \ \ \ \ \ \ \ \ \ \ {\isacharparenleft}{\kern0pt}{\isasymforall}v\ {\isasymin}\ V\isactrlbsub A{\isacharprime}{\kern0pt}\isactrlesub {\isachardot}{\kern0pt}\ \isactrlbsub b\ {\isasymcirc}\isactrlsub {\isasymrightarrow}\ u\isactrlesub \isactrlsub V\ v\ {\isacharequal}{\kern0pt}\ \isactrlbsub b{\isacharprime}{\kern0pt}\isactrlesub \isactrlsub V\ v{\isacharparenright}{\kern0pt}\ {\isasymand}\isanewline
\ \ \ \ \ \ \ \ \ \ \ \ {\isacharparenleft}{\kern0pt}{\isasymforall}e\ {\isasymin}\ E\isactrlbsub A{\isacharprime}{\kern0pt}\isactrlesub {\isachardot}{\kern0pt}\ \isactrlbsub b\ {\isasymcirc}\isactrlsub {\isasymrightarrow}\ u\isactrlesub \isactrlsub E\ e\ {\isacharequal}{\kern0pt}\ \isactrlbsub b{\isacharprime}{\kern0pt}\isactrlesub \isactrlsub E\ e{\isacharparenright}{\kern0pt}\ {\isasymand}\isanewline
\ \ \ \ \ \ \ \ \ \ \ \ {\isacharparenleft}{\kern0pt}{\isasymforall}v\ {\isasymin}\ V\isactrlbsub A{\isacharprime}{\kern0pt}\isactrlesub {\isachardot}{\kern0pt}\ \isactrlbsub c\ {\isasymcirc}\isactrlsub {\isasymrightarrow}\ u\isactrlesub \isactrlsub V\ v\ {\isacharequal}{\kern0pt}\ \isactrlbsub c{\isacharprime}{\kern0pt}\isactrlesub \isactrlsub V\ v{\isacharparenright}{\kern0pt}\ {\isasymand}\isanewline
\ \ \ \ \ \ \ \ \ \ \ \ {\isacharparenleft}{\kern0pt}{\isasymforall}e\ {\isasymin}\ E\isactrlbsub A{\isacharprime}{\kern0pt}\isactrlesub {\isachardot}{\kern0pt}\ \isactrlbsub c\ {\isasymcirc}\isactrlsub {\isasymrightarrow}\ u\isactrlesub \isactrlsub E\ e\ {\isacharequal}{\kern0pt}\ \isactrlbsub c{\isacharprime}{\kern0pt}\isactrlesub \isactrlsub E\ e{\isacharparenright}{\kern0pt}{\isacharparenright}{\kern0pt}\isanewline
\ \ \ \ \ \ \ \ \ \ \ \ A{\isacharprime}{\kern0pt}{\isacartoucheclose}%
}
\DefineSnippet{pullback-unique}{
\isacommand{theorem}\isamarkupfalse%
\ uniqueness{\isacharunderscore}{\kern0pt}pb{\isacharcolon}{\kern0pt}\isanewline
\ \ \isakeyword{fixes}\ A{\isacharprime}{\kern0pt}\ b{\isacharprime}{\kern0pt}\ c{\isacharprime}{\kern0pt}\isanewline
\ \ \isakeyword{assumes}\isanewline
\ \ \ \ A{\isacharprime}{\kern0pt}{\isacharcolon}{\kern0pt}\ {\isacartoucheopen}graph\ A{\isacharprime}{\kern0pt}{\isacartoucheclose}\ \isakeyword{and}\isanewline
\ \ \ \ b{\isacharprime}{\kern0pt}{\isacharcolon}{\kern0pt}\ {\isacartoucheopen}morphism\ A{\isacharprime}{\kern0pt}\ B\ b{\isacharprime}{\kern0pt}{\isacartoucheclose}\ \isakeyword{and}\isanewline
\ \ \ \ c{\isacharprime}{\kern0pt}{\isacharcolon}{\kern0pt}\ {\isacartoucheopen}morphism\ A{\isacharprime}{\kern0pt}\ C\ c{\isacharprime}{\kern0pt}{\isacartoucheclose}\isanewline
\ \ \isakeyword{shows}\ {\isacartoucheopen}pullback{\isacharunderscore}{\kern0pt}diagram\ A{\isacharprime}{\kern0pt}\ B\ C\ D\ b{\isacharprime}{\kern0pt}\ c{\isacharprime}{\kern0pt}\ f\ g\ \isanewline
\ \ \ \ \ \ \ \ \ \ {\isasymlongleftrightarrow}\ {\isacharparenleft}{\kern0pt}{\isasymexists}u{\isachardot}{\kern0pt}\ bijective{\isacharunderscore}{\kern0pt}morphism\ A{\isacharprime}{\kern0pt}\ A\ u\ \isanewline
\ \ \ \ \ \ \ \ \ \ \ \ \ \ \ \ {\isasymand}\ {\isacharparenleft}{\kern0pt}{\isasymforall}v\ {\isasymin}\ V\isactrlbsub A{\isacharprime}{\kern0pt}\isactrlesub {\isachardot}{\kern0pt}\ \isactrlbsub b\ {\isasymcirc}\isactrlsub {\isasymrightarrow}\ u\isactrlesub \isactrlsub V\ v\ {\isacharequal}{\kern0pt}\ \isactrlbsub b{\isacharprime}{\kern0pt}\isactrlesub \isactrlsub V\ v{\isacharparenright}{\kern0pt}\ \isanewline
\ \ \ \ \ \ \ \ \ \ \ \ \ \ \ \ {\isasymand}\ {\isacharparenleft}{\kern0pt}{\isasymforall}e\ {\isasymin}\ E\isactrlbsub A{\isacharprime}{\kern0pt}\isactrlesub {\isachardot}{\kern0pt}\ \isactrlbsub b\ {\isasymcirc}\isactrlsub {\isasymrightarrow}\ u\isactrlesub \isactrlsub E\ e\ {\isacharequal}{\kern0pt}\ \isactrlbsub b{\isacharprime}{\kern0pt}\isactrlesub \isactrlsub E\ e{\isacharparenright}{\kern0pt}\isanewline
\ \ \ \ \ \ \ \ \ \ \ \ \ \ \ \ {\isasymand}\ {\isacharparenleft}{\kern0pt}{\isasymforall}v\ {\isasymin}\ V\isactrlbsub A{\isacharprime}{\kern0pt}\isactrlesub {\isachardot}{\kern0pt}\ \isactrlbsub c\ {\isasymcirc}\isactrlsub {\isasymrightarrow}\ u\isactrlesub \isactrlsub V\ v\ {\isacharequal}{\kern0pt}\ \isactrlbsub c{\isacharprime}{\kern0pt}\isactrlesub \isactrlsub V\ v{\isacharparenright}{\kern0pt}\ \isanewline
\ \ \ \ \ \ \ \ \ \ \ \ \ \ \ \ {\isasymand}\ {\isacharparenleft}{\kern0pt}{\isasymforall}e\ {\isasymin}\ E\isactrlbsub A{\isacharprime}{\kern0pt}\isactrlesub {\isachardot}{\kern0pt}\ \isactrlbsub c\ {\isasymcirc}\isactrlsub {\isasymrightarrow}\ u\isactrlesub \isactrlsub E\ e\ {\isacharequal}{\kern0pt}\ \isactrlbsub c{\isacharprime}{\kern0pt}\isactrlesub \isactrlsub E\ e{\isacharparenright}{\kern0pt}{\isacharparenright}{\kern0pt}{\isacartoucheclose}%
}
\DefineSnippet{special-diag-is-pb}{
\isacommand{lemma}\isamarkupfalse%
\ fun{\isacharunderscore}{\kern0pt}algrtr{\isacharunderscore}{\kern0pt}{\isadigit{4}}{\isacharunderscore}{\kern0pt}{\isadigit{7}}{\isacharunderscore}{\kern0pt}{\isadigit{2}}{\isacharcolon}{\kern0pt}\isanewline
\ \ \isakeyword{fixes}\ C\ A\ m\isanewline
\ \ \isakeyword{assumes}\ {\isacartoucheopen}injective{\isacharunderscore}{\kern0pt}morphism\ C\ A\ m{\isacartoucheclose}\isanewline
\ \ \isakeyword{shows}\ {\isacartoucheopen}pullback{\isacharunderscore}{\kern0pt}diagram\ C\ C\ C\ A\ idM\ idM\ m\ m{\isacartoucheclose}%
}
\DefineSnippet{pullback-composition}{
\isacommand{lemma}\isamarkupfalse%
\ pullback{\isacharunderscore}{\kern0pt}composition{\isacharcolon}{\kern0pt}\isanewline
\ \ \isakeyword{assumes}\ \isanewline
\ \ \ \ {\isadigit{1}}{\isacharcolon}{\kern0pt}\ {\isacartoucheopen}pullback{\isacharunderscore}{\kern0pt}diagram\ A\ B\ C\ D\ f\ g\ g{\isacharprime}{\kern0pt}\ f{\isacharprime}{\kern0pt}{\isacartoucheclose}\ \isakeyword{and}\isanewline
\ \ \ \ {\isadigit{2}}{\isacharcolon}{\kern0pt}\ {\isacartoucheopen}pullback{\isacharunderscore}{\kern0pt}diagram\ B\ E\ D\ F\ e\ g{\isacharprime}{\kern0pt}\ e{\isacharprime}{\kern0pt}{\isacharprime}{\kern0pt}\ e{\isacharprime}{\kern0pt}{\isacartoucheclose}\isanewline
\ \ \isakeyword{shows}\ {\isacartoucheopen}pullback{\isacharunderscore}{\kern0pt}diagram\ A\ E\ C\ F\ {\isacharparenleft}{\kern0pt}e\ {\isasymcirc}\isactrlsub {\isasymrightarrow}\ f{\isacharparenright}{\kern0pt}\ g\ e{\isacharprime}{\kern0pt}{\isacharprime}{\kern0pt}\ {\isacharparenleft}{\kern0pt}e{\isacharprime}{\kern0pt}\ {\isasymcirc}\isactrlsub {\isasymrightarrow}\ f{\isacharprime}{\kern0pt}{\isacharparenright}{\kern0pt}{\isacartoucheclose}%
}
\DefineSnippet{pullback-decomposition}{
\isacommand{lemma}\isamarkupfalse%
\ pullback{\isacharunderscore}{\kern0pt}decomposition{\isacharcolon}{\kern0pt}\isanewline
\ \ \isakeyword{assumes}\isanewline
\ \ \ \ f{\isacharcolon}{\kern0pt}\ \ {\isacartoucheopen}morphism\ A\ B\ f{\isacartoucheclose}\ \isakeyword{and}\isanewline
\ \ \ \ f{\isacharprime}{\kern0pt}{\isacharcolon}{\kern0pt}\ {\isacartoucheopen}morphism\ C\ D\ f{\isacharprime}{\kern0pt}{\isacartoucheclose}\ \isakeyword{and}\isanewline
\ \ \ \ {\isadigit{2}}{\isacharcolon}{\kern0pt}\ {\isacartoucheopen}pullback{\isacharunderscore}{\kern0pt}diagram\ B\ E\ D\ F\ e\ g{\isacharprime}{\kern0pt}\ e{\isacharprime}{\kern0pt}{\isacharprime}{\kern0pt}\ e{\isacharprime}{\kern0pt}{\isacartoucheclose}\ \isakeyword{and}\isanewline
\ \ \ \ {\isachardoublequoteopen}{\isadigit{1}}{\isacharplus}{\kern0pt}{\isadigit{2}}{\isachardoublequoteclose}{\isacharcolon}{\kern0pt}\ \ {\isacartoucheopen}pullback{\isacharunderscore}{\kern0pt}diagram\ A\ E\ C\ F\ {\isacharparenleft}{\kern0pt}e\ {\isasymcirc}\isactrlsub {\isasymrightarrow}\ f{\isacharparenright}{\kern0pt}\ g\ \ e{\isacharprime}{\kern0pt}{\isacharprime}{\kern0pt}\ {\isacharparenleft}{\kern0pt}e{\isacharprime}{\kern0pt}\ {\isasymcirc}\isactrlsub {\isasymrightarrow}\ f{\isacharprime}{\kern0pt}{\isacharparenright}{\kern0pt}{\isacartoucheclose}\ \isakeyword{and}\isanewline
\ \ \ \ {\isachardoublequoteopen}{\isadigit{1}}cv{\isachardoublequoteclose}{\isacharcolon}{\kern0pt}\ {\isacartoucheopen}{\isasymAnd}v{\isachardot}{\kern0pt}\ v\ {\isasymin}\ V\isactrlbsub A\isactrlesub \ {\isasymLongrightarrow}\ \isactrlbsub g{\isacharprime}{\kern0pt}\ {\isasymcirc}\isactrlsub {\isasymrightarrow}\ f\isactrlesub \isactrlsub V\ v\ {\isacharequal}{\kern0pt}\ \isactrlbsub f{\isacharprime}{\kern0pt}\ {\isasymcirc}\isactrlsub {\isasymrightarrow}\ g\isactrlesub \isactrlsub V\ v{\isacartoucheclose}\ \isakeyword{and}\isanewline
\ \ \ \ {\isachardoublequoteopen}{\isadigit{1}}ce{\isachardoublequoteclose}{\isacharcolon}{\kern0pt}\ {\isacartoucheopen}{\isasymAnd}ea{\isachardot}{\kern0pt}\ ea\ {\isasymin}\ E\isactrlbsub A\isactrlesub \ {\isasymLongrightarrow}\ \isactrlbsub g{\isacharprime}{\kern0pt}\ {\isasymcirc}\isactrlsub {\isasymrightarrow}\ f\isactrlesub \isactrlsub E\ ea\ {\isacharequal}{\kern0pt}\ \isactrlbsub f{\isacharprime}{\kern0pt}\ {\isasymcirc}\isactrlsub {\isasymrightarrow}\ g\isactrlesub \isactrlsub E\ ea{\isacartoucheclose}\isanewline
\ \ \isakeyword{shows}\ {\isacartoucheopen}pullback{\isacharunderscore}{\kern0pt}diagram\ A\ B\ C\ D\ f\ g\ g{\isacharprime}{\kern0pt}\ f{\isacharprime}{\kern0pt}{\isacartoucheclose}%
}
\DefineSnippet{gluing-locale}{
\isacommand{locale}\isamarkupfalse%
\ gluing\ {\isacharequal}{\kern0pt}\isanewline
\ \ d{\isacharcolon}{\kern0pt}\ injective{\isacharunderscore}{\kern0pt}morphism\ K\ D\ d\ {\isacharplus}{\kern0pt}\isanewline
\ \ r{\isacharcolon}{\kern0pt}\ injective{\isacharunderscore}{\kern0pt}morphism\ K\ R\ b\isanewline
\ \ \isakeyword{for}\ K\ D\ R\ d\ b%
}
\DefineSnippet{gluing-sets}{
\isacommand{abbreviation}\isamarkupfalse%
\ V\ \isakeyword{where}\ {\isacartoucheopen}V\ {\isasymequiv}\ V\isactrlbsub D\isactrlesub \ {\isacharless}{\kern0pt}{\isacharplus}{\kern0pt}{\isachargreater}{\kern0pt}\ {\isacharparenleft}{\kern0pt}V\isactrlbsub R\isactrlesub \ {\isacharminus}{\kern0pt}\ \isactrlbsub b\isactrlesub \isactrlsub V\ {\isacharbackquote}{\kern0pt}\ V\isactrlbsub K\isactrlesub {\isacharparenright}{\kern0pt}{\isacartoucheclose}\isanewline
\isacommand{abbreviation}\isamarkupfalse%
\ E\ \isakeyword{where}\ {\isacartoucheopen}E\ {\isasymequiv}\ E\isactrlbsub D\isactrlesub \ {\isacharless}{\kern0pt}{\isacharplus}{\kern0pt}{\isachargreater}{\kern0pt}\ {\isacharparenleft}{\kern0pt}E\isactrlbsub R\isactrlesub \ {\isacharminus}{\kern0pt}\ \isactrlbsub b\isactrlesub \isactrlsub E\ {\isacharbackquote}{\kern0pt}\ E\isactrlbsub K\isactrlesub {\isacharparenright}{\kern0pt}{\isacartoucheclose}%
}
\DefineSnippet{gluing-source}{
\isacommand{fun}\isamarkupfalse%
\ s\ \isakeyword{where}\isanewline
\ {\isachardoublequoteopen}s\ {\isacharparenleft}{\kern0pt}Inl\ e{\isacharparenright}{\kern0pt}\ {\isacharequal}{\kern0pt}\ Inl\ {\isacharparenleft}{\kern0pt}s\isactrlbsub D\isactrlesub \ e{\isacharparenright}{\kern0pt}{\isachardoublequoteclose}\isanewline
{\isacharbar}{\kern0pt}{\isachardoublequoteopen}s\ {\isacharparenleft}{\kern0pt}Inr\ e{\isacharparenright}{\kern0pt}\ {\isacharequal}{\kern0pt}\ {\isacharparenleft}{\kern0pt}if\ e\ {\isasymin}\ {\isacharparenleft}{\kern0pt}E\isactrlbsub R\isactrlesub \ {\isacharminus}{\kern0pt}\ \isactrlbsub b\isactrlesub \isactrlsub E\ {\isacharbackquote}{\kern0pt}\ E\isactrlbsub K\isactrlesub {\isacharparenright}{\kern0pt}\ {\isasymand}\ {\isacharparenleft}{\kern0pt}s\isactrlbsub R\isactrlesub \ e\ {\isasymin}\ \isactrlbsub b\isactrlesub \isactrlsub V\ {\isacharbackquote}{\kern0pt}\ V\isactrlbsub K\isactrlesub {\isacharparenright}{\kern0pt}\ \isanewline
\ \ then\ Inl\ {\isacharparenleft}{\kern0pt}\isactrlbsub d\isactrlesub \isactrlsub V\ {\isacharparenleft}{\kern0pt}{\isacharparenleft}{\kern0pt}inv{\isacharunderscore}{\kern0pt}into\ V\isactrlbsub K\isactrlesub \ \isactrlbsub b\isactrlesub \isactrlsub V{\isacharparenright}{\kern0pt}\ {\isacharparenleft}{\kern0pt}s\isactrlbsub R\isactrlesub \ e{\isacharparenright}{\kern0pt}{\isacharparenright}{\kern0pt}{\isacharparenright}{\kern0pt}\ else\ Inr\ {\isacharparenleft}{\kern0pt}s\isactrlbsub R\isactrlesub \ e{\isacharparenright}{\kern0pt}{\isacharparenright}{\kern0pt}{\isachardoublequoteclose}%
}
\DefineSnippet{gluing-label}{
\isacommand{fun}\isamarkupfalse%
\ l\ \isakeyword{where}\isanewline
\ \ {\isachardoublequoteopen}l\ {\isacharparenleft}{\kern0pt}Inl\ v{\isacharparenright}{\kern0pt}\ {\isacharequal}{\kern0pt}\ l\isactrlbsub D\isactrlesub \ v{\isachardoublequoteclose}\isanewline
{\isacharbar}{\kern0pt}\ {\isachardoublequoteopen}l\ {\isacharparenleft}{\kern0pt}Inr\ v{\isacharparenright}{\kern0pt}\ {\isacharequal}{\kern0pt}\ l\isactrlbsub R\isactrlesub \ v{\isachardoublequoteclose}%
}
\DefineSnippet{gluing-H}{
\isacommand{definition}\isamarkupfalse%
\ H\ \isakeyword{where}\ \isanewline
{\isacartoucheopen}H\ {\isasymequiv}\ {\isasymlparr}nodes{\isacharequal}{\kern0pt}V{\isacharcomma}{\kern0pt}edges{\isacharequal}{\kern0pt}E\isanewline
\ \ \ \ \ {\isacharcomma}{\kern0pt}source{\isacharequal}{\kern0pt}s{\isacharcomma}{\kern0pt}target{\isacharequal}{\kern0pt}t\isanewline
\ \ \ \ \ {\isacharcomma}{\kern0pt}node{\isacharunderscore}{\kern0pt}label{\isacharequal}{\kern0pt}l{\isacharcomma}{\kern0pt}edge{\isacharunderscore}{\kern0pt}label{\isacharequal}{\kern0pt}m{\isasymrparr}{\isacartoucheclose}%
}
\DefineSnippet{gluing-H-sublocale}{
\isacommand{sublocale}\isamarkupfalse%
\ h{\isacharcolon}{\kern0pt}\ graph\ H%
}
\DefineSnippet{gluing-morph-h}{
\isacommand{definition}\isamarkupfalse%
\ h\ \isakeyword{where}\ \isanewline
{\isacartoucheopen}h\ {\isasymequiv}\ {\isasymlparr}node{\isacharunderscore}{\kern0pt}map\ {\isacharequal}{\kern0pt}\ {\isasymlambda}v{\isachardot}{\kern0pt}\ if\ v\ {\isasymin}\ V\isactrlbsub R\isactrlesub \ {\isacharminus}{\kern0pt}\ \isactrlbsub b\isactrlesub \isactrlsub V\ {\isacharbackquote}{\kern0pt}\ V\isactrlbsub K\isactrlesub \ \isanewline
\ \ \ \ \ \ \ \ \ \ \ \ \ \ \ \ \ \ \ \ \ then\ Inr\ v\ \isanewline
\ \ \ \ \ \ \ \ \ \ \ \ \ \ \ \ \ \ \ \ \ else\ Inl\ {\isacharparenleft}{\kern0pt}\isactrlbsub d\isactrlesub \isactrlsub V\ {\isacharparenleft}{\kern0pt}{\isacharparenleft}{\kern0pt}inv{\isacharunderscore}{\kern0pt}into\ V\isactrlbsub K\isactrlesub \ \isactrlbsub b\isactrlesub \isactrlsub V{\isacharparenright}{\kern0pt}\ v{\isacharparenright}{\kern0pt}{\isacharparenright}{\kern0pt}{\isacharcomma}{\kern0pt}\isanewline
\ \ \ \ \ \ edge{\isacharunderscore}{\kern0pt}map\ {\isacharequal}{\kern0pt}\ {\isasymlambda}e{\isachardot}{\kern0pt}\ if\ e\ {\isasymin}\ E\isactrlbsub R\isactrlesub \ {\isacharminus}{\kern0pt}\ \isactrlbsub b\isactrlesub \isactrlsub E\ {\isacharbackquote}{\kern0pt}\ E\isactrlbsub K\isactrlesub \ \isanewline
\ \ \ \ \ \ \ \ \ \ \ \ \ \ \ \ \ \ \ \ \ then\ Inr\ e\ \isanewline
\ \ \ \ \ \ \ \ \ \ \ \ \ \ \ \ \ \ \ \ \ else\ Inl\ {\isacharparenleft}{\kern0pt}\isactrlbsub d\isactrlesub \isactrlsub E\ {\isacharparenleft}{\kern0pt}{\isacharparenleft}{\kern0pt}inv{\isacharunderscore}{\kern0pt}into\ E\isactrlbsub K\isactrlesub \ \isactrlbsub b\isactrlesub \isactrlsub E{\isacharparenright}{\kern0pt}\ e{\isacharparenright}{\kern0pt}{\isacharparenright}{\kern0pt}{\isasymrparr}{\isacartoucheclose}%
}
\DefineSnippet{gluing-morph-c}{
\isacommand{definition}\isamarkupfalse%
\ c\ {\isacharcolon}{\kern0pt}{\isacharcolon}{\kern0pt}\ {\isachardoublequoteopen}{\isacharparenleft}{\kern0pt}{\isacharprime}{\kern0pt}e{\isacharcomma}{\kern0pt}\ {\isacharprime}{\kern0pt}e\ {\isacharplus}{\kern0pt}\ {\isacharprime}{\kern0pt}g{\isacharcomma}{\kern0pt}\ {\isacharprime}{\kern0pt}f{\isacharcomma}{\kern0pt}\ {\isacharprime}{\kern0pt}f\ {\isacharplus}{\kern0pt}\ {\isacharprime}{\kern0pt}h{\isacharparenright}{\kern0pt}\ pre{\isacharunderscore}{\kern0pt}morph{\isachardoublequoteclose}\ \isakeyword{where}\ \isanewline
\ \ {\isacartoucheopen}c\ {\isasymequiv}\ {\isasymlparr}node{\isacharunderscore}{\kern0pt}map\ {\isacharequal}{\kern0pt}\ Inl{\isacharcomma}{\kern0pt}\ edge{\isacharunderscore}{\kern0pt}map\ {\isacharequal}{\kern0pt}\ Inl{\isasymrparr}{\isacartoucheclose}%
}
\DefineSnippet{inj-po-is-pb}{
\isacommand{lemma}\isamarkupfalse%
\ pushout{\isacharunderscore}{\kern0pt}pullback{\isacharunderscore}{\kern0pt}inj{\isacharunderscore}{\kern0pt}b{\isacharcolon}{\kern0pt}\isanewline
\ \ \isakeyword{assumes}\ \isanewline
\ \ \ \ b{\isacharcolon}{\kern0pt}\ {\isacartoucheopen}injective{\isacharunderscore}{\kern0pt}morphism\ A\ B\ b{\isacartoucheclose}\ \isakeyword{and}\isanewline
\ \ \ \ c{\isacharcolon}{\kern0pt}\ {\isacartoucheopen}injective{\isacharunderscore}{\kern0pt}morphism\ A\ C\ c{\isacartoucheclose}\isanewline
\ \ \isakeyword{shows}\ {\isacartoucheopen}pullback{\isacharunderscore}{\kern0pt}diagram\ A\ B\ C\ D\ b\ c\ f\ g{\isacartoucheclose}%
}
\DefineSnippet{pushout-characterization}{
\isacommand{lemma}\isamarkupfalse%
\ po{\isacharunderscore}{\kern0pt}characterization{\isacharcolon}{\kern0pt}\isanewline
\ \ \isakeyword{assumes}\ \isanewline
\ \ \ \ b{\isacharcolon}{\kern0pt}\ {\isacartoucheopen}injective{\isacharunderscore}{\kern0pt}morphism\ A\ B\ b{\isacartoucheclose}\ \isakeyword{and}\isanewline
\ \ \ \ c{\isacharcolon}{\kern0pt}\ {\isacartoucheopen}injective{\isacharunderscore}{\kern0pt}morphism\ A\ C\ c{\isacartoucheclose}\ \isakeyword{and}\isanewline
\ \ \ \ f{\isacharcolon}{\kern0pt}\ {\isacartoucheopen}injective{\isacharunderscore}{\kern0pt}morphism\ B\ D\ f{\isacartoucheclose}\ \isakeyword{and}\isanewline
\ \ \ \ g{\isacharcolon}{\kern0pt}\ {\isacartoucheopen}injective{\isacharunderscore}{\kern0pt}morphism\ C\ D\ g{\isacartoucheclose}\ \isakeyword{and}\isanewline
\ \ \ \ node{\isacharunderscore}{\kern0pt}commutativity{\isacharcolon}{\kern0pt}\ {\isacartoucheopen}{\isasymAnd}v{\isachardot}{\kern0pt}\ v\ {\isasymin}\ V\isactrlbsub A\isactrlesub \ {\isasymLongrightarrow}\ \isactrlbsub f\ {\isasymcirc}\isactrlsub {\isasymrightarrow}\ b\isactrlesub \isactrlsub V\ v\ {\isacharequal}{\kern0pt}\ \isactrlbsub g\ {\isasymcirc}\isactrlsub {\isasymrightarrow}\ c\isactrlesub \isactrlsub V\ v{\isacartoucheclose}\ \isakeyword{and}\isanewline
\ \ \ \ edge{\isacharunderscore}{\kern0pt}commutativity{\isacharcolon}{\kern0pt}\ {\isacartoucheopen}{\isasymAnd}e{\isachardot}{\kern0pt}\ e\ {\isasymin}\ E\isactrlbsub A\isactrlesub \ {\isasymLongrightarrow}\ \isactrlbsub f\ {\isasymcirc}\isactrlsub {\isasymrightarrow}\ b\isactrlesub \isactrlsub E\ e\ {\isacharequal}{\kern0pt}\ \isactrlbsub g\ {\isasymcirc}\isactrlsub {\isasymrightarrow}\ c\isactrlesub \isactrlsub E\ e{\isacartoucheclose}\ \isakeyword{and}\isanewline
\ \ \ \ reduced{\isacharunderscore}{\kern0pt}chain{\isacharunderscore}{\kern0pt}condition{\isacharunderscore}{\kern0pt}nodes{\isacharcolon}{\kern0pt}\ \isanewline
\ \ \ \ \ \ {\isacartoucheopen}{\isasymAnd}x\ y{\isachardot}{\kern0pt}\ x\ {\isasymin}\ V\isactrlbsub B\isactrlesub \ {\isasymLongrightarrow}\ y\ {\isasymin}\ V\isactrlbsub C\isactrlesub \ {\isasymLongrightarrow}\ \isactrlbsub f\isactrlesub \isactrlsub V\ x\ {\isacharequal}{\kern0pt}\ \isactrlbsub g\isactrlesub \isactrlsub V\ y\ \isanewline
\ \ \ \ \ \ \ \ \ \ \ \ \ \ \ \ \ \ \ \ \ \ \ \ \ {\isasymLongrightarrow}\ {\isacharparenleft}{\kern0pt}{\isasymexists}a\ {\isasymin}\ V\isactrlbsub A\isactrlesub {\isachardot}{\kern0pt}\ {\isacharparenleft}{\kern0pt}\isactrlbsub b\isactrlesub \isactrlsub V\ a\ {\isacharequal}{\kern0pt}\ x\ {\isasymand}\ \isactrlbsub c\isactrlesub \isactrlsub V\ a\ {\isacharequal}{\kern0pt}\ y{\isacharparenright}{\kern0pt}{\isacharparenright}{\kern0pt}{\isacartoucheclose}\ \isakeyword{and}\isanewline
\ \ \ \ reduced{\isacharunderscore}{\kern0pt}chain{\isacharunderscore}{\kern0pt}condition{\isacharunderscore}{\kern0pt}edges{\isacharcolon}{\kern0pt}\ \isanewline
\ \ \ \ \ \ {\isacartoucheopen}{\isasymAnd}x\ y{\isachardot}{\kern0pt}\ x\ {\isasymin}\ E\isactrlbsub B\isactrlesub \ {\isasymLongrightarrow}\ y\ {\isasymin}\ E\isactrlbsub C\isactrlesub \ {\isasymLongrightarrow}\ \isactrlbsub f\isactrlesub \isactrlsub E\ x\ {\isacharequal}{\kern0pt}\ \isactrlbsub g\isactrlesub \isactrlsub E\ y\ \isanewline
\ \ \ \ \ \ \ \ \ \ \ \ \ \ \ \ \ \ \ \ \ \ \ \ \ {\isasymLongrightarrow}\ {\isacharparenleft}{\kern0pt}{\isasymexists}a\ {\isasymin}\ E\isactrlbsub A\isactrlesub {\isachardot}{\kern0pt}\ {\isacharparenleft}{\kern0pt}\isactrlbsub b\isactrlesub \isactrlsub E\ a\ {\isacharequal}{\kern0pt}\ x\ {\isasymand}\ \isactrlbsub c\isactrlesub \isactrlsub E\ a\ {\isacharequal}{\kern0pt}\ y{\isacharparenright}{\kern0pt}{\isacharparenright}{\kern0pt}{\isacartoucheclose}\ \isakeyword{and}\isanewline
\ \ \ \ joint{\isacharunderscore}{\kern0pt}surjectivity{\isacharunderscore}{\kern0pt}nodes{\isacharcolon}{\kern0pt}\ \isanewline
\ \ \ \ \ \ \ \ \ \ {\isacartoucheopen}{\isasymAnd}x{\isachardot}{\kern0pt}\ x\ {\isasymin}\ V\isactrlbsub D\isactrlesub \ {\isasymLongrightarrow}\ {\isacharparenleft}{\kern0pt}{\isasymexists}v\ {\isasymin}\ V\isactrlbsub C\isactrlesub {\isachardot}{\kern0pt}\ \isactrlbsub g\isactrlesub \isactrlsub V\ v\ {\isacharequal}{\kern0pt}\ x{\isacharparenright}{\kern0pt}\ {\isasymor}\ {\isacharparenleft}{\kern0pt}{\isasymexists}v\ {\isasymin}\ V\isactrlbsub B\isactrlesub {\isachardot}{\kern0pt}\ \isactrlbsub f\isactrlesub \isactrlsub V\ v\ {\isacharequal}{\kern0pt}\ x{\isacharparenright}{\kern0pt}{\isacartoucheclose}\ \isakeyword{and}\isanewline
\ \ \ \ joint{\isacharunderscore}{\kern0pt}surjectivity{\isacharunderscore}{\kern0pt}edges{\isacharcolon}{\kern0pt}\ \isanewline
\ \ \ \ \ \ \ \ \ \ {\isacartoucheopen}{\isasymAnd}x{\isachardot}{\kern0pt}\ x\ {\isasymin}\ E\isactrlbsub D\isactrlesub \ {\isasymLongrightarrow}\ {\isacharparenleft}{\kern0pt}{\isasymexists}e\ {\isasymin}\ E\isactrlbsub C\isactrlesub {\isachardot}{\kern0pt}\ \isactrlbsub g\isactrlesub \isactrlsub E\ e\ {\isacharequal}{\kern0pt}\ x{\isacharparenright}{\kern0pt}\ {\isasymor}\ {\isacharparenleft}{\kern0pt}{\isasymexists}e\ {\isasymin}\ E\isactrlbsub B\isactrlesub {\isachardot}{\kern0pt}\ \isactrlbsub f\isactrlesub \isactrlsub E\ e\ {\isacharequal}{\kern0pt}\ x{\isacharparenright}{\kern0pt}{\isacartoucheclose}\isanewline
\ \ \isakeyword{shows}\ {\isacartoucheopen}pushout{\isacharunderscore}{\kern0pt}diagram\ A\ B\ C\ D\ b\ c\ f\ g{\isacartoucheclose}%
}
\DefineSnippet{pushout-comp-uniq}{
\isacommand{theorem}\isamarkupfalse%
\ uniqueness{\isacharunderscore}{\kern0pt}pc{\isacharcolon}{\kern0pt}\isanewline
\ \ \isakeyword{fixes}\ C{\isacharprime}{\kern0pt}\ c{\isacharprime}{\kern0pt}\ g{\isacharprime}{\kern0pt}\isanewline
\ \ \isakeyword{assumes}\ \isanewline
\ \ \ \ b{\isacharcolon}{\kern0pt}\ {\isacartoucheopen}injective{\isacharunderscore}{\kern0pt}morphism\ A\ B\ b{\isacartoucheclose}\ \isakeyword{and}\isanewline
\ \ \ \ c{\isacharcolon}{\kern0pt}\ {\isacartoucheopen}injective{\isacharunderscore}{\kern0pt}morphism\ A\ C\ c{\isacartoucheclose}\ \isakeyword{and}\isanewline
\ \ \ \ C{\isacharprime}{\kern0pt}{\isacharcolon}{\kern0pt}\ {\isacartoucheopen}graph\ C{\isacharprime}{\kern0pt}{\isacartoucheclose}\ \isakeyword{and}\isanewline
\ \ \ \ c{\isacharprime}{\kern0pt}{\isacharcolon}{\kern0pt}\ {\isacartoucheopen}injective{\isacharunderscore}{\kern0pt}morphism\ A\ C{\isacharprime}{\kern0pt}\ c{\isacharprime}{\kern0pt}{\isacartoucheclose}\ \isakeyword{and}\isanewline
\ \ \ \ g{\isacharprime}{\kern0pt}{\isacharcolon}{\kern0pt}\ {\isacartoucheopen}morphism\ C{\isacharprime}{\kern0pt}\ D\ g{\isacharprime}{\kern0pt}{\isacartoucheclose}\isanewline
\ \ \isakeyword{shows}\ {\isacartoucheopen}pushout{\isacharunderscore}{\kern0pt}diagram\ A\ B\ C{\isacharprime}{\kern0pt}\ D\ b\ c{\isacharprime}{\kern0pt}\ f\ g{\isacharprime}{\kern0pt}\ \isanewline
\ \ \ \ \ \ \ \ \ \ {\isasymlongrightarrow}\ {\isacharparenleft}{\kern0pt}{\isasymexists}u{\isachardot}{\kern0pt}\ bijective{\isacharunderscore}{\kern0pt}morphism\ C\ C{\isacharprime}{\kern0pt}\ u{\isacharparenright}{\kern0pt}{\isacartoucheclose}%
}
\DefineSnippet{pre-graph}{
\isacommand{record}\isamarkupfalse%
\ {\isacharparenleft}{\kern0pt}{\isacharprime}{\kern0pt}v{\isacharcomma}{\kern0pt}{\isacharprime}{\kern0pt}e{\isacharcomma}{\kern0pt}{\isacharprime}{\kern0pt}l{\isacharcomma}{\kern0pt}{\isacharprime}{\kern0pt}m{\isacharparenright}{\kern0pt}\ pre{\isacharunderscore}{\kern0pt}graph\ {\isacharequal}{\kern0pt}\isanewline
\ \ nodes\ \ {\isacharcolon}{\kern0pt}{\isacharcolon}{\kern0pt}\ {\isachardoublequoteopen}{\isacharprime}{\kern0pt}v\ set{\isachardoublequoteclose}\isanewline
\ \ edges\ \ {\isacharcolon}{\kern0pt}{\isacharcolon}{\kern0pt}\ {\isachardoublequoteopen}{\isacharprime}{\kern0pt}e\ set{\isachardoublequoteclose}\isanewline
\ \ source\ {\isacharcolon}{\kern0pt}{\isacharcolon}{\kern0pt}\ {\isachardoublequoteopen}{\isacharprime}{\kern0pt}e\ {\isasymRightarrow}\ {\isacharprime}{\kern0pt}v{\isachardoublequoteclose}\isanewline
\ \ target\ {\isacharcolon}{\kern0pt}{\isacharcolon}{\kern0pt}\ {\isachardoublequoteopen}{\isacharprime}{\kern0pt}e\ {\isasymRightarrow}\ {\isacharprime}{\kern0pt}v{\isachardoublequoteclose}\isanewline
\ \ node{\isacharunderscore}{\kern0pt}label\ {\isacharcolon}{\kern0pt}{\isacharcolon}{\kern0pt}\ {\isachardoublequoteopen}{\isacharprime}{\kern0pt}v\ {\isasymRightarrow}\ {\isacharprime}{\kern0pt}l{\isachardoublequoteclose}\isanewline
\ \ edge{\isacharunderscore}{\kern0pt}label\ {\isacharcolon}{\kern0pt}{\isacharcolon}{\kern0pt}\ {\isachardoublequoteopen}{\isacharprime}{\kern0pt}e\ {\isasymRightarrow}\ {\isacharprime}{\kern0pt}m{\isachardoublequoteclose}%
}
\DefineSnippet{graph-syntax}{
\isacommand{notation}\isamarkupfalse%
\ nodes\ \ {\isacharparenleft}{\kern0pt}{\isachardoublequoteopen}V\isactrlbsub {\isacharunderscore}{\kern0pt}\isactrlesub {\isachardoublequoteclose}{\isacharparenright}{\kern0pt}\isanewline
\isacommand{notation}\isamarkupfalse%
\ edges\ \ {\isacharparenleft}{\kern0pt}{\isachardoublequoteopen}E\isactrlbsub {\isacharunderscore}{\kern0pt}\isactrlesub {\isachardoublequoteclose}{\isacharparenright}{\kern0pt}\isanewline
\isacommand{notation}\isamarkupfalse%
\ source\ {\isacharparenleft}{\kern0pt}{\isachardoublequoteopen}s\isactrlbsub {\isacharunderscore}{\kern0pt}\isactrlesub {\isachardoublequoteclose}{\isacharparenright}{\kern0pt}\isanewline
\isacommand{notation}\isamarkupfalse%
\ target\ {\isacharparenleft}{\kern0pt}{\isachardoublequoteopen}t\isactrlbsub {\isacharunderscore}{\kern0pt}\isactrlesub {\isachardoublequoteclose}{\isacharparenright}{\kern0pt}\isanewline
\isacommand{notation}\isamarkupfalse%
\ node{\isacharunderscore}{\kern0pt}label\ {\isacharparenleft}{\kern0pt}{\isachardoublequoteopen}l\isactrlbsub {\isacharunderscore}{\kern0pt}\isactrlesub {\isachardoublequoteclose}{\isacharparenright}{\kern0pt}\isanewline
\isacommand{notation}\isamarkupfalse%
\ edge{\isacharunderscore}{\kern0pt}label\ {\isacharparenleft}{\kern0pt}{\isachardoublequoteopen}m\isactrlbsub {\isacharunderscore}{\kern0pt}\isactrlesub {\isachardoublequoteclose}{\isacharparenright}{\kern0pt}%
}
\DefineSnippet{graph-locale}{
\isacommand{locale}\isamarkupfalse%
\ graph\ {\isacharequal}{\kern0pt}\isanewline
\ \ \isakeyword{fixes}\ G\ {\isacharcolon}{\kern0pt}{\isacharcolon}{\kern0pt}\ {\isachardoublequoteopen}{\isacharparenleft}{\kern0pt}{\isacharprime}{\kern0pt}v{\isacharcolon}{\kern0pt}{\isacharcolon}{\kern0pt}countable{\isacharcomma}{\kern0pt}{\isacharprime}{\kern0pt}e{\isacharcolon}{\kern0pt}{\isacharcolon}{\kern0pt}countable{\isacharcomma}{\kern0pt}{\isacharprime}{\kern0pt}l{\isacharcomma}{\kern0pt}{\isacharprime}{\kern0pt}m{\isacharparenright}{\kern0pt}\ pre{\isacharunderscore}{\kern0pt}graph{\isachardoublequoteclose}\isanewline
\ \ \isakeyword{assumes}\ \isanewline
\ \ \ \ finite{\isacharunderscore}{\kern0pt}nodes{\isacharcolon}{\kern0pt}\ {\isachardoublequoteopen}finite\ V\isactrlbsub G\isactrlesub {\isachardoublequoteclose}\ \isakeyword{and}\isanewline
\ \ \ \ finite{\isacharunderscore}{\kern0pt}edges{\isacharcolon}{\kern0pt}\ {\isachardoublequoteopen}finite\ E\isactrlbsub G\isactrlesub {\isachardoublequoteclose}\ \isakeyword{and}\isanewline
\ \ \ \ source{\isacharunderscore}{\kern0pt}integrity{\isacharcolon}{\kern0pt}\ {\isachardoublequoteopen}e\ {\isasymin}\ E\isactrlbsub G\isactrlesub \ {\isasymLongrightarrow}\ s\isactrlbsub G\isactrlesub \ e\ {\isasymin}\ V\isactrlbsub G\isactrlesub {\isachardoublequoteclose}\ \isakeyword{and}\isanewline
\ \ \ \ target{\isacharunderscore}{\kern0pt}integrity{\isacharcolon}{\kern0pt}\ {\isachardoublequoteopen}e\ {\isasymin}\ E\isactrlbsub G\isactrlesub \ {\isasymLongrightarrow}\ t\isactrlbsub G\isactrlesub \ e\ {\isasymin}\ V\isactrlbsub G\isactrlesub {\isachardoublequoteclose}%
}
\DefineSnippet{pullback-construction}{
\isacommand{locale}\isamarkupfalse%
\ pullback{\isacharunderscore}{\kern0pt}construction\ {\isacharequal}{\kern0pt}\isanewline
\ \ f{\isacharcolon}{\kern0pt}\ morphism\ B\ D\ f\ {\isacharplus}{\kern0pt}\isanewline
\ \ g{\isacharcolon}{\kern0pt}\ morphism\ C\ D\ g\ \isanewline
\ \ \isakeyword{for}\ B\ D\ C\ f\ g%
}
\DefineSnippet{pullback-construction-comp}{
\isacommand{abbreviation}\isamarkupfalse%
\ V\ \isakeyword{where}\ \isanewline
\ \ {\isacartoucheopen}V\ {\isasymequiv}\ {\isacharbraceleft}{\kern0pt}{\isacharparenleft}{\kern0pt}x{\isacharcomma}{\kern0pt}y{\isacharparenright}{\kern0pt}{\isachardot}{\kern0pt}\ x\ {\isasymin}\ V\isactrlbsub B\isactrlesub \ {\isasymand}\ y\ {\isasymin}\ V\isactrlbsub C\isactrlesub \ {\isasymand}\ \isactrlbsub f\isactrlesub \isactrlsub V\ x\ {\isacharequal}{\kern0pt}\ \isactrlbsub g\isactrlesub \isactrlsub V\ y{\isacharbraceright}{\kern0pt}{\isacartoucheclose}\isanewline
\isanewline
\isacommand{abbreviation}\isamarkupfalse%
\ E\ \isakeyword{where}\ \isanewline
\ \ {\isacartoucheopen}E\ {\isasymequiv}\ {\isacharbraceleft}{\kern0pt}{\isacharparenleft}{\kern0pt}x{\isacharcomma}{\kern0pt}y{\isacharparenright}{\kern0pt}{\isachardot}{\kern0pt}\ x\ {\isasymin}\ E\isactrlbsub B\isactrlesub \ {\isasymand}\ y\ {\isasymin}\ E\isactrlbsub C\isactrlesub \ {\isasymand}\ \isactrlbsub f\isactrlesub \isactrlsub E\ x\ {\isacharequal}{\kern0pt}\ \isactrlbsub g\isactrlesub \isactrlsub E\ y{\isacharbraceright}{\kern0pt}{\isacartoucheclose}%
}
\DefineSnippet{pullback-construction-comp2}{
\isacommand{fun}\isamarkupfalse%
\ s\ \isakeyword{where}\ {\isacartoucheopen}s\ {\isacharparenleft}{\kern0pt}x{\isacharcomma}{\kern0pt}y{\isacharparenright}{\kern0pt}\ {\isacharequal}{\kern0pt}\ {\isacharparenleft}{\kern0pt}s\isactrlbsub B\isactrlesub \ x{\isacharcomma}{\kern0pt}\ s\isactrlbsub C\isactrlesub \ y{\isacharparenright}{\kern0pt}{\isacartoucheclose}\isanewline
\isacommand{fun}\isamarkupfalse%
\ t\ \isakeyword{where}\ {\isacartoucheopen}t\ {\isacharparenleft}{\kern0pt}x{\isacharcomma}{\kern0pt}y{\isacharparenright}{\kern0pt}\ {\isacharequal}{\kern0pt}\ {\isacharparenleft}{\kern0pt}t\isactrlbsub B\isactrlesub \ x{\isacharcomma}{\kern0pt}\ t\isactrlbsub C\isactrlesub \ y{\isacharparenright}{\kern0pt}{\isacartoucheclose}\isanewline
\isacommand{fun}\isamarkupfalse%
\ l\ \isakeyword{where}\ {\isacartoucheopen}l\ {\isacharparenleft}{\kern0pt}x{\isacharcomma}{\kern0pt}{\isacharunderscore}{\kern0pt}{\isacharparenright}{\kern0pt}\ {\isacharequal}{\kern0pt}\ l\isactrlbsub B\isactrlesub \ x{\isacartoucheclose}\isanewline
\isacommand{fun}\isamarkupfalse%
\ m\ \isakeyword{where}\ {\isacartoucheopen}m\ {\isacharparenleft}{\kern0pt}x{\isacharcomma}{\kern0pt}{\isacharunderscore}{\kern0pt}{\isacharparenright}{\kern0pt}\ {\isacharequal}{\kern0pt}\ m\isactrlbsub B\isactrlesub \ x{\isacartoucheclose}%
}
\DefineSnippet{pullback-construction-obj}{
\isacommand{definition}\isamarkupfalse%
\ A\ \isakeyword{where}\ \isanewline
\ \ {\isacartoucheopen}A\ {\isasymequiv}\ {\isasymlparr}nodes\ {\isacharequal}{\kern0pt}\ V{\isacharcomma}{\kern0pt}\ edges\ {\isacharequal}{\kern0pt}\ E{\isacharcomma}{\kern0pt}\ source\ {\isacharequal}{\kern0pt}\ s{\isacharcomma}{\kern0pt}\ target\ {\isacharequal}{\kern0pt}\ t\isanewline
\ \ \ \ \ \ \ {\isacharcomma}{\kern0pt}node{\isacharunderscore}{\kern0pt}label\ {\isacharequal}{\kern0pt}\ l{\isacharcomma}{\kern0pt}\ edge{\isacharunderscore}{\kern0pt}label\ {\isacharequal}{\kern0pt}\ m{\isasymrparr}{\isacartoucheclose}%
}
\DefineSnippet{pullback-morphism-b-def}{
\isacommand{definition}\isamarkupfalse%
\ b\ {\isacharcolon}{\kern0pt}{\isacharcolon}{\kern0pt}\ {\isachardoublequoteopen}{\isacharparenleft}{\kern0pt}{\isacharprime}{\kern0pt}a\ {\isasymtimes}\ {\isacharprime}{\kern0pt}g{\isacharcomma}{\kern0pt}\ {\isacharprime}{\kern0pt}a{\isacharcomma}{\kern0pt}\ {\isacharprime}{\kern0pt}b\ {\isasymtimes}\ {\isacharprime}{\kern0pt}h{\isacharcomma}{\kern0pt}\ {\isacharprime}{\kern0pt}b{\isacharparenright}{\kern0pt}\ pre{\isacharunderscore}{\kern0pt}morph{\isachardoublequoteclose}\ \isakeyword{where}\ \isanewline
\ \ {\isacartoucheopen}b\ {\isasymequiv}\ {\isasymlparr}node{\isacharunderscore}{\kern0pt}map\ {\isacharequal}{\kern0pt}\ fst{\isacharcomma}{\kern0pt}\ edge{\isacharunderscore}{\kern0pt}map\ {\isacharequal}{\kern0pt}\ fst{\isasymrparr}{\isacartoucheclose}%
}
\DefineSnippet{pullback-morphism-c-def}{
\isacommand{definition}\isamarkupfalse%
\ c\ {\isacharcolon}{\kern0pt}{\isacharcolon}{\kern0pt}\ {\isachardoublequoteopen}{\isacharparenleft}{\kern0pt}{\isacharprime}{\kern0pt}a\ {\isasymtimes}\ {\isacharprime}{\kern0pt}g{\isacharcomma}{\kern0pt}\ {\isacharprime}{\kern0pt}g{\isacharcomma}{\kern0pt}\ {\isacharprime}{\kern0pt}b\ {\isasymtimes}\ {\isacharprime}{\kern0pt}h{\isacharcomma}{\kern0pt}\ {\isacharprime}{\kern0pt}h{\isacharparenright}{\kern0pt}\ pre{\isacharunderscore}{\kern0pt}morph{\isachardoublequoteclose}\isanewline
\ \ \isakeyword{where}\ {\isacartoucheopen}c\ {\isasymequiv}\ {\isasymlparr}node{\isacharunderscore}{\kern0pt}map\ {\isacharequal}{\kern0pt}\ snd{\isacharcomma}{\kern0pt}\ edge{\isacharunderscore}{\kern0pt}map\ {\isacharequal}{\kern0pt}\ snd{\isasymrparr}{\isacartoucheclose}%
}
\DefineSnippet{pullback-construction-correctness}{
\isacommand{sublocale}\isamarkupfalse%
\ pb{\isacharcolon}{\kern0pt}\ pullback{\isacharunderscore}{\kern0pt}diagram\ A\ B\ C\ D\ b\ c\ f\ g%
}
\DefineSnippet{pbconstr-rchain}{
\isacommand{lemma}\isamarkupfalse%
\ reduced{\isacharunderscore}{\kern0pt}chain{\isacharunderscore}{\kern0pt}condition{\isacharunderscore}{\kern0pt}nodes{\isacharcolon}{\kern0pt}\isanewline
\ \ \isakeyword{fixes}\ x\ y\isanewline
\ \ \isakeyword{assumes}\ {\isacartoucheopen}x\ {\isasymin}\ V\isactrlbsub B\isactrlesub {\isacartoucheclose}\ {\isacartoucheopen}y\ {\isasymin}\ V\isactrlbsub C\isactrlesub {\isacartoucheclose}\ {\isacartoucheopen}\isactrlbsub f\isactrlesub \isactrlsub V\ x\ {\isacharequal}{\kern0pt}\ \isactrlbsub g\isactrlesub \isactrlsub V\ y{\isacartoucheclose}\isanewline
\ \ \isakeyword{shows}\ {\isacartoucheopen}{\isasymexists}a\ {\isasymin}\ V\isactrlbsub A\isactrlesub {\isachardot}{\kern0pt}\ {\isacharparenleft}{\kern0pt}\isactrlbsub b\isactrlesub \isactrlsub V\ a\ {\isacharequal}{\kern0pt}\ x\ {\isasymand}\ \isactrlbsub c\isactrlesub \isactrlsub V\ a\ {\isacharequal}{\kern0pt}\ y{\isacharparenright}{\kern0pt}{\isacartoucheclose}%
}
\DefineSnippet{pullback-reduced-chain}{
\isacommand{lemma}\isamarkupfalse%
\ {\isacharparenleft}{\kern0pt}\isakeyword{in}\ pullback{\isacharunderscore}{\kern0pt}diagram{\isacharparenright}{\kern0pt}\ reduced{\isacharunderscore}{\kern0pt}chain{\isacharunderscore}{\kern0pt}condition{\isacharunderscore}{\kern0pt}nodes{\isacharcolon}{\kern0pt}\isanewline
\ \ \isakeyword{fixes}\ x\ y\isanewline
\ \ \isakeyword{assumes}\ {\isacartoucheopen}x\ {\isasymin}\ V\isactrlbsub B\isactrlesub {\isacartoucheclose}\ {\isacartoucheopen}y\ {\isasymin}\ V\isactrlbsub C\isactrlesub {\isacartoucheclose}\ {\isacartoucheopen}\isactrlbsub f\isactrlesub \isactrlsub V\ x\ {\isacharequal}{\kern0pt}\ \isactrlbsub g\isactrlesub \isactrlsub V\ y{\isacartoucheclose}\isanewline
\ \ \isakeyword{shows}\ {\isacartoucheopen}{\isasymexists}a\ {\isasymin}\ V\isactrlbsub A\isactrlesub {\isachardot}{\kern0pt}\ {\isacharparenleft}{\kern0pt}\isactrlbsub b\isactrlesub \isactrlsub V\ a\ {\isacharequal}{\kern0pt}\ x\ {\isasymand}\ \isactrlbsub c\isactrlesub \isactrlsub V\ a\ {\isacharequal}{\kern0pt}\ y{\isacharparenright}{\kern0pt}{\isacartoucheclose}%
}
\DefineSnippet{pre-rule}{
\isacommand{record}\isamarkupfalse%
\ {\isacharparenleft}{\kern0pt}{\isacharprime}{\kern0pt}v\isactrlsub {\isadigit{1}}{\isacharcomma}{\kern0pt}{\isacharprime}{\kern0pt}e\isactrlsub {\isadigit{1}}{\isacharcomma}{\kern0pt}{\isacharprime}{\kern0pt}v\isactrlsub {\isadigit{2}}{\isacharcomma}{\kern0pt}{\isacharprime}{\kern0pt}e\isactrlsub {\isadigit{2}}{\isacharcomma}{\kern0pt}{\isacharprime}{\kern0pt}v\isactrlsub {\isadigit{3}}{\isacharcomma}{\kern0pt}{\isacharprime}{\kern0pt}e\isactrlsub {\isadigit{3}}{\isacharcomma}{\kern0pt}{\isacharprime}{\kern0pt}l{\isacharcomma}{\kern0pt}{\isacharprime}{\kern0pt}m{\isacharparenright}{\kern0pt}\ pre{\isacharunderscore}{\kern0pt}rule\ {\isacharequal}{\kern0pt}\isanewline
\ \ lhs\ \ \ \ {\isacharcolon}{\kern0pt}{\isacharcolon}{\kern0pt}\ {\isachardoublequoteopen}{\isacharparenleft}{\kern0pt}{\isacharprime}{\kern0pt}v\isactrlsub {\isadigit{1}}{\isacharcomma}{\kern0pt}{\isacharprime}{\kern0pt}e\isactrlsub {\isadigit{1}}{\isacharcomma}{\kern0pt}{\isacharprime}{\kern0pt}l{\isacharcomma}{\kern0pt}{\isacharprime}{\kern0pt}m{\isacharparenright}{\kern0pt}\ pre{\isacharunderscore}{\kern0pt}graph{\isachardoublequoteclose}\isanewline
\ \ interf\ {\isacharcolon}{\kern0pt}{\isacharcolon}{\kern0pt}\ {\isachardoublequoteopen}{\isacharparenleft}{\kern0pt}{\isacharprime}{\kern0pt}v\isactrlsub {\isadigit{2}}{\isacharcomma}{\kern0pt}{\isacharprime}{\kern0pt}e\isactrlsub {\isadigit{2}}{\isacharcomma}{\kern0pt}{\isacharprime}{\kern0pt}l{\isacharcomma}{\kern0pt}{\isacharprime}{\kern0pt}m{\isacharparenright}{\kern0pt}\ pre{\isacharunderscore}{\kern0pt}graph{\isachardoublequoteclose}\isanewline
\ \ rhs\ \ \ \ {\isacharcolon}{\kern0pt}{\isacharcolon}{\kern0pt}\ {\isachardoublequoteopen}{\isacharparenleft}{\kern0pt}{\isacharprime}{\kern0pt}v\isactrlsub {\isadigit{3}}{\isacharcomma}{\kern0pt}{\isacharprime}{\kern0pt}e\isactrlsub {\isadigit{3}}{\isacharcomma}{\kern0pt}{\isacharprime}{\kern0pt}l{\isacharcomma}{\kern0pt}{\isacharprime}{\kern0pt}m{\isacharparenright}{\kern0pt}\ pre{\isacharunderscore}{\kern0pt}graph{\isachardoublequoteclose}%
}
\DefineSnippet{rule-locale}{
\isacommand{locale}\isamarkupfalse%
\ rule\ {\isacharequal}{\kern0pt}\isanewline
\ \ k{\isacharcolon}{\kern0pt}\ injective{\isacharunderscore}{\kern0pt}morphism\ {\isachardoublequoteopen}interf\ r{\isachardoublequoteclose}\ {\isachardoublequoteopen}lhs\ r{\isachardoublequoteclose}\ b\ {\isacharplus}{\kern0pt}\isanewline
\ \ r{\isacharcolon}{\kern0pt}\ injective{\isacharunderscore}{\kern0pt}morphism\ {\isachardoublequoteopen}interf\ r{\isachardoublequoteclose}\ {\isachardoublequoteopen}rhs\ r{\isachardoublequoteclose}\ b{\isacharprime}{\kern0pt}\isanewline
\ \ \isakeyword{for}\ r\ {\isacharcolon}{\kern0pt}{\isacharcolon}{\kern0pt}\ {\isachardoublequoteopen}{\isacharparenleft}{\kern0pt}{\isacharprime}{\kern0pt}v\isactrlsub {\isadigit{1}}{\isacharcolon}{\kern0pt}{\isacharcolon}{\kern0pt}countable{\isacharcomma}{\kern0pt}{\isacharprime}{\kern0pt}e\isactrlsub {\isadigit{1}}{\isacharcolon}{\kern0pt}{\isacharcolon}{\kern0pt}countable\isanewline
\ \ \ \ \ \ \ \ \ \ \ \ {\isacharcomma}{\kern0pt}{\isacharprime}{\kern0pt}v\isactrlsub {\isadigit{2}}{\isacharcolon}{\kern0pt}{\isacharcolon}{\kern0pt}countable{\isacharcomma}{\kern0pt}{\isacharprime}{\kern0pt}e\isactrlsub {\isadigit{2}}{\isacharcolon}{\kern0pt}{\isacharcolon}{\kern0pt}countable\isanewline
\ \ \ \ \ \ \ \ \ \ \ \ {\isacharcomma}{\kern0pt}{\isacharprime}{\kern0pt}v\isactrlsub {\isadigit{3}}{\isacharcolon}{\kern0pt}{\isacharcolon}{\kern0pt}countable{\isacharcomma}{\kern0pt}{\isacharprime}{\kern0pt}e\isactrlsub {\isadigit{3}}{\isacharcolon}{\kern0pt}{\isacharcolon}{\kern0pt}countable\isanewline
\ \ \ \ \ \ \ \ \ \ \ \ {\isacharcomma}{\kern0pt}{\isacharprime}{\kern0pt}l{\isacharcomma}{\kern0pt}{\isacharprime}{\kern0pt}m{\isacharparenright}{\kern0pt}\ pre{\isacharunderscore}{\kern0pt}rule{\isachardoublequoteclose}\ \isakeyword{and}\ b\ b{\isacharprime}{\kern0pt}%
}
\DefineSnippet{pre-morph}{
\isacommand{record}\isamarkupfalse%
\ {\isacharparenleft}{\kern0pt}{\isacharprime}{\kern0pt}v\isactrlsub {\isadigit{1}}{\isacharcomma}{\kern0pt}{\isacharprime}{\kern0pt}v\isactrlsub {\isadigit{2}}{\isacharcomma}{\kern0pt}{\isacharprime}{\kern0pt}e\isactrlsub {\isadigit{1}}{\isacharcomma}{\kern0pt}{\isacharprime}{\kern0pt}e\isactrlsub {\isadigit{2}}{\isacharparenright}{\kern0pt}\ pre{\isacharunderscore}{\kern0pt}morph\ {\isacharequal}{\kern0pt}\isanewline
\ \ node{\isacharunderscore}{\kern0pt}map\ {\isacharcolon}{\kern0pt}{\isacharcolon}{\kern0pt}\ {\isachardoublequoteopen}{\isacharprime}{\kern0pt}v\isactrlsub {\isadigit{1}}\ {\isasymRightarrow}\ {\isacharprime}{\kern0pt}v\isactrlsub {\isadigit{2}}{\isachardoublequoteclose}\isanewline
\ \ edge{\isacharunderscore}{\kern0pt}map\ {\isacharcolon}{\kern0pt}{\isacharcolon}{\kern0pt}\ {\isachardoublequoteopen}{\isacharprime}{\kern0pt}e\isactrlsub {\isadigit{1}}\ {\isasymRightarrow}\ {\isacharprime}{\kern0pt}e\isactrlsub {\isadigit{2}}{\isachardoublequoteclose}%
}
\DefineSnippet{morph-syntax}{
\isacommand{notation}\isamarkupfalse%
\ node{\isacharunderscore}{\kern0pt}map\ {\isacharparenleft}{\kern0pt}{\isachardoublequoteopen}\isactrlbsub {\isacharunderscore}{\kern0pt}\isactrlesub \isactrlsub V{\isachardoublequoteclose}{\isacharparenright}{\kern0pt}\isanewline
\isacommand{notation}\isamarkupfalse%
\ edge{\isacharunderscore}{\kern0pt}map\ {\isacharparenleft}{\kern0pt}{\isachardoublequoteopen}\isactrlbsub {\isacharunderscore}{\kern0pt}\isactrlesub \isactrlsub E{\isachardoublequoteclose}{\isacharparenright}{\kern0pt}%
}
\DefineSnippet{morph-locale}{
\isacommand{locale}\isamarkupfalse%
\ morphism\ {\isacharequal}{\kern0pt}\isanewline
\ \ G{\isacharcolon}{\kern0pt}\ graph\ G\ {\isacharplus}{\kern0pt}\isanewline
\ \ H{\isacharcolon}{\kern0pt}\ graph\ H\ \isakeyword{for}\isanewline
\ \ \ \ G\ {\isacharcolon}{\kern0pt}{\isacharcolon}{\kern0pt}\ {\isachardoublequoteopen}{\isacharparenleft}{\kern0pt}{\isacharprime}{\kern0pt}v\isactrlsub {\isadigit{1}}{\isacharcolon}{\kern0pt}{\isacharcolon}{\kern0pt}countable{\isacharcomma}{\kern0pt}{\isacharprime}{\kern0pt}e\isactrlsub {\isadigit{1}}{\isacharcolon}{\kern0pt}{\isacharcolon}{\kern0pt}countable{\isacharcomma}{\kern0pt}{\isacharprime}{\kern0pt}l{\isacharcomma}{\kern0pt}{\isacharprime}{\kern0pt}m{\isacharparenright}{\kern0pt}\ pre{\isacharunderscore}{\kern0pt}graph{\isachardoublequoteclose}\ \isakeyword{and}\ \isanewline
\ \ \ \ H\ {\isacharcolon}{\kern0pt}{\isacharcolon}{\kern0pt}\ {\isachardoublequoteopen}{\isacharparenleft}{\kern0pt}{\isacharprime}{\kern0pt}v\isactrlsub {\isadigit{2}}{\isacharcolon}{\kern0pt}{\isacharcolon}{\kern0pt}countable{\isacharcomma}{\kern0pt}{\isacharprime}{\kern0pt}e\isactrlsub {\isadigit{2}}{\isacharcolon}{\kern0pt}{\isacharcolon}{\kern0pt}countable{\isacharcomma}{\kern0pt}{\isacharprime}{\kern0pt}l{\isacharcomma}{\kern0pt}{\isacharprime}{\kern0pt}m{\isacharparenright}{\kern0pt}\ pre{\isacharunderscore}{\kern0pt}graph{\isachardoublequoteclose}\ {\isacharplus}{\kern0pt}\isanewline
\ \ \isakeyword{fixes}\ \isanewline
\ \ \ \ f\ {\isacharcolon}{\kern0pt}{\isacharcolon}{\kern0pt}\ {\isachardoublequoteopen}{\isacharparenleft}{\kern0pt}{\isacharprime}{\kern0pt}v\isactrlsub {\isadigit{1}}{\isacharcomma}{\kern0pt}{\isacharprime}{\kern0pt}v\isactrlsub {\isadigit{2}}{\isacharcomma}{\kern0pt}{\isacharprime}{\kern0pt}e\isactrlsub {\isadigit{1}}{\isacharcomma}{\kern0pt}{\isacharprime}{\kern0pt}e\isactrlsub {\isadigit{2}}{\isacharparenright}{\kern0pt}\ pre{\isacharunderscore}{\kern0pt}morph{\isachardoublequoteclose}\ \ \isanewline
\ \ \isakeyword{assumes}\isanewline
\ \ \ \ morph{\isacharunderscore}{\kern0pt}edge{\isacharunderscore}{\kern0pt}range{\isacharcolon}{\kern0pt}\ {\isachardoublequoteopen}e\ {\isasymin}\ E\isactrlbsub G\isactrlesub \ {\isasymLongrightarrow}\ \isactrlbsub f\isactrlesub \isactrlsub E\ e\ {\isasymin}\ E\isactrlbsub H\isactrlesub {\isachardoublequoteclose}\ \isakeyword{and}\isanewline
\ \ \ \ morph{\isacharunderscore}{\kern0pt}node{\isacharunderscore}{\kern0pt}range{\isacharcolon}{\kern0pt}\ {\isachardoublequoteopen}v\ {\isasymin}\ V\isactrlbsub G\isactrlesub \ {\isasymLongrightarrow}\ \isactrlbsub f\isactrlesub \isactrlsub V\ v\ {\isasymin}\ V\isactrlbsub H\isactrlesub {\isachardoublequoteclose}\ \isakeyword{and}\isanewline
\ \ \ \ source{\isacharunderscore}{\kern0pt}preserve\ {\isacharcolon}{\kern0pt}\ {\isachardoublequoteopen}e\ {\isasymin}\ E\isactrlbsub G\isactrlesub \ {\isasymLongrightarrow}\ \isactrlbsub f\isactrlesub \isactrlsub V\ {\isacharparenleft}{\kern0pt}s\isactrlbsub G\isactrlesub \ e{\isacharparenright}{\kern0pt}\ {\isacharequal}{\kern0pt}\ s\isactrlbsub H\isactrlesub \ {\isacharparenleft}{\kern0pt}\isactrlbsub f\isactrlesub \isactrlsub E\ e{\isacharparenright}{\kern0pt}{\isachardoublequoteclose}\ \isakeyword{and}\isanewline
\ \ \ \ target{\isacharunderscore}{\kern0pt}preserve\ {\isacharcolon}{\kern0pt}\ {\isachardoublequoteopen}e\ {\isasymin}\ E\isactrlbsub G\isactrlesub \ {\isasymLongrightarrow}\ \isactrlbsub f\isactrlesub \isactrlsub V\ {\isacharparenleft}{\kern0pt}t\isactrlbsub G\isactrlesub \ e{\isacharparenright}{\kern0pt}\ {\isacharequal}{\kern0pt}\ t\isactrlbsub H\isactrlesub \ {\isacharparenleft}{\kern0pt}\isactrlbsub f\isactrlesub \isactrlsub E\ e{\isacharparenright}{\kern0pt}{\isachardoublequoteclose}\ \isakeyword{and}\isanewline
\ \ \ \ label{\isacharunderscore}{\kern0pt}preserve\ \ {\isacharcolon}{\kern0pt}\ {\isachardoublequoteopen}v\ {\isasymin}\ V\isactrlbsub G\isactrlesub \ {\isasymLongrightarrow}\ l\isactrlbsub G\isactrlesub \ v\ {\isacharequal}{\kern0pt}\ l\isactrlbsub H\isactrlesub \ {\isacharparenleft}{\kern0pt}\isactrlbsub f\isactrlesub \isactrlsub V\ v{\isacharparenright}{\kern0pt}{\isachardoublequoteclose}\ \isakeyword{and}\isanewline
\ \ \ \ mark{\isacharunderscore}{\kern0pt}preserve\ \ \ {\isacharcolon}{\kern0pt}\ {\isachardoublequoteopen}e\ {\isasymin}\ E\isactrlbsub G\isactrlesub \ {\isasymLongrightarrow}\ m\isactrlbsub G\isactrlesub \ e\ {\isacharequal}{\kern0pt}\ m\isactrlbsub H\isactrlesub \ {\isacharparenleft}{\kern0pt}\isactrlbsub f\isactrlesub \isactrlsub E\ e{\isacharparenright}{\kern0pt}{\isachardoublequoteclose}%
}
\DefineSnippet{morph-comp}{
\isacommand{definition}\isamarkupfalse%
\ morph{\isacharunderscore}{\kern0pt}comp\isanewline
\ \ {\isacharcolon}{\kern0pt}{\isacharcolon}{\kern0pt}\ {\isachardoublequoteopen}{\isacharparenleft}{\kern0pt}{\isacharprime}{\kern0pt}v\isactrlsub {\isadigit{2}}{\isacharcomma}{\kern0pt}{\isacharprime}{\kern0pt}v\isactrlsub {\isadigit{3}}{\isacharcomma}{\kern0pt}{\isacharprime}{\kern0pt}e\isactrlsub {\isadigit{2}}{\isacharcomma}{\kern0pt}{\isacharprime}{\kern0pt}e\isactrlsub {\isadigit{3}}{\isacharparenright}{\kern0pt}\ pre{\isacharunderscore}{\kern0pt}morph\ \isanewline
\ \ {\isasymRightarrow}\ \ {\isacharparenleft}{\kern0pt}{\isacharprime}{\kern0pt}v\isactrlsub {\isadigit{1}}{\isacharcomma}{\kern0pt}{\isacharprime}{\kern0pt}v\isactrlsub {\isadigit{2}}{\isacharcomma}{\kern0pt}{\isacharprime}{\kern0pt}e\isactrlsub {\isadigit{1}}{\isacharcomma}{\kern0pt}{\isacharprime}{\kern0pt}e\isactrlsub {\isadigit{2}}{\isacharparenright}{\kern0pt}\ pre{\isacharunderscore}{\kern0pt}morph\ \isanewline
\ \ {\isasymRightarrow}{\isacharparenleft}{\kern0pt}{\isacharprime}{\kern0pt}v\isactrlsub {\isadigit{1}}{\isacharcomma}{\kern0pt}{\isacharprime}{\kern0pt}v\isactrlsub {\isadigit{3}}{\isacharcomma}{\kern0pt}{\isacharprime}{\kern0pt}e\isactrlsub {\isadigit{1}}{\isacharcomma}{\kern0pt}{\isacharprime}{\kern0pt}e\isactrlsub {\isadigit{3}}{\isacharparenright}{\kern0pt}\ pre{\isacharunderscore}{\kern0pt}morph{\isachardoublequoteclose}\ {\isacharparenleft}{\kern0pt}\isakeyword{infixl}\ {\isachardoublequoteopen}{\isasymcirc}\isactrlsub {\isasymrightarrow}{\isachardoublequoteclose}\ {\isadigit{5}}{\isadigit{5}}{\isacharparenright}{\kern0pt}\ \isakeyword{where}\isanewline
{\isachardoublequoteopen}g\ {\isasymcirc}\isactrlsub {\isasymrightarrow}\ f\ {\isacharequal}{\kern0pt}\ {\isasymlparr}node{\isacharunderscore}{\kern0pt}map\ {\isacharequal}{\kern0pt}\ \isactrlbsub g\isactrlesub \isactrlsub V\ {\isasymcirc}\ \isactrlbsub f\isactrlesub \isactrlsub V{\isacharcomma}{\kern0pt}\ edge{\isacharunderscore}{\kern0pt}map\ {\isacharequal}{\kern0pt}\ \isactrlbsub g\isactrlesub \isactrlsub E\ {\isasymcirc}\ \isactrlbsub f\isactrlesub \isactrlsub E{\isasymrparr}{\isachardoublequoteclose}%
}
\DefineSnippet{wf-morph-comp}{
\isacommand{lemma}\isamarkupfalse%
\ wf{\isacharunderscore}{\kern0pt}morph{\isacharunderscore}{\kern0pt}comp{\isacharcolon}{\kern0pt}\isanewline
\ \ \isakeyword{assumes}\isanewline
\ \ \ \ f{\isacharcolon}{\kern0pt}\ {\isacartoucheopen}morphism\ G\ H\ f{\isacartoucheclose}\ \isakeyword{and}\isanewline
\ \ \ \ g{\isacharcolon}{\kern0pt}\ {\isacartoucheopen}morphism\ H\ K\ g{\isacartoucheclose}\isanewline
\ \ \isakeyword{shows}\ {\isacartoucheopen}morphism\ G\ K\ {\isacharparenleft}{\kern0pt}g\ {\isasymcirc}\isactrlsub {\isasymrightarrow}\ f{\isacharparenright}{\kern0pt}{\isacartoucheclose}%
}
\DefineSnippet{wf-morph-comp-proof-1}{
\isadelimproof
\endisadelimproof
\isatagproof
\isacommand{proof}\isamarkupfalse%
\ intro{\isacharunderscore}{\kern0pt}locales\isanewline
\ \ \isacommand{show}\isamarkupfalse%
\ {\isacartoucheopen}graph\ G{\isacartoucheclose}\ \isacommand{by}\isamarkupfalse%
\ {\isacharparenleft}{\kern0pt}fact\ morphism{\isachardot}{\kern0pt}axioms{\isacharbrackleft}{\kern0pt}OF\ f{\isacharbrackright}{\kern0pt}{\isacharparenright}{\kern0pt}\isanewline
\isacommand{next}\isamarkupfalse%
\isanewline
\ \ \isacommand{show}\isamarkupfalse%
\ {\isacartoucheopen}graph\ K{\isacartoucheclose}\ \isacommand{by}\isamarkupfalse%
\ {\isacharparenleft}{\kern0pt}fact\ morphism{\isachardot}{\kern0pt}axioms{\isacharbrackleft}{\kern0pt}OF\ g{\isacharbrackright}{\kern0pt}{\isacharparenright}{\kern0pt}%
}
\DefineSnippet{wf-morph-comp-proof-2}{
\ \ \isacommand{show}\isamarkupfalse%
\ {\isacartoucheopen}morphism{\isacharunderscore}{\kern0pt}axioms\ G\ K\ {\isacharparenleft}{\kern0pt}g\ {\isasymcirc}\isactrlsub {\isasymrightarrow}\ f{\isacharparenright}{\kern0pt}{\isacartoucheclose}\ \isanewline
\ \ \isacommand{proof}\isamarkupfalse%
\ standard\isanewline
\ \ \ \ \isacommand{show}\isamarkupfalse%
\ {\isacartoucheopen}\isactrlbsub g\ {\isasymcirc}\isactrlsub {\isasymrightarrow}\ f\isactrlesub \isactrlsub E\ e\ {\isasymin}\ E\isactrlbsub K\isactrlesub {\isacartoucheclose}\ \isakeyword{if}\ {\isacartoucheopen}e\ {\isasymin}\ E\isactrlbsub G\isactrlesub {\isacartoucheclose}\ \isakeyword{for}\ e\isanewline
\ \ \ \ \ \ \isacommand{by}\isamarkupfalse%
\ {\isacharparenleft}{\kern0pt}simp\ add{\isacharcolon}{\kern0pt}\ morph{\isacharunderscore}{\kern0pt}comp{\isacharunderscore}{\kern0pt}def\ \isanewline
\ \ \ \ \ \ \ \ \ \ \ \ \ \ \ \ \ \ \ \ morphism{\isachardot}{\kern0pt}morph{\isacharunderscore}{\kern0pt}edge{\isacharunderscore}{\kern0pt}range{\isacharbrackleft}{\kern0pt}OF\ g{\isacharbrackright}{\kern0pt}\ \isanewline
\ \ \ \ \ \ \ \ \ \ \ \ \ \ \ \ \ \ \ \ morphism{\isachardot}{\kern0pt}morph{\isacharunderscore}{\kern0pt}edge{\isacharunderscore}{\kern0pt}range{\isacharbrackleft}{\kern0pt}OF\ f{\isacharbrackright}{\kern0pt}\ \isanewline
\ \ \ \ \ \ \ \ \ \ \ \ \ \ \ \ \ \ \ \ that{\isacharparenright}{\kern0pt}%
}
\DefineSnippet{wf-morph-comp-proof-3}{
\ \ \ \ \isacommand{show}\isamarkupfalse%
\ {\isacartoucheopen}\isactrlbsub g\ {\isasymcirc}\isactrlsub {\isasymrightarrow}\ f\isactrlesub \isactrlsub V\ {\isacharparenleft}{\kern0pt}s\isactrlbsub G\isactrlesub \ e{\isacharparenright}{\kern0pt}\ {\isacharequal}{\kern0pt}\ s\isactrlbsub K\isactrlesub \ {\isacharparenleft}{\kern0pt}\isactrlbsub g\ {\isasymcirc}\isactrlsub {\isasymrightarrow}\ f\isactrlesub \isactrlsub E\ e{\isacharparenright}{\kern0pt}{\isacartoucheclose}\ \isakeyword{if}\ {\isacartoucheopen}e\ {\isasymin}\ E\isactrlbsub G\isactrlesub {\isacartoucheclose}\ \isakeyword{for}\ e\isanewline
\ \ \ \ \ \ \isacommand{by}\isamarkupfalse%
\ {\isacharparenleft}{\kern0pt}simp\ add{\isacharcolon}{\kern0pt}\ morph{\isacharunderscore}{\kern0pt}comp{\isacharunderscore}{\kern0pt}def\ \isanewline
\ \ \ \ \ \ \ \ \ \ \ \ morphism{\isachardot}{\kern0pt}morph{\isacharunderscore}{\kern0pt}edge{\isacharunderscore}{\kern0pt}range{\isacharbrackleft}{\kern0pt}OF\ f{\isacharbrackright}{\kern0pt}\ \isanewline
\ \ \ \ \ \ \ \ \ \ \ \ morphism{\isachardot}{\kern0pt}morph{\isacharunderscore}{\kern0pt}edge{\isacharunderscore}{\kern0pt}range{\isacharbrackleft}{\kern0pt}OF\ g{\isacharbrackright}{\kern0pt}\ \isanewline
\ \ \ \ \ \ \ \ \ \ \ \ morphism{\isachardot}{\kern0pt}source{\isacharunderscore}{\kern0pt}preserve{\isacharbrackleft}{\kern0pt}OF\ f{\isacharbrackright}{\kern0pt}\ \isanewline
\ \ \ \ \ \ \ \ \ \ \ \ morphism{\isachardot}{\kern0pt}source{\isacharunderscore}{\kern0pt}preserve{\isacharbrackleft}{\kern0pt}OF\ g{\isacharbrackright}{\kern0pt}\ that{\isacharparenright}{\kern0pt}%
}
\DefineSnippet{wf-morph-comp-proof-4}{
\ \ \ \ \isacommand{show}\isamarkupfalse%
\ {\isacartoucheopen}l\isactrlbsub G\isactrlesub \ v\ {\isacharequal}{\kern0pt}\ l\isactrlbsub K\isactrlesub \ {\isacharparenleft}{\kern0pt}\isactrlbsub g\ {\isasymcirc}\isactrlsub {\isasymrightarrow}\ f\isactrlesub \isactrlsub V\ v{\isacharparenright}{\kern0pt}{\isacartoucheclose}\ \isakeyword{if}\ {\isacartoucheopen}v\ {\isasymin}\ V\isactrlbsub G\isactrlesub {\isacartoucheclose}\ \isakeyword{for}\ v\isanewline
\ \ \ \ \ \ \isacommand{by}\isamarkupfalse%
\ {\isacharparenleft}{\kern0pt}simp\ add{\isacharcolon}{\kern0pt}\ morph{\isacharunderscore}{\kern0pt}comp{\isacharunderscore}{\kern0pt}def\ \isanewline
\ \ \ \ \ \ \ \ \ \ morphism{\isachardot}{\kern0pt}label{\isacharunderscore}{\kern0pt}preserve{\isacharbrackleft}{\kern0pt}OF\ f{\isacharbrackright}{\kern0pt}\isanewline
\ \ \ \ \ \ \ \ \ \ morphism{\isachardot}{\kern0pt}label{\isacharunderscore}{\kern0pt}preserve{\isacharbrackleft}{\kern0pt}OF\ g{\isacharbrackright}{\kern0pt}\isanewline
\ \ \ \ \ \ \ \ \ \ morphism{\isachardot}{\kern0pt}morph{\isacharunderscore}{\kern0pt}node{\isacharunderscore}{\kern0pt}range{\isacharbrackleft}{\kern0pt}OF\ f{\isacharbrackright}{\kern0pt}\ that{\isacharparenright}{\kern0pt}%
}
\DefineSnippet{inj-morph}{
\isacommand{locale}\isamarkupfalse%
\ injective{\isacharunderscore}{\kern0pt}morphism\ {\isacharequal}{\kern0pt}\ morphism\ {\isacharplus}{\kern0pt}\isanewline
\ \ \isakeyword{assumes}\ \isanewline
\ \ \ \ inj{\isacharunderscore}{\kern0pt}nodes{\isacharcolon}{\kern0pt}\ {\isachardoublequoteopen}inj{\isacharunderscore}{\kern0pt}on\ \isactrlbsub f\isactrlesub \isactrlsub V\ V\isactrlbsub G\isactrlesub {\isachardoublequoteclose}\ \isakeyword{and}\isanewline
\ \ \ \ inj{\isacharunderscore}{\kern0pt}edges{\isacharcolon}{\kern0pt}\ {\isachardoublequoteopen}inj{\isacharunderscore}{\kern0pt}on\ \isactrlbsub f\isactrlesub \isactrlsub E\ E\isactrlbsub G\isactrlesub {\isachardoublequoteclose}%
}
\DefineSnippet{inj-morph-comp-inj1}{
\isacommand{lemma}\isamarkupfalse%
\ inj{\isacharunderscore}{\kern0pt}comp{\isacharunderscore}{\kern0pt}fg{\isacharunderscore}{\kern0pt}g{\isacharunderscore}{\kern0pt}inj{\isacharcolon}{\kern0pt}\isanewline
\ \ \isakeyword{assumes}\ {\isacartoucheopen}morphism\ G\ H\ g{\isacartoucheclose}\ {\isacartoucheopen}morphism\ H\ K\ f{\isacartoucheclose}\isanewline
\ \ \isakeyword{assumes}\ {\isacartoucheopen}injective{\isacharunderscore}{\kern0pt}morphism\ G\ H\ {\isacharparenleft}{\kern0pt}f\ {\isasymcirc}\isactrlsub {\isasymrightarrow}\ g{\isacharparenright}{\kern0pt}{\isacartoucheclose}\isanewline
\ \ \isakeyword{shows}\ {\isacartoucheopen}injective{\isacharunderscore}{\kern0pt}morphism\ G\ H\ g{\isacartoucheclose}%
}
\DefineSnippet{surj-morph}{
\isacommand{locale}\isamarkupfalse%
\ surjective{\isacharunderscore}{\kern0pt}morphism\ {\isacharequal}{\kern0pt}\ morphism\ {\isacharplus}{\kern0pt}\isanewline
\ \ \isakeyword{assumes}\isanewline
\ \ \ \ surj{\isacharunderscore}{\kern0pt}nodes{\isacharcolon}{\kern0pt}\ {\isacartoucheopen}v\ {\isasymin}\ V\isactrlbsub H\isactrlesub \ {\isasymLongrightarrow}\ {\isasymexists}v{\isacharprime}{\kern0pt}\ {\isasymin}\ V\isactrlbsub G\isactrlesub {\isachardot}{\kern0pt}\ \isactrlbsub f\isactrlesub \isactrlsub V\ v{\isacharprime}{\kern0pt}\ {\isacharequal}{\kern0pt}\ v{\isacartoucheclose}\ \isakeyword{and}\isanewline
\ \ \ \ surj{\isacharunderscore}{\kern0pt}edges{\isacharcolon}{\kern0pt}\ {\isacartoucheopen}e\ {\isasymin}\ E\isactrlbsub H\isactrlesub \ {\isasymLongrightarrow}\ {\isasymexists}e{\isacharprime}{\kern0pt}\ {\isasymin}\ E\isactrlbsub G\isactrlesub {\isachardot}{\kern0pt}\ \isactrlbsub f\isactrlesub \isactrlsub E\ e{\isacharprime}{\kern0pt}\ {\isacharequal}{\kern0pt}\ e{\isacartoucheclose}%
}
\DefineSnippet{bij-morph}{
\isacommand{locale}\isamarkupfalse%
\ bijective{\isacharunderscore}{\kern0pt}morphism\ {\isacharequal}{\kern0pt}\ morphism\ {\isacharplus}{\kern0pt}\isanewline
\ \ \isakeyword{assumes}\isanewline
\ \ \ \ bij{\isacharunderscore}{\kern0pt}nodes{\isacharcolon}{\kern0pt}\ {\isachardoublequoteopen}bij{\isacharunderscore}{\kern0pt}betw\ \isactrlbsub f\isactrlesub \isactrlsub V\ V\isactrlbsub G\isactrlesub \ V\isactrlbsub H\isactrlesub {\isachardoublequoteclose}\ \isakeyword{and}\isanewline
\ \ \ \ bij{\isacharunderscore}{\kern0pt}edges{\isacharcolon}{\kern0pt}\ {\isachardoublequoteopen}bij{\isacharunderscore}{\kern0pt}betw\ \isactrlbsub f\isactrlesub \isactrlsub E\ E\isactrlbsub G\isactrlesub \ E\isactrlbsub H\isactrlesub {\isachardoublequoteclose}%
}
\DefineSnippet{id-morph-locale}{
\isacommand{locale}\isamarkupfalse%
\ identity{\isacharunderscore}{\kern0pt}morphism\ {\isacharequal}{\kern0pt}\ \isanewline
\ \ bijective{\isacharunderscore}{\kern0pt}morphism\ G\ G\ f\ \isakeyword{for}\ G\ f\ {\isacharplus}{\kern0pt}\isanewline
\isakeyword{assumes}\isanewline
\ \ id{\isacharunderscore}{\kern0pt}nodes{\isacharcolon}{\kern0pt}\ {\isacartoucheopen}v\ {\isasymin}\ V\isactrlbsub G\isactrlesub \ {\isasymLongrightarrow}\ \isactrlbsub f\isactrlesub \isactrlsub V\ v\ {\isacharequal}{\kern0pt}\ v{\isacartoucheclose}\ \isakeyword{and}\isanewline
\ \ id{\isacharunderscore}{\kern0pt}edges{\isacharcolon}{\kern0pt}\ {\isacartoucheopen}e\ {\isasymin}\ E\isactrlbsub G\isactrlesub \ {\isasymLongrightarrow}\ \isactrlbsub f\isactrlesub \isactrlsub E\ e\ {\isacharequal}{\kern0pt}\ e{\isacartoucheclose}%
}
\DefineSnippet{id-morph}{
\isacommand{abbreviation}\isamarkupfalse%
\ idM\ \isakeyword{where}\ {\isacartoucheopen}idM\ {\isasymequiv}\ {\isasymlparr}node{\isacharunderscore}{\kern0pt}map\ {\isacharequal}{\kern0pt}\ id{\isacharcomma}{\kern0pt}\ edge{\isacharunderscore}{\kern0pt}map\ {\isacharequal}{\kern0pt}\ id{\isasymrparr}{\isacartoucheclose}%
}
\DefineSnippet{id-morph-self}{
\isacommand{context}\isamarkupfalse%
\ graph\ \isakeyword{begin}\isanewline
\isacommand{sublocale}\isamarkupfalse%
\ idm{\isacharcolon}{\kern0pt}\ identity{\isacharunderscore}{\kern0pt}morphism\ G\ idM\isanewline
\isadelimproof
\ \ %
\endisadelimproof
\isatagproof
\isacommand{by}\isamarkupfalse%
\ standard\ simp{\isacharunderscore}{\kern0pt}all%
\endisatagproof
{\isafoldproof}%
\isadelimproof
\isanewline
\endisadelimproof
\isacommand{end}\isamarkupfalse%
}
\DefineSnippet{morph-comp-id-simps}{
\isacommand{lemma}\isamarkupfalse%
\ morph{\isacharunderscore}{\kern0pt}comp{\isacharunderscore}{\kern0pt}id{\isacharbrackleft}{\kern0pt}simp{\isacharbrackright}{\kern0pt}{\isacharcolon}{\kern0pt}\isanewline
\ \ \isakeyword{shows}\ \ {\isacartoucheopen}\isactrlbsub f\ {\isasymcirc}\isactrlsub {\isasymrightarrow}\ idM\isactrlesub \isactrlsub V\ {\isacharequal}{\kern0pt}\ \isactrlbsub f\isactrlesub \isactrlsub V{\isacartoucheclose}\ \isanewline
\ \ \ \ \isakeyword{and}\ \ {\isacartoucheopen}\isactrlbsub f\ {\isasymcirc}\isactrlsub {\isasymrightarrow}\ idM\isactrlesub \isactrlsub E\ {\isacharequal}{\kern0pt}\ \isactrlbsub f\isactrlesub \isactrlsub E{\isacartoucheclose}\isanewline
\ \ \ \ \isakeyword{and}\ \ {\isacartoucheopen}\isactrlbsub idM\ {\isasymcirc}\isactrlsub {\isasymrightarrow}\ f\isactrlesub \isactrlsub V\ {\isacharequal}{\kern0pt}\ \isactrlbsub f\isactrlesub \isactrlsub V{\isacartoucheclose}\isanewline
\ \ \ \ \isakeyword{and}\ \ {\isacartoucheopen}\isactrlbsub idM\ {\isasymcirc}\isactrlsub {\isasymrightarrow}\ f\isactrlesub \isactrlsub E\ {\isacharequal}{\kern0pt}\ \isactrlbsub f\isactrlesub \isactrlsub E{\isacartoucheclose}\isanewline
\isadelimproof
\ \ %
\endisadelimproof
\isatagproof
\isacommand{by}\isamarkupfalse%
\ {\isacharparenleft}{\kern0pt}simp{\isacharunderscore}{\kern0pt}all\ add{\isacharcolon}{\kern0pt}\ morph{\isacharunderscore}{\kern0pt}comp{\isacharunderscore}{\kern0pt}def{\isacharparenright}{\kern0pt}%
\endisatagproof
{\isafoldproof}%
\isadelimproof
\endisadelimproof
}
\DefineSnippet{comp-id-bij}{
\isacommand{lemma}\isamarkupfalse%
\isanewline
\ \ \isakeyword{assumes}\ {\isacartoucheopen}morphism\ G\ H\ f{\isacartoucheclose}\isanewline
\ \ \isakeyword{assumes}\ {\isacartoucheopen}morphism\ H\ G\ g{\isacartoucheclose}\isanewline
\ \ \isakeyword{assumes}\ {\isacartoucheopen}{\isasymforall}v\ {\isasymin}\ V\isactrlbsub G\isactrlesub {\isachardot}{\kern0pt}\ \isactrlbsub g\ {\isasymcirc}\isactrlsub {\isasymrightarrow}\ f\isactrlesub \isactrlsub V\ v\ {\isacharequal}{\kern0pt}\ v{\isacartoucheclose}\isanewline
\ \ \isakeyword{assumes}\ {\isacartoucheopen}{\isasymforall}e\ {\isasymin}\ E\isactrlbsub G\isactrlesub {\isachardot}{\kern0pt}\ \isactrlbsub g\ {\isasymcirc}\isactrlsub {\isasymrightarrow}\ f\isactrlesub \isactrlsub E\ e\ {\isacharequal}{\kern0pt}\ e{\isacartoucheclose}\isanewline
\ \ \isakeyword{assumes}\ {\isacartoucheopen}{\isasymforall}v\ {\isasymin}\ V\isactrlbsub H\isactrlesub {\isachardot}{\kern0pt}\ \isactrlbsub f\ {\isasymcirc}\isactrlsub {\isasymrightarrow}\ g\isactrlesub \isactrlsub V\ v\ {\isacharequal}{\kern0pt}\ v{\isacartoucheclose}\isanewline
\ \ \isakeyword{assumes}\ {\isacartoucheopen}{\isasymforall}e\ {\isasymin}\ E\isactrlbsub H\isactrlesub {\isachardot}{\kern0pt}\ \isactrlbsub f\ {\isasymcirc}\isactrlsub {\isasymrightarrow}\ g\isactrlesub \isactrlsub E\ e\ {\isacharequal}{\kern0pt}\ e{\isacartoucheclose}\isanewline
\ \ \isakeyword{shows}\ {\isacartoucheopen}bijective{\isacharunderscore}{\kern0pt}morphism\ G\ H\ f{\isacartoucheclose}%
}
\DefineSnippet{direct-derivation}{
\isacommand{locale}\isamarkupfalse%
\ direct{\isacharunderscore}{\kern0pt}derivation\ {\isacharequal}{\kern0pt}\ \isanewline
\ \ r{\isacharcolon}{\kern0pt}\ rule\ r\ b\ b{\isacharprime}{\kern0pt}\ {\isacharplus}{\kern0pt}\isanewline
\ \ gi{\isacharcolon}{\kern0pt}\ injective{\isacharunderscore}{\kern0pt}morphism\ {\isachardoublequoteopen}lhs\ r{\isachardoublequoteclose}\ G\ g\ \ {\isacharplus}{\kern0pt}\ \isanewline
\ \ po{\isadigit{1}}{\isacharcolon}{\kern0pt}\ pushout{\isacharunderscore}{\kern0pt}diagram\ {\isachardoublequoteopen}interf\ r{\isachardoublequoteclose}\ {\isachardoublequoteopen}lhs\ r{\isachardoublequoteclose}\ D\ G\ b\ d\ g\ c\ {\isacharplus}{\kern0pt}\isanewline
\ \ po{\isadigit{2}}{\isacharcolon}{\kern0pt}\ pushout{\isacharunderscore}{\kern0pt}diagram\ {\isachardoublequoteopen}interf\ r{\isachardoublequoteclose}\ {\isachardoublequoteopen}rhs\ r{\isachardoublequoteclose}\ D\ H\ b{\isacharprime}{\kern0pt}\ d\ f\ c{\isacharprime}{\kern0pt}\isanewline
\ \ \isakeyword{for}\ r\ b\ b{\isacharprime}{\kern0pt}\ G\ g\ D\ d\ c\ H\ f\ c{\isacharprime}{\kern0pt}%
}
\DefineSnippet{direct-derivation-unique}{
\isacommand{theorem}\isamarkupfalse%
\ uniqueness{\isacharunderscore}{\kern0pt}direct{\isacharunderscore}{\kern0pt}derivation{\isacharcolon}{\kern0pt}\isanewline
\ \ \isakeyword{assumes}\ \isanewline
\ \ \ \ dd{\isadigit{2}}{\isacharcolon}{\kern0pt}\ {\isacartoucheopen}direct{\isacharunderscore}{\kern0pt}derivation\ r\ b\ b{\isacharprime}{\kern0pt}\ G\ g\ D{\isacharprime}{\kern0pt}\ d{\isacharprime}{\kern0pt}\ m\ H{\isacharprime}{\kern0pt}\ f{\isacharprime}{\kern0pt}\ m{\isacharprime}{\kern0pt}{\isacartoucheclose}\isanewline
\ \ \isakeyword{shows}\ {\isacartoucheopen}{\isacharparenleft}{\kern0pt}{\isasymexists}u{\isachardot}{\kern0pt}\ bijective{\isacharunderscore}{\kern0pt}morphism\ D\ D{\isacharprime}{\kern0pt}\ u{\isacharparenright}{\kern0pt}\ \isanewline
\ \ \ \ \ \ \ \ \ \ {\isasymand}\ {\isacharparenleft}{\kern0pt}{\isasymexists}u{\isachardot}{\kern0pt}\ bijective{\isacharunderscore}{\kern0pt}morphism\ H\ H{\isacharprime}{\kern0pt}\ u{\isacharparenright}{\kern0pt}{\isacartoucheclose}%
}
\DefineSnippet{udd-fr}{
\isacommand{interpret}\isamarkupfalse%
\ fr{\isacharcolon}{\kern0pt}\ pullback{\isacharunderscore}{\kern0pt}construction\ D\ G\ D{\isacharprime}{\kern0pt}\ c\ m\ \isacommand{{\isachardot}{\kern0pt}{\isachardot}{\kern0pt}}\isamarkupfalse%
}
\DefineSnippet{udd-bl}{
\isacommand{interpret}\isamarkupfalse%
\ bl{\isacharcolon}{\kern0pt}\ pullback{\isacharunderscore}{\kern0pt}diagram\ {\isachardoublequoteopen}interf\ r{\isachardoublequoteclose}\ {\isachardoublequoteopen}interf\ r{\isachardoublequoteclose}\ \isanewline
\ \ \ \ \ \ \ \ \ \ \ \ \ \ \ \ {\isachardoublequoteopen}interf\ r{\isachardoublequoteclose}\ {\isachardoublequoteopen}lhs\ r{\isachardoublequoteclose}\ idM\ idM\ b\ b\isanewline
\ \ \isacommand{using}\isamarkupfalse%
\ fun{\isacharunderscore}{\kern0pt}algrtr{\isacharunderscore}{\kern0pt}{\isadigit{4}}{\isacharunderscore}{\kern0pt}{\isadigit{7}}{\isacharunderscore}{\kern0pt}{\isadigit{2}}{\isacharbrackleft}{\kern0pt}OF\ r{\isachardot}{\kern0pt}k{\isachardot}{\kern0pt}injective{\isacharunderscore}{\kern0pt}morphism{\isacharunderscore}{\kern0pt}axioms{\isacharbrackright}{\kern0pt}\isanewline
\ \ \isacommand{by}\isamarkupfalse%
\ assumption%
}
\DefineSnippet{udd-b}{
\isacommand{interpret}\isamarkupfalse%
\ backside{\isacharcolon}{\kern0pt}\ pullback{\isacharunderscore}{\kern0pt}diagram\ {\isachardoublequoteopen}interf\ r{\isachardoublequoteclose}\ D{\isacharprime}{\kern0pt}\ {\isachardoublequoteopen}interf\ r{\isachardoublequoteclose}\ G\ \isanewline
\ \ \ \ \ \ \ \ \ \ \ \ \ \ \ \ \ \ \ \ \ \ {\isacartoucheopen}d{\isacharprime}{\kern0pt}\ {\isasymcirc}\isactrlsub {\isasymrightarrow}\ idM{\isacartoucheclose}\ idM\ m\ {\isacartoucheopen}g\ {\isasymcirc}\isactrlsub {\isasymrightarrow}\ b{\isacartoucheclose}\isanewline
\ \ \isacommand{using}\isamarkupfalse%
\ pullback{\isacharunderscore}{\kern0pt}composition{\isacharbrackleft}{\kern0pt}OF\ bl{\isachardot}{\kern0pt}pullback{\isacharunderscore}{\kern0pt}diagram{\isacharunderscore}{\kern0pt}axioms\ \isanewline
\ \ \ \ \ \ \ \ \ \ \ \ \ \ \ \ \ \ \ \ \ \ \ \ \ \ \ \ \ \ \ \ dd{\isadigit{2}}{\isachardot}{\kern0pt}pb{\isadigit{1}}{\isachardot}{\kern0pt}flip{\isacharunderscore}{\kern0pt}diagram{\isacharbrackright}{\kern0pt}\ \isanewline
\ \ \isacommand{by}\isamarkupfalse%
\ assumption%
}
\DefineSnippet{udd-defh}{
\isacommand{define}\isamarkupfalse%
\ h\ \isakeyword{where}\isanewline
\ \ \ \ {\isacartoucheopen}h\ {\isasymequiv}\ {\isasymlparr}node{\isacharunderscore}{\kern0pt}map\ {\isacharequal}{\kern0pt}\ {\isasymlambda}v{\isachardot}{\kern0pt}\ {\isacharparenleft}{\kern0pt}\isactrlbsub d\isactrlesub \isactrlsub V\ v{\isacharcomma}{\kern0pt}\ \isactrlbsub d{\isacharprime}{\kern0pt}\isactrlesub \isactrlsub V\ v{\isacharparenright}{\kern0pt}\isanewline
\ \ \ \ \ \ \ \ \ {\isacharcomma}{\kern0pt}edge{\isacharunderscore}{\kern0pt}map\ {\isacharequal}{\kern0pt}\ {\isasymlambda}e{\isachardot}{\kern0pt}\ {\isacharparenleft}{\kern0pt}\isactrlbsub d\isactrlesub \isactrlsub E\ e{\isacharcomma}{\kern0pt}\ \isactrlbsub d{\isacharprime}{\kern0pt}\isactrlesub \isactrlsub E\ e{\isacharparenright}{\kern0pt}{\isasymrparr}{\isacartoucheclose}%
}
\DefineSnippet{udd-injh}{
\isacommand{interpret}\isamarkupfalse%
\ h{\isacharcolon}{\kern0pt}\ injective{\isacharunderscore}{\kern0pt}morphism\ {\isachardoublequoteopen}interf\ r{\isachardoublequoteclose}\ fr{\isachardot}{\kern0pt}A\ h\isanewline
\isacommand{proof}\isamarkupfalse%
\isanewline
\ \ \isacommand{show}\isamarkupfalse%
\ {\isacartoucheopen}inj{\isacharunderscore}{\kern0pt}on\ \isactrlbsub h\isactrlesub \isactrlsub V\ V\isactrlbsub interf\ r\isactrlesub {\isacartoucheclose}\isanewline
\ \ \ \ \isacommand{using}\isamarkupfalse%
\ d{\isacharunderscore}{\kern0pt}inj{\isachardot}{\kern0pt}inj{\isacharunderscore}{\kern0pt}nodes\isanewline
\ \ \ \ \isacommand{by}\isamarkupfalse%
\ {\isacharparenleft}{\kern0pt}simp\ add{\isacharcolon}{\kern0pt}\ h{\isacharunderscore}{\kern0pt}def\ inj{\isacharunderscore}{\kern0pt}on{\isacharunderscore}{\kern0pt}def{\isacharparenright}{\kern0pt}\isanewline
\isacommand{next}\isamarkupfalse%
\isanewline
\ \ \isacommand{show}\isamarkupfalse%
\ {\isacartoucheopen}inj{\isacharunderscore}{\kern0pt}on\ \isactrlbsub h\isactrlesub \isactrlsub E\ E\isactrlbsub interf\ r\isactrlesub {\isacartoucheclose}\isanewline
\ \ \ \ \isacommand{using}\isamarkupfalse%
\ d{\isacharunderscore}{\kern0pt}inj{\isachardot}{\kern0pt}inj{\isacharunderscore}{\kern0pt}edges\isanewline
\ \ \ \ \isacommand{by}\isamarkupfalse%
\ {\isacharparenleft}{\kern0pt}simp\ add{\isacharcolon}{\kern0pt}\ h{\isacharunderscore}{\kern0pt}def\ inj{\isacharunderscore}{\kern0pt}on{\isacharunderscore}{\kern0pt}def{\isacharparenright}{\kern0pt}\isanewline
\isacommand{qed}\isamarkupfalse%
}
\DefineSnippet{udd-bijk}{
\isacommand{interpret}\isamarkupfalse%
\ k{\isacharunderscore}{\kern0pt}bij{\isacharcolon}{\kern0pt}\ bijective{\isacharunderscore}{\kern0pt}morphism\ fr{\isachardot}{\kern0pt}A\ D{\isacharprime}{\kern0pt}\ fr{\isachardot}{\kern0pt}c\isanewline
\ \ \isacommand{using}\isamarkupfalse%
\ top{\isachardot}{\kern0pt}b{\isacharunderscore}{\kern0pt}bij{\isacharunderscore}{\kern0pt}imp{\isacharunderscore}{\kern0pt}g{\isacharunderscore}{\kern0pt}bij{\isacharbrackleft}{\kern0pt}OF\ r{\isachardot}{\kern0pt}k{\isachardot}{\kern0pt}G{\isachardot}{\kern0pt}idm{\isachardot}{\kern0pt}bijective{\isacharunderscore}{\kern0pt}morphism{\isacharunderscore}{\kern0pt}axioms{\isacharbrackright}{\kern0pt}\isanewline
\ \ \isacommand{by}\isamarkupfalse%
\ assumption%
}
\DefineSnippet{udd-linv}{
\isacommand{obtain}\isamarkupfalse%
\ linv\ \isakeyword{where}\ linv{\isacharcolon}{\kern0pt}{\isacartoucheopen}bijective{\isacharunderscore}{\kern0pt}morphism\ D\ fr{\isachardot}{\kern0pt}A\ linv{\isacartoucheclose}\isanewline
\ \ \isakeyword{and}\ {\isacartoucheopen}{\isasymAnd}v{\isachardot}{\kern0pt}\ v\ {\isasymin}\ V\isactrlbsub D\isactrlesub {\isasymLongrightarrow}\ \isactrlbsub fr{\isachardot}{\kern0pt}b\ {\isasymcirc}\isactrlsub {\isasymrightarrow}linv\isactrlesub \isactrlsub V\ v\ {\isacharequal}{\kern0pt}\ v{\isacartoucheclose}\ \isanewline
\ \ \ \ \ \ {\isacartoucheopen}{\isasymAnd}e{\isachardot}{\kern0pt}\ e\ {\isasymin}\ E\isactrlbsub D\isactrlesub {\isasymLongrightarrow}\ \isactrlbsub fr{\isachardot}{\kern0pt}b\ {\isasymcirc}\isactrlsub {\isasymrightarrow}linv\isactrlesub \isactrlsub E\ e\ {\isacharequal}{\kern0pt}\ e{\isacartoucheclose}\isanewline
\ \ \isakeyword{and}\ {\isacartoucheopen}{\isasymAnd}v{\isachardot}{\kern0pt}\ v\ {\isasymin}\ V\isactrlbsub fr{\isachardot}{\kern0pt}A\isactrlesub {\isasymLongrightarrow}\ \isactrlbsub linv\ {\isasymcirc}\isactrlsub {\isasymrightarrow}\ fr{\isachardot}{\kern0pt}b\ \isactrlesub \isactrlsub V\ v\ {\isacharequal}{\kern0pt}\ v{\isacartoucheclose}\ \isanewline
\ \ \ \ \ \ {\isacartoucheopen}{\isasymAnd}e{\isachardot}{\kern0pt}\ e\ {\isasymin}\ E\isactrlbsub fr{\isachardot}{\kern0pt}A\isactrlesub {\isasymLongrightarrow}\ \isactrlbsub linv\ {\isasymcirc}\isactrlsub {\isasymrightarrow}\ fr{\isachardot}{\kern0pt}b\ \isactrlesub \isactrlsub E\ e{\isacharequal}{\kern0pt}\ e{\isacartoucheclose}\isanewline
\ \ \isacommand{by}\isamarkupfalse%
\ {\isacharparenleft}{\kern0pt}metis\ l{\isacharunderscore}{\kern0pt}bij{\isachardot}{\kern0pt}ex{\isacharunderscore}{\kern0pt}inv{\isacharparenright}{\kern0pt}%
}
\DefineSnippet{udd-defu}{
\isacommand{define}\isamarkupfalse%
\ u\ \isakeyword{where}\ {\isacartoucheopen}u\ {\isasymequiv}\ fr{\isachardot}{\kern0pt}c\ {\isasymcirc}\isactrlsub {\isasymrightarrow}\ linv{\isacartoucheclose}%
}
\DefineSnippet{udd-biju}{
\isacommand{interpret}\isamarkupfalse%
\ u{\isacharcolon}{\kern0pt}\ bijective{\isacharunderscore}{\kern0pt}morphism\ D\ D{\isacharprime}{\kern0pt}\ u\isanewline
\ \ \isacommand{using}\isamarkupfalse%
\ bij{\isacharunderscore}{\kern0pt}comp{\isacharunderscore}{\kern0pt}bij{\isacharunderscore}{\kern0pt}is{\isacharunderscore}{\kern0pt}bij{\isacharbrackleft}{\kern0pt}OF\ linv\ k{\isacharunderscore}{\kern0pt}bij{\isachardot}{\kern0pt}bijective{\isacharunderscore}{\kern0pt}morphism{\isacharunderscore}{\kern0pt}axioms{\isacharbrackright}{\kern0pt}\isanewline
\ \ \isacommand{by}\isamarkupfalse%
\ {\isacharparenleft}{\kern0pt}simp\ add{\isacharcolon}{\kern0pt}\ u{\isacharunderscore}{\kern0pt}def{\isacharparenright}{\kern0pt}%
}
\DefineSnippet{pushout-complement-unique}{
\isacommand{corollary}\isamarkupfalse%
\ {\isacharparenleft}{\kern0pt}\isakeyword{in}\ pushout{\isacharunderscore}{\kern0pt}diagram{\isacharparenright}{\kern0pt}\ uniqueness{\isacharunderscore}{\kern0pt}pc{\isacharcolon}{\kern0pt}\isanewline
\ \ \isakeyword{assumes}\ po{\isadigit{2}}{\isacharcolon}{\kern0pt}\ {\isacartoucheopen}pushout{\isacharunderscore}{\kern0pt}diagram\ A\ B\ C{\isacharprime}{\kern0pt}\ D\ b\ c{\isacharprime}{\kern0pt}\ f\ g{\isacharprime}{\kern0pt}{\isacartoucheclose}\isanewline
\ \ \isakeyword{and}\ \ \ \ \ \ \ b{\isacharprime}{\kern0pt}{\isacharcolon}{\kern0pt}\ {\isacartoucheopen}injective{\isacharunderscore}{\kern0pt}morphism\ A\ B\ b{\isacartoucheclose}\isanewline
\ \ \isakeyword{and}\ \ \ \ \ \ \ f{\isacharprime}{\kern0pt}{\isacharcolon}{\kern0pt}\ {\isacartoucheopen}injective{\isacharunderscore}{\kern0pt}morphism\ B\ D\ f{\isacartoucheclose}\isanewline
\isakeyword{shows}\ {\isacartoucheopen}{\isasymexists}u{\isachardot}{\kern0pt}\ bijective{\isacharunderscore}{\kern0pt}morphism\ C\ C{\isacharprime}{\kern0pt}\ u{\isacartoucheclose}%
}
\DefineSnippet{direct-derivation-construction}{
\isacommand{locale}\isamarkupfalse%
\ direct{\isacharunderscore}{\kern0pt}derivation{\isacharunderscore}{\kern0pt}construction\ {\isacharequal}{\kern0pt}\isanewline
\ \ r{\isacharcolon}{\kern0pt}\ rule\ r\ b\ b{\isacharprime}{\kern0pt}\ {\isacharplus}{\kern0pt}\isanewline
\ \ d{\isacharcolon}{\kern0pt}\ deletion\ {\isachardoublequoteopen}interf\ r{\isachardoublequoteclose}\ G\ {\isachardoublequoteopen}lhs\ r{\isachardoublequoteclose}\ g\ b\ {\isacharplus}{\kern0pt}\isanewline
\ \ g{\isacharcolon}{\kern0pt}\ gluing\ {\isachardoublequoteopen}interf\ r{\isachardoublequoteclose}\ d{\isachardot}{\kern0pt}D\ {\isachardoublequoteopen}rhs\ r{\isachardoublequoteclose}\ d{\isachardot}{\kern0pt}d\ b{\isacharprime}{\kern0pt}\ \isakeyword{for}\ G\ r\ b\ b{\isacharprime}{\kern0pt}\ g\ H\ {\isacharplus}{\kern0pt}\isanewline
\isakeyword{assumes}\ a{\isacharcolon}{\kern0pt}\ {\isacartoucheopen}H\ {\isacharequal}{\kern0pt}\ g{\isachardot}{\kern0pt}H{\isacartoucheclose}%
}
\DefineSnippet{parallelindependence}{
\isacommand{locale}\isamarkupfalse%
\ parallel{\isacharunderscore}{\kern0pt}independence\ {\isacharequal}{\kern0pt}\isanewline
\ \ p\isactrlsub {\isadigit{1}}{\isacharcolon}{\kern0pt}\ direct{\isacharunderscore}{\kern0pt}derivation\ r\isactrlsub {\isadigit{1}}\ b\isactrlsub {\isadigit{1}}\ b\isactrlsub {\isadigit{1}}{\isacharprime}{\kern0pt}\ G\ g\isactrlsub {\isadigit{1}}\ D\isactrlsub {\isadigit{1}}\ m\isactrlsub {\isadigit{1}}\ c\isactrlsub {\isadigit{1}}\ H\isactrlsub {\isadigit{1}}\ f\isactrlsub {\isadigit{1}}\ h\isactrlsub {\isadigit{1}}\ {\isacharplus}{\kern0pt}\isanewline
\ \ p\isactrlsub {\isadigit{2}}{\isacharcolon}{\kern0pt}\ direct{\isacharunderscore}{\kern0pt}derivation\ r\isactrlsub {\isadigit{2}}\ b\isactrlsub {\isadigit{2}}\ b\isactrlsub {\isadigit{2}}{\isacharprime}{\kern0pt}\ G\ g\isactrlsub {\isadigit{2}}\ D\isactrlsub {\isadigit{2}}\ m\isactrlsub {\isadigit{2}}\ c\isactrlsub {\isadigit{2}}\ H\isactrlsub {\isadigit{2}}\ f\isactrlsub {\isadigit{2}}\ h\isactrlsub {\isadigit{2}}\isanewline
\ \ \isakeyword{for}\ r\isactrlsub {\isadigit{1}}\ b\isactrlsub {\isadigit{1}}\ b\isactrlsub {\isadigit{1}}{\isacharprime}{\kern0pt}\ G\ g\isactrlsub {\isadigit{1}}\ D\isactrlsub {\isadigit{1}}\ m\isactrlsub {\isadigit{1}}\ c\isactrlsub {\isadigit{1}}\ H\isactrlsub {\isadigit{1}}\ f\isactrlsub {\isadigit{1}}\ h\isactrlsub {\isadigit{1}}\ \isanewline
\ \ \ \ \ \ r\isactrlsub {\isadigit{2}}\ b\isactrlsub {\isadigit{2}}\ b\isactrlsub {\isadigit{2}}{\isacharprime}{\kern0pt}\ g\isactrlsub {\isadigit{2}}\ D\isactrlsub {\isadigit{2}}\ m\isactrlsub {\isadigit{2}}\ c\isactrlsub {\isadigit{2}}\ H\isactrlsub {\isadigit{2}}\ f\isactrlsub {\isadigit{2}}\ h\isactrlsub {\isadigit{2}}\ {\isacharplus}{\kern0pt}\isanewline
\ \ \isakeyword{assumes}\isanewline
\ \ \ \ i{\isacharcolon}{\kern0pt}\ {\isacartoucheopen}{\isasymexists}i{\isachardot}{\kern0pt}\ morphism\ {\isacharparenleft}{\kern0pt}lhs\ r\isactrlsub {\isadigit{1}}{\isacharparenright}{\kern0pt}\ D\isactrlsub {\isadigit{2}}\ i\ \isanewline
\ \ \ \ \ \ \ \ \ \ \ \ {\isasymand}\ {\isacharparenleft}{\kern0pt}{\isasymforall}v\ {\isasymin}\ V\isactrlbsub lhs\ r\isactrlsub {\isadigit{1}}\isactrlesub {\isachardot}{\kern0pt}\ \isactrlbsub c\isactrlsub {\isadigit{2}}\ {\isasymcirc}\isactrlsub {\isasymrightarrow}\ i\isactrlesub \isactrlsub V\ v\ {\isacharequal}{\kern0pt}\ \isactrlbsub g\isactrlsub {\isadigit{1}}\isactrlesub \isactrlsub V\ v{\isacharparenright}{\kern0pt}\ \isanewline
\ \ \ \ \ \ \ \ \ \ \ \ {\isasymand}\ {\isacharparenleft}{\kern0pt}{\isasymforall}e\ {\isasymin}\ E\isactrlbsub lhs\ r\isactrlsub {\isadigit{1}}\isactrlesub {\isachardot}{\kern0pt}\ \isactrlbsub c\isactrlsub {\isadigit{2}}\ {\isasymcirc}\isactrlsub {\isasymrightarrow}\ i\isactrlesub \isactrlsub E\ e\ {\isacharequal}{\kern0pt}\ \isactrlbsub g\isactrlsub {\isadigit{1}}\isactrlesub \isactrlsub E\ e{\isacharparenright}{\kern0pt}{\isacartoucheclose}\ \isakeyword{and}\isanewline
\ \ \ \ j{\isacharcolon}{\kern0pt}\ {\isacartoucheopen}{\isasymexists}j{\isachardot}{\kern0pt}\ morphism\ {\isacharparenleft}{\kern0pt}lhs\ r\isactrlsub {\isadigit{2}}{\isacharparenright}{\kern0pt}\ D\isactrlsub {\isadigit{1}}\ j\ \isanewline
\ \ \ \ \ \ \ \ \ \ \ \ {\isasymand}\ {\isacharparenleft}{\kern0pt}{\isasymforall}v\ {\isasymin}\ V\isactrlbsub lhs\ r\isactrlsub {\isadigit{2}}\isactrlesub {\isachardot}{\kern0pt}\ \isactrlbsub c\isactrlsub {\isadigit{1}}\ {\isasymcirc}\isactrlsub {\isasymrightarrow}\ j\isactrlesub \isactrlsub V\ v\ {\isacharequal}{\kern0pt}\ \isactrlbsub g\isactrlsub {\isadigit{2}}\isactrlesub \isactrlsub V\ v{\isacharparenright}{\kern0pt}\ \isanewline
\ \ \ \ \ \ \ \ \ \ \ \ {\isasymand}\ {\isacharparenleft}{\kern0pt}{\isasymforall}e\ {\isasymin}\ E\isactrlbsub lhs\ r\isactrlsub {\isadigit{2}}\isactrlesub {\isachardot}{\kern0pt}\ \isactrlbsub c\isactrlsub {\isadigit{1}}\ {\isasymcirc}\isactrlsub {\isasymrightarrow}\ j\isactrlesub \isactrlsub E\ e\ {\isacharequal}{\kern0pt}\ \isactrlbsub g\isactrlsub {\isadigit{2}}\isactrlesub \isactrlsub E\ e{\isacharparenright}{\kern0pt}{\isacartoucheclose}%
}
\DefineSnippet{sequentialindependence}{
\isacommand{locale}\isamarkupfalse%
\ sequential{\isacharunderscore}{\kern0pt}independence\ {\isacharequal}{\kern0pt}\isanewline
\ \ p\isactrlsub {\isadigit{1}}{\isacharcolon}{\kern0pt}\ direct{\isacharunderscore}{\kern0pt}derivation\ r\isactrlsub {\isadigit{1}}\ b\isactrlsub {\isadigit{1}}\ b\isactrlsub {\isadigit{1}}{\isacharprime}{\kern0pt}\ G\ g\isactrlsub {\isadigit{1}}\ D\isactrlsub {\isadigit{1}}\ m\isactrlsub {\isadigit{1}}\ c\isactrlsub {\isadigit{1}}\ H\isactrlsub {\isadigit{1}}\ f\isactrlsub {\isadigit{1}}\ h\isactrlsub {\isadigit{1}}\ {\isacharplus}{\kern0pt}\isanewline
\ \ p\isactrlsub {\isadigit{2}}{\isacharcolon}{\kern0pt}\ direct{\isacharunderscore}{\kern0pt}derivation\ r\isactrlsub {\isadigit{2}}\ b\isactrlsub {\isadigit{2}}\ b\isactrlsub {\isadigit{2}}{\isacharprime}{\kern0pt}\ H\isactrlsub {\isadigit{1}}\ g\isactrlsub {\isadigit{2}}\ D\isactrlsub {\isadigit{2}}\ m\isactrlsub {\isadigit{2}}\ c\isactrlsub {\isadigit{2}}\ H\isactrlsub {\isadigit{2}}\ f\isactrlsub {\isadigit{2}}\ h\isactrlsub {\isadigit{2}}\isanewline
\ \ \isakeyword{for}\ r\isactrlsub {\isadigit{1}}\ b\isactrlsub {\isadigit{1}}\ b\isactrlsub {\isadigit{1}}{\isacharprime}{\kern0pt}\ G\ g\isactrlsub {\isadigit{1}}\ D\isactrlsub {\isadigit{1}}\ m\isactrlsub {\isadigit{1}}\ c\isactrlsub {\isadigit{1}}\ H\isactrlsub {\isadigit{1}}\ f\isactrlsub {\isadigit{1}}\ h\isactrlsub {\isadigit{1}}\ \isanewline
\ \ \ \ \ \ r\isactrlsub {\isadigit{2}}\ b\isactrlsub {\isadigit{2}}\ b\isactrlsub {\isadigit{2}}{\isacharprime}{\kern0pt}\ g\isactrlsub {\isadigit{2}}\ D\isactrlsub {\isadigit{2}}\ m\isactrlsub {\isadigit{2}}\ c\isactrlsub {\isadigit{2}}\ H\isactrlsub {\isadigit{2}}\ f\isactrlsub {\isadigit{2}}\ h\isactrlsub {\isadigit{2}}\ {\isacharplus}{\kern0pt}\isanewline
\ \ \isakeyword{assumes}\isanewline
\ \ \ \ i{\isacharcolon}{\kern0pt}\ {\isacartoucheopen}{\isasymexists}i{\isachardot}{\kern0pt}\ morphism\ {\isacharparenleft}{\kern0pt}rhs\ r\isactrlsub {\isadigit{1}}{\isacharparenright}{\kern0pt}\ D\isactrlsub {\isadigit{2}}\ i\ \isanewline
\ \ \ \ \ \ \ \ \ \ \ \ {\isasymand}\ {\isacharparenleft}{\kern0pt}{\isasymforall}v\ {\isasymin}\ V\isactrlbsub rhs\ r\isactrlsub {\isadigit{1}}\isactrlesub {\isachardot}{\kern0pt}\ \isactrlbsub c\isactrlsub {\isadigit{2}}\ {\isasymcirc}\isactrlsub {\isasymrightarrow}\ i\isactrlesub \isactrlsub V\ v\ {\isacharequal}{\kern0pt}\ \isactrlbsub f\isactrlsub {\isadigit{1}}\isactrlesub \isactrlsub V\ v{\isacharparenright}{\kern0pt}\ \isanewline
\ \ \ \ \ \ \ \ \ \ \ \ {\isasymand}\ {\isacharparenleft}{\kern0pt}{\isasymforall}e\ {\isasymin}\ E\isactrlbsub rhs\ r\isactrlsub {\isadigit{1}}\isactrlesub {\isachardot}{\kern0pt}\ \isactrlbsub c\isactrlsub {\isadigit{2}}\ {\isasymcirc}\isactrlsub {\isasymrightarrow}\ i\isactrlesub \isactrlsub E\ e\ {\isacharequal}{\kern0pt}\ \isactrlbsub f\isactrlsub {\isadigit{1}}\isactrlesub \isactrlsub E\ e{\isacharparenright}{\kern0pt}{\isacartoucheclose}\ \isakeyword{and}\isanewline
\ \ \ \ j{\isacharcolon}{\kern0pt}\ {\isacartoucheopen}{\isasymexists}j{\isachardot}{\kern0pt}\ morphism\ {\isacharparenleft}{\kern0pt}lhs\ r\isactrlsub {\isadigit{2}}{\isacharparenright}{\kern0pt}\ D\isactrlsub {\isadigit{1}}\ j\ \isanewline
\ \ \ \ \ \ \ \ \ \ \ \ {\isasymand}\ {\isacharparenleft}{\kern0pt}{\isasymforall}v\ {\isasymin}\ V\isactrlbsub lhs\ r\isactrlsub {\isadigit{2}}\isactrlesub {\isachardot}{\kern0pt}\ \isactrlbsub h\isactrlsub {\isadigit{1}}\ {\isasymcirc}\isactrlsub {\isasymrightarrow}\ j\isactrlesub \isactrlsub V\ v\ {\isacharequal}{\kern0pt}\ \isactrlbsub g\isactrlsub {\isadigit{2}}\isactrlesub \isactrlsub V\ v{\isacharparenright}{\kern0pt}\ \isanewline
\ \ \ \ \ \ \ \ \ \ \ \ {\isasymand}\ {\isacharparenleft}{\kern0pt}{\isasymforall}e\ {\isasymin}\ E\isactrlbsub lhs\ r\isactrlsub {\isadigit{2}}\isactrlesub {\isachardot}{\kern0pt}\ \isactrlbsub h\isactrlsub {\isadigit{1}}\ {\isasymcirc}\isactrlsub {\isasymrightarrow}\ j\isactrlesub \isactrlsub E\ e\ {\isacharequal}{\kern0pt}\ \isactrlbsub g\isactrlsub {\isadigit{2}}\isactrlesub \isactrlsub E\ e{\isacharparenright}{\kern0pt}{\isacartoucheclose}%
}
\DefineSnippet{churchrosser}{
\isacommand{theorem}\isamarkupfalse%
\ {\isacharparenleft}{\kern0pt}\isakeyword{in}\ parallel{\isacharunderscore}{\kern0pt}independence{\isacharparenright}{\kern0pt}\ church{\isacharunderscore}{\kern0pt}rosser{\isacharcolon}{\kern0pt}\isanewline
\ \ \isakeyword{shows}\ {\isacartoucheopen}{\isasymexists}g{\isacharprime}{\kern0pt}\ D{\isacharprime}{\kern0pt}\ m{\isacharprime}{\kern0pt}c{\isacharprime}{\kern0pt}\ H\ f{\isacharprime}{\kern0pt}\ h{\isacharprime}{\kern0pt}\ g{\isacharprime}{\kern0pt}{\isacharprime}{\kern0pt}\ D{\isacharprime}{\kern0pt}{\isacharprime}{\kern0pt}\ m{\isacharprime}{\kern0pt}{\isacharprime}{\kern0pt}\ c{\isacharprime}{\kern0pt}{\isacharprime}{\kern0pt}\ H\ f{\isacharprime}{\kern0pt}{\isacharprime}{\kern0pt}\ h{\isacharprime}{\kern0pt}{\isacharprime}{\kern0pt}{\isachardot}{\kern0pt}\ \isanewline
\ \ \ \ sequential{\isacharunderscore}{\kern0pt}independence\ r\isactrlsub {\isadigit{1}}\ b\isactrlsub {\isadigit{1}}\ b\isactrlsub {\isadigit{1}}{\isacharprime}{\kern0pt}\ G\ g\isactrlsub {\isadigit{1}}\ D\isactrlsub {\isadigit{1}}\ m\isactrlsub {\isadigit{1}}\ c\isactrlsub {\isadigit{1}}\ H\isactrlsub {\isadigit{1}}\ f\isactrlsub {\isadigit{1}}\ h\isactrlsub {\isadigit{1}}\ r\isactrlsub {\isadigit{2}}\ \isanewline
\ \ \ \ \ \ \ \ \ \ \ \ \ \ \ \ \ \ \ \ \ \ \ \ \ \ \ \ b\isactrlsub {\isadigit{2}}\ b\isactrlsub {\isadigit{2}}{\isacharprime}{\kern0pt}\ g{\isacharprime}{\kern0pt}\ D{\isacharprime}{\kern0pt}\ m{\isacharprime}{\kern0pt}\ c{\isacharprime}{\kern0pt}\ H\ f{\isacharprime}{\kern0pt}\ h{\isacharprime}{\kern0pt}\isanewline
\ \ {\isasymand}\ sequential{\isacharunderscore}{\kern0pt}independence\ r\isactrlsub {\isadigit{2}}\ b\isactrlsub {\isadigit{2}}\ b\isactrlsub {\isadigit{2}}{\isacharprime}{\kern0pt}\ G\ g\isactrlsub {\isadigit{2}}\ D\isactrlsub {\isadigit{2}}\ m\isactrlsub {\isadigit{2}}\ c\isactrlsub {\isadigit{2}}\ H\isactrlsub {\isadigit{2}}\ f\isactrlsub {\isadigit{2}}\ h\isactrlsub {\isadigit{2}}\ r\isactrlsub {\isadigit{1}}\isanewline
\ \ \ \ \ \ \ \ \ \ \ \ \ \ \ \ \ \ \ \ \ \ \ \ \ \ \ \ b\isactrlsub {\isadigit{1}}\ b\isactrlsub {\isadigit{1}}{\isacharprime}{\kern0pt}\ g{\isacharprime}{\kern0pt}{\isacharprime}{\kern0pt}\ D{\isacharprime}{\kern0pt}{\isacharprime}{\kern0pt}\ m{\isacharprime}{\kern0pt}{\isacharprime}{\kern0pt}\ c{\isacharprime}{\kern0pt}{\isacharprime}{\kern0pt}\ H\ f{\isacharprime}{\kern0pt}{\isacharprime}{\kern0pt}\ h{\isacharprime}{\kern0pt}{\isacharprime}{\kern0pt}{\isacartoucheclose}%
}
\DefineSnippet{churchrosser-pb}{
\isacommand{interpret}\isamarkupfalse%
\ {\isachardoublequoteopen}c{\isadigit{1}}{\isadigit{2}}{\isachardoublequoteclose}{\isacharcolon}{\kern0pt}\ pullback{\isacharunderscore}{\kern0pt}construction\ D\isactrlsub {\isadigit{1}}\ G\ D\isactrlsub {\isadigit{2}}\ c\isactrlsub {\isadigit{1}}\ c\isactrlsub {\isadigit{2}}\ \isacommand{{\isachardot}{\kern0pt}{\isachardot}{\kern0pt}}\isamarkupfalse%
}
\DefineSnippet{churchrosser-j1}{
\ \ \isacommand{obtain}\isamarkupfalse%
\ j\isactrlsub {\isadigit{1}}\ \isakeyword{where}\ {\isacartoucheopen}morphism\ {\isacharparenleft}{\kern0pt}interf\ r\isactrlsub {\isadigit{1}}{\isacharparenright}{\kern0pt}\ c{\isadigit{1}}{\isadigit{2}}{\isachardot}{\kern0pt}A\ j\isactrlsub {\isadigit{1}}{\isacartoucheclose}\isanewline
\ \ \ \ \isakeyword{and}\ {\isacartoucheopen}{\isasymAnd}v{\isachardot}{\kern0pt}\ v{\isasymin}V\isactrlbsub interf\ r\isactrlsub {\isadigit{1}}\isactrlesub \ {\isasymLongrightarrow}\ \isactrlbsub c{\isadigit{1}}{\isadigit{2}}{\isachardot}{\kern0pt}b\ {\isasymcirc}\isactrlsub {\isasymrightarrow}\ j\isactrlsub {\isadigit{1}}\isactrlesub \isactrlsub V\ v\ {\isacharequal}{\kern0pt}\ \isactrlbsub m\isactrlsub {\isadigit{1}}\isactrlesub \isactrlsub V\ v{\isacartoucheclose}\isanewline
\ \ \ \ \ \ \ \ {\isacartoucheopen}{\isasymAnd}e{\isachardot}{\kern0pt}\ e{\isasymin}E\isactrlbsub interf\ r\isactrlsub {\isadigit{1}}\isactrlesub \ {\isasymLongrightarrow}\ \isactrlbsub c{\isadigit{1}}{\isadigit{2}}{\isachardot}{\kern0pt}b\ {\isasymcirc}\isactrlsub {\isasymrightarrow}\ j\isactrlsub {\isadigit{1}}\isactrlesub \isactrlsub E\ e\ {\isacharequal}{\kern0pt}\ \isactrlbsub m\isactrlsub {\isadigit{1}}\isactrlesub \isactrlsub E\ e{\isacartoucheclose}\isanewline
\ \ \ \ \isakeyword{and}\ {\isacartoucheopen}{\isasymAnd}v{\isachardot}{\kern0pt}\ v{\isasymin}V\isactrlbsub interf\ r\isactrlsub {\isadigit{1}}\isactrlesub \ {\isasymLongrightarrow}\ \isactrlbsub c{\isadigit{1}}{\isadigit{2}}{\isachardot}{\kern0pt}c\ {\isasymcirc}\isactrlsub {\isasymrightarrow}\ j\isactrlsub {\isadigit{1}}\isactrlesub \isactrlsub V\ v\ {\isacharequal}{\kern0pt}\ \isactrlbsub i\isactrlsub {\isadigit{1}}\ {\isasymcirc}\isactrlsub {\isasymrightarrow}\ b\isactrlsub {\isadigit{1}}\isactrlesub \isactrlsub V\ v{\isacartoucheclose}\isanewline
\ \ \ \ \ \ \ \ {\isacartoucheopen}{\isasymAnd}e{\isachardot}{\kern0pt}\ e{\isasymin}E\isactrlbsub interf\ r\isactrlsub {\isadigit{1}}\isactrlesub \ {\isasymLongrightarrow}\ \isactrlbsub c{\isadigit{1}}{\isadigit{2}}{\isachardot}{\kern0pt}c\ {\isasymcirc}\isactrlsub {\isasymrightarrow}\ j\isactrlsub {\isadigit{1}}\isactrlesub \isactrlsub E\ e\ {\isacharequal}{\kern0pt}\ \isactrlbsub i\isactrlsub {\isadigit{1}}\ {\isasymcirc}\isactrlsub {\isasymrightarrow}\ b\isactrlsub {\isadigit{1}}\isactrlesub \isactrlsub E\ e{\isacartoucheclose}\isanewline
\ \ \ \ \isacommand{using}\isamarkupfalse%
\ c{\isadigit{1}}{\isadigit{2}}{\isachardot}{\kern0pt}pb{\isachardot}{\kern0pt}universal{\isacharunderscore}{\kern0pt}property{\isacharunderscore}{\kern0pt}exist{\isacharunderscore}{\kern0pt}gen{\isacharbrackleft}{\kern0pt}OF\ p\isactrlsub {\isadigit{1}}{\isachardot}{\kern0pt}r{\isachardot}{\kern0pt}k{\isachardot}{\kern0pt}G{\isachardot}{\kern0pt}graph{\isacharunderscore}{\kern0pt}axioms\isanewline
\ \ \ \ \ \ \ \ \ \ \ \ \ \ \ wf{\isacharunderscore}{\kern0pt}b\isactrlsub {\isadigit{1}}i\isactrlsub {\isadigit{1}}{\isachardot}{\kern0pt}morphism{\isacharunderscore}{\kern0pt}axioms\ p\isactrlsub {\isadigit{1}}{\isachardot}{\kern0pt}po{\isadigit{1}}{\isachardot}{\kern0pt}c{\isachardot}{\kern0pt}morphism{\isacharunderscore}{\kern0pt}axioms\ a\ b{\isacharbrackright}{\kern0pt}\isanewline
\ \ \ \ \isacommand{by}\isamarkupfalse%
\ fast%
}
\DefineSnippet{churchrosser-c21-c41}{
\ \ \isacommand{interpret}\isamarkupfalse%
\ {\isachardoublequoteopen}c{\isadigit{2}}{\isadigit{1}}{\isachardoublequoteclose}{\isacharcolon}{\kern0pt}\ gluing\ {\isachardoublequoteopen}interf\ r\isactrlsub {\isadigit{1}}{\isachardoublequoteclose}\ c{\isadigit{1}}{\isadigit{2}}{\isachardot}{\kern0pt}A\ {\isachardoublequoteopen}rhs\ r\isactrlsub {\isadigit{1}}{\isachardoublequoteclose}\ j\isactrlsub {\isadigit{1}}\ b\isactrlsub {\isadigit{1}}{\isacharprime}{\kern0pt}\ \isacommand{{\isachardot}{\kern0pt}{\isachardot}{\kern0pt}}\isamarkupfalse%
\isanewline
\ \ \isacommand{interpret}\isamarkupfalse%
\ {\isachardoublequoteopen}c{\isadigit{4}}{\isadigit{1}}{\isachardoublequoteclose}{\isacharcolon}{\kern0pt}\ gluing\ {\isachardoublequoteopen}interf\ r\isactrlsub {\isadigit{2}}{\isachardoublequoteclose}\ c{\isadigit{1}}{\isadigit{2}}{\isachardot}{\kern0pt}A\ {\isachardoublequoteopen}rhs\ r\isactrlsub {\isadigit{2}}{\isachardoublequoteclose}\ j\isactrlsub {\isadigit{2}}\ b\isactrlsub {\isadigit{2}}{\isacharprime}{\kern0pt}\ \isacommand{{\isachardot}{\kern0pt}{\isachardot}{\kern0pt}}\isamarkupfalse%
}
\DefineSnippet{churchrosser-31-22}{
\ \ \isacommand{interpret}\isamarkupfalse%
\ {\isachardoublequoteopen}{\isadigit{3}}{\isadigit{1}}{\isacharplus}{\kern0pt}{\isadigit{2}}{\isadigit{2}}{\isachardoublequoteclose}{\isacharcolon}{\kern0pt}\ pushout{\isacharunderscore}{\kern0pt}diagram\ {\isachardoublequoteopen}interf\ r\isactrlsub {\isadigit{2}}{\isachardoublequoteclose}\ c{\isadigit{2}}{\isadigit{1}}{\isachardot}{\kern0pt}H\ {\isachardoublequoteopen}lhs\ r\isactrlsub {\isadigit{2}}{\isachardoublequoteclose}\ H\isactrlsub {\isadigit{1}}\ \isanewline
\ \ \ \ \ \ \ \ \ \ \ \ \ \ \ \ \ \ \ \ \ \ \ \ \ \ \ \ \ \ \ \ \ \ \ \ {\isachardoublequoteopen}c{\isadigit{2}}{\isadigit{1}}{\isachardot}{\kern0pt}c\ {\isasymcirc}\isactrlsub {\isasymrightarrow}\ j\isactrlsub {\isadigit{2}}{\isachardoublequoteclose}\ b\isactrlsub {\isadigit{2}}\ s\isactrlsub {\isadigit{1}}\ {\isachardoublequoteopen}h\isactrlsub {\isadigit{1}}\ {\isasymcirc}\isactrlsub {\isasymrightarrow}\ i\isactrlsub {\isadigit{2}}{\isachardoublequoteclose}\isanewline
\ \ \ \ \isacommand{using}\isamarkupfalse%
\ pushout{\isacharunderscore}{\kern0pt}composition{\isacharbrackleft}{\kern0pt}OF\ \isanewline
\ \ \ \ \ \ \ \ {\isachardoublequoteopen}{\isadigit{3}}{\isadigit{1}}{\isachardot}{\kern0pt}flip{\isacharunderscore}{\kern0pt}diagram{\isachardoublequoteclose}\ {\isachardoublequoteopen}{\isadigit{2}}{\isadigit{2}}{\isachardot}{\kern0pt}flip{\isacharunderscore}{\kern0pt}diagram{\isachardoublequoteclose}\ {\isacharbrackright}{\kern0pt}\ \isacommand{by}\isamarkupfalse%
\ assumption%
}
\DefineSnippet{ngraph}{
\isacommand{type{\isacharunderscore}{\kern0pt}synonym}\isamarkupfalse%
\ {\isacharparenleft}{\kern0pt}{\isacharprime}{\kern0pt}l{\isacharcomma}{\kern0pt}{\isacharprime}{\kern0pt}m{\isacharparenright}{\kern0pt}\ ngraph\ {\isacharequal}{\kern0pt}\ {\isachardoublequoteopen}{\isacharparenleft}{\kern0pt}nat{\isacharcomma}{\kern0pt}nat{\isacharcomma}{\kern0pt}{\isacharprime}{\kern0pt}l{\isacharcomma}{\kern0pt}{\isacharprime}{\kern0pt}m{\isacharparenright}{\kern0pt}\ pre{\isacharunderscore}{\kern0pt}graph{\isachardoublequoteclose}%
}
\DefineSnippet{to-ngraph}{
\isacommand{definition}\isamarkupfalse%
\ to{\isacharunderscore}{\kern0pt}ngraph\ \isanewline
\ \ {\isacharcolon}{\kern0pt}{\isacharcolon}{\kern0pt}\ {\isachardoublequoteopen}{\isacharparenleft}{\kern0pt}{\isacharprime}{\kern0pt}v{\isacharcolon}{\kern0pt}{\isacharcolon}{\kern0pt}countable{\isacharcomma}{\kern0pt}{\isacharprime}{\kern0pt}e\ {\isacharcolon}{\kern0pt}{\isacharcolon}{\kern0pt}\ countable{\isacharcomma}{\kern0pt}{\isacharprime}{\kern0pt}l{\isacharcomma}{\kern0pt}{\isacharprime}{\kern0pt}m{\isacharparenright}{\kern0pt}\ pre{\isacharunderscore}{\kern0pt}graph\ \isanewline
\ \ \ {\isasymRightarrow}\ {\isacharparenleft}{\kern0pt}{\isacharprime}{\kern0pt}l{\isacharcomma}{\kern0pt}{\isacharprime}{\kern0pt}m{\isacharparenright}{\kern0pt}\ ngraph{\isachardoublequoteclose}\ \isakeyword{where}\isanewline
{\isacartoucheopen}to{\isacharunderscore}{\kern0pt}ngraph\ G\ {\isasymequiv}\ {\isasymlparr}nodes\ \ {\isacharequal}{\kern0pt}\ to{\isacharunderscore}{\kern0pt}nat\ {\isacharbackquote}{\kern0pt}\ V\isactrlbsub G\isactrlesub \isanewline
\ \ \ \ \ \ \ \ \ \ \ \ \ \ \ {\isacharcomma}{\kern0pt}edges\ \ {\isacharequal}{\kern0pt}\ to{\isacharunderscore}{\kern0pt}nat\ {\isacharbackquote}{\kern0pt}\ E\isactrlbsub G\isactrlesub \isanewline
\ \ \ \ \ \ \ \ \ \ \ \ \ \ \ {\isacharcomma}{\kern0pt}source\ {\isacharequal}{\kern0pt}\ {\isasymlambda}e{\isachardot}{\kern0pt}\ to{\isacharunderscore}{\kern0pt}nat\ {\isacharparenleft}{\kern0pt}s\isactrlbsub G\isactrlesub \ {\isacharparenleft}{\kern0pt}from{\isacharunderscore}{\kern0pt}nat\ e{\isacharparenright}{\kern0pt}{\isacharparenright}{\kern0pt}\isanewline
\ \ \ \ \ \ \ \ \ \ \ \ \ \ \ {\isacharcomma}{\kern0pt}target\ {\isacharequal}{\kern0pt}\ {\isasymlambda}e{\isachardot}{\kern0pt}\ to{\isacharunderscore}{\kern0pt}nat\ {\isacharparenleft}{\kern0pt}t\isactrlbsub G\isactrlesub \ {\isacharparenleft}{\kern0pt}from{\isacharunderscore}{\kern0pt}nat\ e{\isacharparenright}{\kern0pt}{\isacharparenright}{\kern0pt}\isanewline
\ \ \ \ \ \ \ \ \ \ \ \ \ \ \ {\isacharcomma}{\kern0pt}node{\isacharunderscore}{\kern0pt}label\ {\isacharequal}{\kern0pt}\ {\isasymlambda}v{\isachardot}{\kern0pt}\ l\isactrlbsub G\isactrlesub \ {\isacharparenleft}{\kern0pt}from{\isacharunderscore}{\kern0pt}nat\ v{\isacharparenright}{\kern0pt}\isanewline
\ \ \ \ \ \ \ \ \ \ \ \ \ \ \ {\isacharcomma}{\kern0pt}edge{\isacharunderscore}{\kern0pt}label\ {\isacharequal}{\kern0pt}\ {\isasymlambda}e{\isachardot}{\kern0pt}\ m\isactrlbsub G\isactrlesub \ {\isacharparenleft}{\kern0pt}from{\isacharunderscore}{\kern0pt}nat\ e{\isacharparenright}{\kern0pt}{\isasymrparr}{\isacartoucheclose}%
}
\DefineSnippet{from-ngraph}{
\isacommand{definition}\isamarkupfalse%
\ from{\isacharunderscore}{\kern0pt}ngraph\ \isanewline
\ \ {\isacharcolon}{\kern0pt}{\isacharcolon}{\kern0pt}\ {\isachardoublequoteopen}\ {\isacharparenleft}{\kern0pt}{\isacharprime}{\kern0pt}l{\isacharcomma}{\kern0pt}{\isacharprime}{\kern0pt}m{\isacharparenright}{\kern0pt}\ ngraph\ \isanewline
\ \ {\isasymRightarrow}\ {\isacharparenleft}{\kern0pt}{\isacharprime}{\kern0pt}v{\isacharcolon}{\kern0pt}{\isacharcolon}{\kern0pt}countable{\isacharcomma}{\kern0pt}{\isacharprime}{\kern0pt}e\ {\isacharcolon}{\kern0pt}{\isacharcolon}{\kern0pt}\ countable{\isacharcomma}{\kern0pt}{\isacharprime}{\kern0pt}l{\isacharcomma}{\kern0pt}{\isacharprime}{\kern0pt}m{\isacharparenright}{\kern0pt}\ pre{\isacharunderscore}{\kern0pt}graph{\isachardoublequoteclose}\ \isakeyword{where}\isanewline
{\isacartoucheopen}from{\isacharunderscore}{\kern0pt}ngraph\ G\ {\isasymequiv}\ {\isasymlparr}nodes\ {\isacharequal}{\kern0pt}\ from{\isacharunderscore}{\kern0pt}nat\ {\isacharbackquote}{\kern0pt}\ V\isactrlbsub G\isactrlesub \isanewline
\ \ \ \ \ \ \ \ \ \ \ \ \ \ \ \ \ {\isacharcomma}{\kern0pt}edges\ {\isacharequal}{\kern0pt}\ from{\isacharunderscore}{\kern0pt}nat\ {\isacharbackquote}{\kern0pt}\ E\isactrlbsub G\isactrlesub \isanewline
\ \ \ \ \ \ \ \ \ \ \ \ \ \ \ \ \ {\isacharcomma}{\kern0pt}source\ {\isacharequal}{\kern0pt}\ {\isasymlambda}e{\isachardot}{\kern0pt}\ from{\isacharunderscore}{\kern0pt}nat\ {\isacharparenleft}{\kern0pt}s\isactrlbsub G\isactrlesub \ {\isacharparenleft}{\kern0pt}to{\isacharunderscore}{\kern0pt}nat\ e{\isacharparenright}{\kern0pt}{\isacharparenright}{\kern0pt}\isanewline
\ \ \ \ \ \ \ \ \ \ \ \ \ \ \ \ \ {\isacharcomma}{\kern0pt}target\ {\isacharequal}{\kern0pt}\ {\isasymlambda}e{\isachardot}{\kern0pt}\ from{\isacharunderscore}{\kern0pt}nat\ {\isacharparenleft}{\kern0pt}t\isactrlbsub G\isactrlesub \ {\isacharparenleft}{\kern0pt}to{\isacharunderscore}{\kern0pt}nat\ e{\isacharparenright}{\kern0pt}{\isacharparenright}{\kern0pt}\isanewline
\ \ \ \ \ \ \ \ \ \ \ \ \ \ \ \ \ {\isacharcomma}{\kern0pt}node{\isacharunderscore}{\kern0pt}label\ {\isacharequal}{\kern0pt}\ {\isasymlambda}e{\isachardot}{\kern0pt}\ l\isactrlbsub G\isactrlesub \ {\isacharparenleft}{\kern0pt}to{\isacharunderscore}{\kern0pt}nat\ e{\isacharparenright}{\kern0pt}\isanewline
\ \ \ \ \ \ \ \ \ \ \ \ \ \ \ \ \ {\isacharcomma}{\kern0pt}edge{\isacharunderscore}{\kern0pt}label\ {\isacharequal}{\kern0pt}\ {\isasymlambda}e{\isachardot}{\kern0pt}\ m\isactrlbsub G\isactrlesub \ {\isacharparenleft}{\kern0pt}to{\isacharunderscore}{\kern0pt}nat\ e{\isacharparenright}{\kern0pt}{\isasymrparr}{\isacartoucheclose}%
}
\DefineSnippet{graph_iff_ngraph}{
\isacommand{lemma}\isamarkupfalse%
\ graph{\isacharunderscore}{\kern0pt}ngraph{\isacharunderscore}{\kern0pt}corres{\isacharunderscore}{\kern0pt}iff{\isacharcolon}{\kern0pt}\ \isanewline
\ \ {\isacartoucheopen}graph\ {\isacharparenleft}{\kern0pt}to{\isacharunderscore}{\kern0pt}ngraph\ G{\isacharparenright}{\kern0pt}\ {\isasymlongleftrightarrow}\ graph\ G\ {\isacartoucheclose}%
}
\DefineSnippet{morph_iff_nmorph}{
\isacommand{lemma}\isamarkupfalse%
\ morph{\isacharunderscore}{\kern0pt}eq{\isacharunderscore}{\kern0pt}nmorph{\isacharunderscore}{\kern0pt}iff{\isacharcolon}{\kern0pt}\ \isanewline
\ \ {\isacartoucheopen}morphism\ G\ H\ m\ {\isasymlongleftrightarrow}\ morphism\ {\isacharparenleft}{\kern0pt}to{\isacharunderscore}{\kern0pt}ngraph\ G{\isacharparenright}{\kern0pt}\ {\isacharparenleft}{\kern0pt}to{\isacharunderscore}{\kern0pt}ngraph\ H{\isacharparenright}{\kern0pt}\ {\isacharparenleft}{\kern0pt}to{\isacharunderscore}{\kern0pt}nmorph\ m{\isacharparenright}{\kern0pt}{\isacartoucheclose}%
}
\DefineSnippet{pushout-ex1m}{
\isacommand{abbreviation}\isamarkupfalse%
\ Ex{\isadigit{1}}M\ \isanewline
\ \ {\isacharcolon}{\kern0pt}{\isacharcolon}{\kern0pt}\ {\isachardoublequoteopen}{\isacharparenleft}{\kern0pt}{\isacharparenleft}{\kern0pt}{\isacharprime}{\kern0pt}v\isactrlsub {\isadigit{1}}{\isacharcomma}{\kern0pt}{\isacharprime}{\kern0pt}v\isactrlsub {\isadigit{2}}{\isacharcomma}{\kern0pt}{\isacharprime}{\kern0pt}e\isactrlsub {\isadigit{1}}{\isacharcomma}{\kern0pt}{\isacharprime}{\kern0pt}e\isactrlsub {\isadigit{2}}{\isacharparenright}{\kern0pt}\ pre{\isacharunderscore}{\kern0pt}morph\ {\isasymRightarrow}\ bool{\isacharparenright}{\kern0pt}\ \isanewline
\ \ {\isasymRightarrow}\ {\isacharparenleft}{\kern0pt}{\isacharprime}{\kern0pt}v\isactrlsub {\isadigit{1}}{\isacharcomma}{\kern0pt}{\isacharprime}{\kern0pt}e\isactrlsub {\isadigit{1}}{\isacharcomma}{\kern0pt}{\isacharprime}{\kern0pt}l{\isacharcomma}{\kern0pt}{\isacharprime}{\kern0pt}m{\isacharparenright}{\kern0pt}\ pre{\isacharunderscore}{\kern0pt}graph\ \isanewline
\ \ {\isasymRightarrow}\ bool{\isachardoublequoteclose}\ \ \isakeyword{where}\ \isanewline
\ \ {\isachardoublequoteopen}Ex{\isadigit{1}}M\ P\ E\ {\isasymequiv}\ {\isasymexists}x{\isachardot}{\kern0pt}\ P\ x\ {\isasymand}\ {\isacharparenleft}{\kern0pt}{\isasymforall}y{\isachardot}{\kern0pt}\ P\ y\ \isanewline
\ \ \ \ \ \ \ \ \ \ \ \ \ \ \ \ {\isasymlongrightarrow}\ {\isacharparenleft}{\kern0pt}\ {\isacharparenleft}{\kern0pt}{\isasymforall}e\ {\isasymin}\ E\isactrlbsub E\isactrlesub {\isachardot}{\kern0pt}\ \isactrlbsub y\isactrlesub \isactrlsub E\ e\ {\isacharequal}{\kern0pt}\ \isactrlbsub x\isactrlesub \isactrlsub E\ e{\isacharparenright}{\kern0pt}\ \isanewline
\ \ \ \ \ \ \ \ \ \ \ \ \ \ \ \ \ \ \ \ {\isasymand}\ {\isacharparenleft}{\kern0pt}{\isasymforall}v\ {\isasymin}\ V\isactrlbsub E\isactrlesub {\isachardot}{\kern0pt}\ \isactrlbsub y\isactrlesub \isactrlsub V\ v\ {\isacharequal}{\kern0pt}\ \isactrlbsub x\isactrlesub \isactrlsub V\ v{\isacharparenright}{\kern0pt}{\isacharparenright}{\kern0pt}{\isacharparenright}{\kern0pt}{\isachardoublequoteclose}%
}
\DefineSnippet{pushout-diagram}{
\isacommand{locale}\isamarkupfalse%
\ pushout{\isacharunderscore}{\kern0pt}diagram\ {\isacharequal}{\kern0pt}\isanewline
\ \ b{\isacharcolon}{\kern0pt}\ morphism\ A\ B\ b\ {\isacharplus}{\kern0pt}\isanewline
\ \ c{\isacharcolon}{\kern0pt}\ morphism\ A\ C\ c\ {\isacharplus}{\kern0pt}\isanewline
\ \ f{\isacharcolon}{\kern0pt}\ morphism\ B\ D\ f\ {\isacharplus}{\kern0pt}\isanewline
\ \ g{\isacharcolon}{\kern0pt}\ morphism\ C\ D\ g\ \isakeyword{for}\ A\ B\ C\ D\ b\ c\ f\ g\ {\isacharplus}{\kern0pt}\isanewline
\isakeyword{assumes}\isanewline
\ \ node{\isacharunderscore}{\kern0pt}commutativity{\isacharcolon}{\kern0pt}\ {\isacartoucheopen}v\ {\isasymin}\ V\isactrlbsub A\isactrlesub \ {\isasymLongrightarrow}\ \isactrlbsub f\ {\isasymcirc}\isactrlsub {\isasymrightarrow}\ b\isactrlesub \isactrlsub V\ v\ {\isacharequal}{\kern0pt}\ \isactrlbsub g\ {\isasymcirc}\isactrlsub {\isasymrightarrow}\ c\isactrlesub \isactrlsub V\ v{\isacartoucheclose}\ \isakeyword{and}\isanewline
\ \ edge{\isacharunderscore}{\kern0pt}commutativity{\isacharcolon}{\kern0pt}\ {\isacartoucheopen}e\ {\isasymin}\ E\isactrlbsub A\isactrlesub \ {\isasymLongrightarrow}\ \isactrlbsub f\ {\isasymcirc}\isactrlsub {\isasymrightarrow}\ b\isactrlesub \isactrlsub E\ e\ {\isacharequal}{\kern0pt}\ \isactrlbsub g\ {\isasymcirc}\isactrlsub {\isasymrightarrow}\ c\isactrlesub \isactrlsub E\ e{\isacartoucheclose}\ \isakeyword{and}\isanewline
\ \ universal{\isacharunderscore}{\kern0pt}property{\isacharcolon}{\kern0pt}\ {\isacartoucheopen}{\isasymlbrakk}\isanewline
\ \ \ \ graph\ {\isacharparenleft}{\kern0pt}D{\isacharprime}{\kern0pt}\ {\isacharcolon}{\kern0pt}{\isacharcolon}{\kern0pt}\ {\isacharparenleft}{\kern0pt}{\isacharprime}{\kern0pt}c{\isacharcomma}{\kern0pt}{\isacharprime}{\kern0pt}d{\isacharparenright}{\kern0pt}\ ngraph{\isacharparenright}{\kern0pt}{\isacharsemicolon}{\kern0pt}\ \isanewline
\ \ \ \ morphism\ {\isacharparenleft}{\kern0pt}to{\isacharunderscore}{\kern0pt}ngraph\ B{\isacharparenright}{\kern0pt}\ D{\isacharprime}{\kern0pt}\ x{\isacharsemicolon}{\kern0pt}\ \isanewline
\ \ \ \ morphism\ {\isacharparenleft}{\kern0pt}to{\isacharunderscore}{\kern0pt}ngraph\ C{\isacharparenright}{\kern0pt}\ D{\isacharprime}{\kern0pt}\ y{\isacharsemicolon}{\kern0pt}\isanewline
\ \ \ \ \ {\isasymforall}v\ {\isasymin}\ V\isactrlbsub to{\isacharunderscore}{\kern0pt}ngraph\ A\isactrlesub {\isachardot}{\kern0pt}\ \isactrlbsub x\ {\isasymcirc}\isactrlsub {\isasymrightarrow}\ {\isacharparenleft}{\kern0pt}to{\isacharunderscore}{\kern0pt}nmorph\ b{\isacharparenright}{\kern0pt}\isactrlesub \isactrlsub V\ v\ {\isacharequal}{\kern0pt}\ \isactrlbsub y\ {\isasymcirc}\isactrlsub {\isasymrightarrow}\ {\isacharparenleft}{\kern0pt}to{\isacharunderscore}{\kern0pt}nmorph\ c{\isacharparenright}{\kern0pt}\isactrlesub \isactrlsub V\ v{\isacharsemicolon}{\kern0pt}\isanewline
\ \ \ \ \ {\isasymforall}e\ {\isasymin}\ E\isactrlbsub to{\isacharunderscore}{\kern0pt}ngraph\ A\isactrlesub {\isachardot}{\kern0pt}\ \isactrlbsub x\ {\isasymcirc}\isactrlsub {\isasymrightarrow}\ {\isacharparenleft}{\kern0pt}to{\isacharunderscore}{\kern0pt}nmorph\ b{\isacharparenright}{\kern0pt}\isactrlesub \isactrlsub E\ e\ {\isacharequal}{\kern0pt}\ \isactrlbsub y\ {\isasymcirc}\isactrlsub {\isasymrightarrow}\ {\isacharparenleft}{\kern0pt}to{\isacharunderscore}{\kern0pt}nmorph\ c{\isacharparenright}{\kern0pt}\isactrlesub \isactrlsub E\ e{\isasymrbrakk}\ \isanewline
\ \ \ \ \ \ {\isasymLongrightarrow}\ Ex{\isadigit{1}}M\ {\isacharparenleft}{\kern0pt}{\isasymlambda}u{\isachardot}{\kern0pt}\ morphism\ {\isacharparenleft}{\kern0pt}to{\isacharunderscore}{\kern0pt}ngraph\ D{\isacharparenright}{\kern0pt}\ D{\isacharprime}{\kern0pt}\ u\ {\isasymand}\isanewline
\ \ \ \ \ \ \ \ \ \ \ \ {\isacharparenleft}{\kern0pt}{\isasymforall}v\ {\isasymin}\ V\isactrlbsub to{\isacharunderscore}{\kern0pt}ngraph\ B\isactrlesub {\isachardot}{\kern0pt}\ \isactrlbsub u\ {\isasymcirc}\isactrlsub {\isasymrightarrow}\ {\isacharparenleft}{\kern0pt}to{\isacharunderscore}{\kern0pt}nmorph\ f{\isacharparenright}{\kern0pt}\isactrlesub \isactrlsub V\ v\ {\isacharequal}{\kern0pt}\ \isactrlbsub x\isactrlesub \isactrlsub V\ v{\isacharparenright}{\kern0pt}\ {\isasymand}\isanewline
\ \ \ \ \ \ \ \ \ \ \ \ {\isacharparenleft}{\kern0pt}{\isasymforall}e\ {\isasymin}\ E\isactrlbsub to{\isacharunderscore}{\kern0pt}ngraph\ B\isactrlesub {\isachardot}{\kern0pt}\ \isactrlbsub u\ {\isasymcirc}\isactrlsub {\isasymrightarrow}\ {\isacharparenleft}{\kern0pt}to{\isacharunderscore}{\kern0pt}nmorph\ f{\isacharparenright}{\kern0pt}\isactrlesub \isactrlsub E\ e\ {\isacharequal}{\kern0pt}\ \isactrlbsub x\isactrlesub \isactrlsub E\ e{\isacharparenright}{\kern0pt}\ {\isasymand}\isanewline
\ \ \ \ \ \ \ \ \ \ \ \ {\isacharparenleft}{\kern0pt}{\isasymforall}v\ {\isasymin}\ V\isactrlbsub to{\isacharunderscore}{\kern0pt}ngraph\ C\isactrlesub {\isachardot}{\kern0pt}\ \isactrlbsub u\ {\isasymcirc}\isactrlsub {\isasymrightarrow}\ {\isacharparenleft}{\kern0pt}to{\isacharunderscore}{\kern0pt}nmorph\ g{\isacharparenright}{\kern0pt}\isactrlesub \isactrlsub V\ v\ {\isacharequal}{\kern0pt}\ \isactrlbsub y\isactrlesub \isactrlsub V\ v{\isacharparenright}{\kern0pt}\ {\isasymand}\isanewline
\ \ \ \ \ \ \ \ \ \ \ \ {\isacharparenleft}{\kern0pt}{\isasymforall}e\ {\isasymin}\ E\isactrlbsub to{\isacharunderscore}{\kern0pt}ngraph\ C\isactrlesub {\isachardot}{\kern0pt}\ \isactrlbsub u\ {\isasymcirc}\isactrlsub {\isasymrightarrow}\ {\isacharparenleft}{\kern0pt}to{\isacharunderscore}{\kern0pt}nmorph\ g{\isacharparenright}{\kern0pt}\isactrlesub \isactrlsub E\ e\ {\isacharequal}{\kern0pt}\ \isactrlbsub y\isactrlesub \isactrlsub E\ e{\isacharparenright}{\kern0pt}{\isacharparenright}{\kern0pt}\isanewline
\ \ \ \ \ \ \ \ \ \ \ \ {\isacharparenleft}{\kern0pt}to{\isacharunderscore}{\kern0pt}ngraph\ D{\isacharparenright}{\kern0pt}{\isacartoucheclose}%
}
\DefineSnippet{pushout-flip}{
\isacommand{lemma}\isamarkupfalse%
\ flip{\isacharunderscore}{\kern0pt}diagram{\isacharcolon}{\kern0pt}\isanewline
\ \ {\isacartoucheopen}pushout{\isacharunderscore}{\kern0pt}diagram\ A\ C\ B\ D\ c\ b\ g\ f{\isacartoucheclose}%
}
\DefineSnippet{pushout-univ}{
\isacommand{lemma}\isamarkupfalse%
\ universal{\isacharunderscore}{\kern0pt}property{\isacharunderscore}{\kern0pt}exist{\isacharunderscore}{\kern0pt}gen{\isacharcolon}{\kern0pt}\isanewline
\ \ \isakeyword{fixes}\ D{\isacharprime}{\kern0pt}\isanewline
\ \ \isakeyword{assumes}\ {\isacartoucheopen}graph\ D{\isacharprime}{\kern0pt}{\isacartoucheclose}\ {\isacartoucheopen}morphism\ B\ D{\isacharprime}{\kern0pt}\ x{\isacartoucheclose}\ {\isacartoucheopen}morphism\ C\ D{\isacharprime}{\kern0pt}\ y{\isacartoucheclose}\isanewline
\ \ \ \ {\isacartoucheopen}{\isasymforall}v\ {\isasymin}\ V\isactrlbsub A\isactrlesub {\isachardot}{\kern0pt}\ \isactrlbsub x\ {\isasymcirc}\isactrlsub {\isasymrightarrow}\ b\isactrlesub \isactrlsub V\ v\ {\isacharequal}{\kern0pt}\ \isactrlbsub y\ {\isasymcirc}\isactrlsub {\isasymrightarrow}\ c\isactrlesub \isactrlsub V\ v{\isacartoucheclose}\isanewline
\ \ \ \ {\isacartoucheopen}{\isasymforall}e\ {\isasymin}\ E\isactrlbsub A\isactrlesub {\isachardot}{\kern0pt}\ \isactrlbsub x\ {\isasymcirc}\isactrlsub {\isasymrightarrow}\ b\isactrlesub \isactrlsub E\ e\ {\isacharequal}{\kern0pt}\ \isactrlbsub y\ {\isasymcirc}\isactrlsub {\isasymrightarrow}\ c\isactrlesub \isactrlsub E\ e{\isacartoucheclose}\isanewline
\ \ \isakeyword{shows}\ {\isacartoucheopen}Ex{\isadigit{1}}M\ {\isacharparenleft}{\kern0pt}{\isasymlambda}u{\isachardot}{\kern0pt}\ morphism\ D\ D{\isacharprime}{\kern0pt}\ u\ {\isasymand}\isanewline
\ \ \ \ \ \ \ \ \ \ \ \ {\isacharparenleft}{\kern0pt}{\isasymforall}v\ {\isasymin}\ V\isactrlbsub B\isactrlesub {\isachardot}{\kern0pt}\ \isactrlbsub u\ {\isasymcirc}\isactrlsub {\isasymrightarrow}\ f\isactrlesub \isactrlsub V\ v\ {\isacharequal}{\kern0pt}\ \isactrlbsub x\isactrlesub \isactrlsub V\ v{\isacharparenright}{\kern0pt}\ {\isasymand}\isanewline
\ \ \ \ \ \ \ \ \ \ \ \ {\isacharparenleft}{\kern0pt}{\isasymforall}e\ {\isasymin}\ E\isactrlbsub B\isactrlesub {\isachardot}{\kern0pt}\ \isactrlbsub u\ {\isasymcirc}\isactrlsub {\isasymrightarrow}\ f\isactrlesub \isactrlsub E\ e\ {\isacharequal}{\kern0pt}\ \isactrlbsub x\isactrlesub \isactrlsub E\ e{\isacharparenright}{\kern0pt}\ {\isasymand}\isanewline
\ \ \ \ \ \ \ \ \ \ \ \ {\isacharparenleft}{\kern0pt}{\isasymforall}v\ {\isasymin}\ V\isactrlbsub C\isactrlesub {\isachardot}{\kern0pt}\ \isactrlbsub u\ {\isasymcirc}\isactrlsub {\isasymrightarrow}\ g\isactrlesub \isactrlsub V\ v\ {\isacharequal}{\kern0pt}\ \isactrlbsub y\isactrlesub \isactrlsub V\ v{\isacharparenright}{\kern0pt}\ {\isasymand}\isanewline
\ \ \ \ \ \ \ \ \ \ \ \ {\isacharparenleft}{\kern0pt}{\isasymforall}e\ {\isasymin}\ E\isactrlbsub C\isactrlesub {\isachardot}{\kern0pt}\ \isactrlbsub u\ {\isasymcirc}\isactrlsub {\isasymrightarrow}\ g\isactrlesub \isactrlsub E\ e\ {\isacharequal}{\kern0pt}\ \isactrlbsub y\isactrlesub \isactrlsub E\ e{\isacharparenright}{\kern0pt}{\isacharparenright}{\kern0pt}\ D{\isacartoucheclose}%
}
\DefineSnippet{pushout-b-inj-g}{
\isacommand{lemma}\isamarkupfalse%
\ b{\isacharunderscore}{\kern0pt}inj{\isacharunderscore}{\kern0pt}imp{\isacharunderscore}{\kern0pt}g{\isacharunderscore}{\kern0pt}inj{\isacharcolon}{\kern0pt}\isanewline
\ \ \isakeyword{assumes}\ {\isacartoucheopen}injective{\isacharunderscore}{\kern0pt}morphism\ A\ B\ b{\isacartoucheclose}\isanewline
\ \ \isakeyword{shows}\ {\isacartoucheopen}injective{\isacharunderscore}{\kern0pt}morphism\ C\ D\ g{\isacartoucheclose}%
}
\DefineSnippet{pushout-joint-surj}{
\isacommand{lemma}\isamarkupfalse%
\ joint{\isacharunderscore}{\kern0pt}surjectivity{\isacharunderscore}{\kern0pt}edges{\isacharcolon}{\kern0pt}\isanewline
\ \ \isakeyword{fixes}\ x\isanewline
\ \ \isakeyword{assumes}\ {\isacartoucheopen}x\ {\isasymin}\ E\isactrlbsub D\isactrlesub {\isacartoucheclose}\isanewline
\ \ \isakeyword{shows}\ {\isacartoucheopen}{\isacharparenleft}{\kern0pt}{\isasymexists}e\ {\isasymin}\ E\isactrlbsub C\isactrlesub {\isachardot}{\kern0pt}\ \isactrlbsub g\isactrlesub \isactrlsub E\ e\ {\isacharequal}{\kern0pt}\ x{\isacharparenright}{\kern0pt}\ {\isasymor}\ {\isacharparenleft}{\kern0pt}{\isasymexists}e\ {\isasymin}\ E\isactrlbsub B\isactrlesub {\isachardot}{\kern0pt}\ \isactrlbsub f\isactrlesub \isactrlsub E\ e\ {\isacharequal}{\kern0pt}\ x{\isacharparenright}{\kern0pt}{\isacartoucheclose}%
}
\DefineSnippet{po-surj-b-is-g}{
\isacommand{lemma}\isamarkupfalse%
\ b{\isacharunderscore}{\kern0pt}surj{\isacharunderscore}{\kern0pt}imp{\isacharunderscore}{\kern0pt}g{\isacharunderscore}{\kern0pt}surj{\isacharcolon}{\kern0pt}\isanewline
\ \ \isakeyword{assumes}\ {\isacartoucheopen}surjective{\isacharunderscore}{\kern0pt}morphism\ A\ B\ b{\isacartoucheclose}\isanewline
\ \ \isakeyword{shows}\ {\isacartoucheopen}surjective{\isacharunderscore}{\kern0pt}morphism\ C\ D\ g{\isacartoucheclose}%
}
\DefineSnippet{po-bij-b-is-g}{
\isacommand{lemma}\isamarkupfalse%
\ b{\isacharunderscore}{\kern0pt}bij{\isacharunderscore}{\kern0pt}imp{\isacharunderscore}{\kern0pt}g{\isacharunderscore}{\kern0pt}bij{\isacharcolon}{\kern0pt}\isanewline
\ \ \isakeyword{assumes}\ {\isacartoucheopen}bijective{\isacharunderscore}{\kern0pt}morphism\ A\ B\ b{\isacartoucheclose}\isanewline
\ \ \isakeyword{shows}\ {\isacartoucheopen}bijective{\isacharunderscore}{\kern0pt}morphism\ C\ D\ g{\isacartoucheclose}%
}
\DefineSnippet{pushout-uniq}{
\isacommand{theorem}\isamarkupfalse%
\ uniqueness{\isacharunderscore}{\kern0pt}po{\isacharcolon}{\kern0pt}\isanewline
\ \ \isakeyword{fixes}\ D{\isacharprime}{\kern0pt}\isanewline
\ \ \isakeyword{assumes}\ \isanewline
\ \ \ \ D{\isacharprime}{\kern0pt}{\isacharcolon}{\kern0pt}\ {\isacartoucheopen}graph\ D{\isacharprime}{\kern0pt}{\isacartoucheclose}\ \isakeyword{and}\ \isanewline
\ \ \ \ f{\isacharprime}{\kern0pt}{\isacharcolon}{\kern0pt}\ {\isacartoucheopen}morphism\ B\ D{\isacharprime}{\kern0pt}\ f{\isacharprime}{\kern0pt}{\isacartoucheclose}\ \isakeyword{and}\ \isanewline
\ \ \ \ g{\isacharprime}{\kern0pt}{\isacharcolon}{\kern0pt}\ {\isacartoucheopen}morphism\ C\ D{\isacharprime}{\kern0pt}\ g{\isacharprime}{\kern0pt}{\isacartoucheclose}\isanewline
\ \ \isakeyword{shows}\ {\isacartoucheopen}pushout{\isacharunderscore}{\kern0pt}diagram\ \ A\ B\ C\ D{\isacharprime}{\kern0pt}\ b\ c\ f{\isacharprime}{\kern0pt}\ g{\isacharprime}{\kern0pt}\ \isanewline
\ \ \ \ {\isasymlongleftrightarrow}\ {\isacharparenleft}{\kern0pt}{\isasymexists}u{\isachardot}{\kern0pt}\ bijective{\isacharunderscore}{\kern0pt}morphism\ D\ D{\isacharprime}{\kern0pt}\ u\ \isanewline
\ \ \ \ \ \ \ \ \ \ {\isasymand}\ {\isacharparenleft}{\kern0pt}{\isasymforall}v\ {\isasymin}\ V\isactrlbsub B\isactrlesub {\isachardot}{\kern0pt}\ \isactrlbsub u\ {\isasymcirc}\isactrlsub {\isasymrightarrow}\ f\isactrlesub \isactrlsub V\ v\ {\isacharequal}{\kern0pt}\ \isactrlbsub f{\isacharprime}{\kern0pt}\isactrlesub \isactrlsub V\ v{\isacharparenright}{\kern0pt}\ {\isasymand}\ {\isacharparenleft}{\kern0pt}{\isasymforall}e\ {\isasymin}\ E\isactrlbsub B\isactrlesub {\isachardot}{\kern0pt}\ \isactrlbsub u\ {\isasymcirc}\isactrlsub {\isasymrightarrow}\ f\isactrlesub \isactrlsub E\ e\ {\isacharequal}{\kern0pt}\ \isactrlbsub f{\isacharprime}{\kern0pt}\isactrlesub \isactrlsub E\ e{\isacharparenright}{\kern0pt}\isanewline
\ \ \ \ \ \ \ \ \ \ {\isasymand}\ {\isacharparenleft}{\kern0pt}{\isasymforall}v\ {\isasymin}\ V\isactrlbsub C\isactrlesub {\isachardot}{\kern0pt}\ \isactrlbsub u\ {\isasymcirc}\isactrlsub {\isasymrightarrow}\ g\isactrlesub \isactrlsub V\ v\ {\isacharequal}{\kern0pt}\ \isactrlbsub g{\isacharprime}{\kern0pt}\isactrlesub \isactrlsub V\ v{\isacharparenright}{\kern0pt}\ {\isasymand}\ {\isacharparenleft}{\kern0pt}{\isasymforall}e\ {\isasymin}\ E\isactrlbsub C\isactrlesub {\isachardot}{\kern0pt}\ \isactrlbsub u\ {\isasymcirc}\isactrlsub {\isasymrightarrow}\ g\isactrlesub \isactrlsub E\ e\ {\isacharequal}{\kern0pt}\ \isactrlbsub g{\isacharprime}{\kern0pt}\isactrlesub \isactrlsub E\ e{\isacharparenright}{\kern0pt}{\isacharparenright}{\kern0pt}{\isacartoucheclose}%
}
\DefineSnippet{pushout-composition}{
\isacommand{lemma}\isamarkupfalse%
\ pushout{\isacharunderscore}{\kern0pt}composition{\isacharcolon}{\kern0pt}\isanewline
\ \ \isakeyword{assumes}\isanewline
\ \ \ \ {\isadigit{1}}{\isacharcolon}{\kern0pt}\ {\isacartoucheopen}pushout{\isacharunderscore}{\kern0pt}diagram\ A\ B\ C\ D\ f\ g\ g{\isacharprime}{\kern0pt}\ f{\isacharprime}{\kern0pt}{\isacartoucheclose}\ \isakeyword{and}\isanewline
\ \ \ \ {\isadigit{2}}{\isacharcolon}{\kern0pt}\ {\isacartoucheopen}pushout{\isacharunderscore}{\kern0pt}diagram\ B\ E\ D\ F\ e\ g{\isacharprime}{\kern0pt}\ e{\isacharprime}{\kern0pt}{\isacharprime}{\kern0pt}\ e{\isacharprime}{\kern0pt}{\isacartoucheclose}\ \ \ \ \ \ \ \ \ \ \ \ \ \ \ \ \ \ \ \isanewline
\ \ \isakeyword{shows}\ {\isacartoucheopen}pushout{\isacharunderscore}{\kern0pt}diagram\ A\ E\ C\ F\ {\isacharparenleft}{\kern0pt}e\ {\isasymcirc}\isactrlsub {\isasymrightarrow}\ f{\isacharparenright}{\kern0pt}\ g\ \ e{\isacharprime}{\kern0pt}{\isacharprime}{\kern0pt}\ {\isacharparenleft}{\kern0pt}e{\isacharprime}{\kern0pt}\ {\isasymcirc}\isactrlsub {\isasymrightarrow}\ f{\isacharprime}{\kern0pt}{\isacharparenright}{\kern0pt}{\isacartoucheclose}%
}
\DefineSnippet{pushout-decomposition}{
\isacommand{lemma}\isamarkupfalse%
\ pushout{\isacharunderscore}{\kern0pt}decomposition{\isacharcolon}{\kern0pt}\isanewline
\ \ \isakeyword{assumes}\isanewline
\ \ \ \ e\ {\isacharcolon}{\kern0pt}\ {\isacartoucheopen}morphism\ B\ E\ e{\isacartoucheclose}\ \isakeyword{and}\isanewline
\ \ \ \ e{\isacharprime}{\kern0pt}{\isacharcolon}{\kern0pt}\ {\isacartoucheopen}morphism\ D\ F\ e{\isacharprime}{\kern0pt}{\isacartoucheclose}\ \isakeyword{and}\isanewline
\ \ \ \ {\isadigit{1}}\ {\isacharcolon}{\kern0pt}\ {\isacartoucheopen}pushout{\isacharunderscore}{\kern0pt}diagram\ A\ B\ C\ D\ f\ g\ g{\isacharprime}{\kern0pt}\ f{\isacharprime}{\kern0pt}{\isacartoucheclose}\ \isakeyword{and}\isanewline
\ \ \ \ {\isachardoublequoteopen}{\isadigit{1}}{\isacharplus}{\kern0pt}{\isadigit{2}}{\isachardoublequoteclose}{\isacharcolon}{\kern0pt}\ {\isacartoucheopen}pushout{\isacharunderscore}{\kern0pt}diagram\ A\ E\ C\ F\ {\isacharparenleft}{\kern0pt}e\ {\isasymcirc}\isactrlsub {\isasymrightarrow}\ f{\isacharparenright}{\kern0pt}\ g\ \ e{\isacharprime}{\kern0pt}{\isacharprime}{\kern0pt}\ {\isacharparenleft}{\kern0pt}e{\isacharprime}{\kern0pt}\ {\isasymcirc}\isactrlsub {\isasymrightarrow}\ f{\isacharprime}{\kern0pt}{\isacharparenright}{\kern0pt}{\isacartoucheclose}\ \isakeyword{and}\isanewline
\ \ \ \ {\isachardoublequoteopen}{\isadigit{2}}cv{\isachardoublequoteclose}{\isacharcolon}{\kern0pt}\ {\isacartoucheopen}{\isasymAnd}v{\isachardot}{\kern0pt}\ v\ {\isasymin}\ V\isactrlbsub B\isactrlesub \ {\isasymLongrightarrow}\ \isactrlbsub e{\isacharprime}{\kern0pt}{\isacharprime}{\kern0pt}\ {\isasymcirc}\isactrlsub {\isasymrightarrow}\ e\isactrlesub \isactrlsub V\ v\ {\isacharequal}{\kern0pt}\ \isactrlbsub e{\isacharprime}{\kern0pt}\ {\isasymcirc}\isactrlsub {\isasymrightarrow}\ g{\isacharprime}{\kern0pt}\isactrlesub \isactrlsub V\ v{\isacartoucheclose}\ \isakeyword{and}\isanewline
\ \ \ \ {\isachardoublequoteopen}{\isadigit{2}}ce{\isachardoublequoteclose}{\isacharcolon}{\kern0pt}\ {\isacartoucheopen}{\isasymAnd}ea{\isachardot}{\kern0pt}\ ea\ {\isasymin}\ E\isactrlbsub B\isactrlesub \ {\isasymLongrightarrow}\ \isactrlbsub e{\isacharprime}{\kern0pt}{\isacharprime}{\kern0pt}\ {\isasymcirc}\isactrlsub {\isasymrightarrow}\ e\isactrlesub \isactrlsub E\ ea\ {\isacharequal}{\kern0pt}\ \isactrlbsub e{\isacharprime}{\kern0pt}\ {\isasymcirc}\isactrlsub {\isasymrightarrow}\ g{\isacharprime}{\kern0pt}\isactrlesub \isactrlsub E\ ea{\isacartoucheclose}\isanewline
\ \ \isakeyword{shows}\ {\isacartoucheopen}pushout{\isacharunderscore}{\kern0pt}diagram\ B\ E\ D\ F\ e\ g{\isacharprime}{\kern0pt}\ e{\isacharprime}{\kern0pt}{\isacharprime}{\kern0pt}\ e{\isacharprime}{\kern0pt}{\isacartoucheclose}%
}

\DefineSnippet{graph-fmap-fset}{
\isacommand{record}\isamarkupfalse%
\ {\isacharparenleft}{\kern0pt}{\isacharprime}{\kern0pt}v{\isacharcomma}{\kern0pt}{\isacharprime}{\kern0pt}e{\isacharcomma}{\kern0pt}{\isacharprime}{\kern0pt}l{\isacharcomma}{\kern0pt}{\isacharprime}{\kern0pt}m{\isacharparenright}{\kern0pt}\ pre{\isacharunderscore}{\kern0pt}graph\ {\isacharequal}{\kern0pt}\isanewline
\ \ nodes\ \ {\isacharcolon}{\kern0pt}{\isacharcolon}{\kern0pt}\ {\isachardoublequoteopen}{\isacharprime}{\kern0pt}v\ fset{\isachardoublequoteclose}\isanewline
\ \ edges\ \ {\isacharcolon}{\kern0pt}{\isacharcolon}{\kern0pt}\ {\isachardoublequoteopen}{\isacharprime}{\kern0pt}e\ fset{\isachardoublequoteclose}\isanewline
\ \ source\ {\isacharcolon}{\kern0pt}{\isacharcolon}{\kern0pt}\ {\isachardoublequoteopen}{\isacharparenleft}{\kern0pt}{\isacharprime}{\kern0pt}e{\isacharcomma}{\kern0pt}{\isacharprime}{\kern0pt}v{\isacharparenright}{\kern0pt}\ fmap{\isachardoublequoteclose}\isanewline
\ \ target\ {\isacharcolon}{\kern0pt}{\isacharcolon}{\kern0pt}\ {\isachardoublequoteopen}{\isacharparenleft}{\kern0pt}{\isacharprime}{\kern0pt}e{\isacharcomma}{\kern0pt}{\isacharprime}{\kern0pt}v{\isacharparenright}{\kern0pt}\ fmap{\isachardoublequoteclose}\isanewline
\ \ node{\isacharunderscore}{\kern0pt}label\ {\isacharcolon}{\kern0pt}{\isacharcolon}{\kern0pt}\ {\isachardoublequoteopen}{\isacharparenleft}{\kern0pt}{\isacharprime}{\kern0pt}v{\isacharcomma}{\kern0pt}{\isacharprime}{\kern0pt}l{\isacharparenright}{\kern0pt}\ fmap{\isachardoublequoteclose}\isanewline
\ \ edge{\isacharunderscore}{\kern0pt}label\ {\isacharcolon}{\kern0pt}{\isacharcolon}{\kern0pt}\ {\isachardoublequoteopen}{\isacharparenleft}{\kern0pt}{\isacharprime}{\kern0pt}e{\isacharcomma}{\kern0pt}{\isacharprime}{\kern0pt}m{\isacharparenright}{\kern0pt}\ fmap{\isachardoublequoteclose}%
}

\begin{document}

\title[Formalising the Double-Pushout Approach to Graph Transformation]{Formalising the Double-Pushout Approach to\\ Graph Transformation}

\author[R.~S\"oldner]{Robert S\"oldner\lmcsorcid{0000-0001-5295-3493}}
\author[D.~Plump]{Detlef Plump\lmcsorcid{0000-0002-1148-822X}}

\address{University of York, UK}	
\email{rs2040@york.ac.uk, detlef.plump@york.ac.uk}  




\begin{abstract}
  \noindent In this paper, we utilize Isabelle/HOL to develop a formal framework for the basic theory of double-pushout graph transformation. Our work includes defining essential concepts like graphs, morphisms, pushouts, and pullbacks, and demonstrating their properties. We establish the uniqueness of derivations, drawing upon Rosen’s 1975 research, and verify the Church-Rosser theorem using Ehrig’s and Kreowski’s 1976 proof, thereby demonstrating the effectiveness of our formalisation approach. The paper details our methodology in employing Isabelle/HOL, including key design decisions that shaped the current iteration. We explore the technical complexities involved in applying higher-order logic, aiming to give readers an insightful perspective into the engaging aspects of working with an Interactive Theorem Prover. This work emphasizes the increasing importance of formal verification tools in clarifying complex mathematical concepts.
\end{abstract}

\maketitle

\section*{Introduction}\label{S:one}
Formal methods are instrumental in minimizing software defects by rigorously verifying software correctness against their specification. These methods employ mathematical techniques to assure that software implementations align with predefined specifications and requirements, and include model checking, theorem proving, and symbolic interpretation~\cite{Huth-Ryan04a}.

Interactive theorem provers (ITPs), have validated significant mathematical theorems such as the Four Colour Theorem~\cite{Gonthier07a}, the Prime Number Theorem~\cite{Avigad-Donnelly-Gray-Raff07a}, and the Kepler Conjecture~\cite{Hales15a}. They have also been instrumental in verifying specific algorithms and software components, like in the cases of the seL4 Microkernel~\cite{Klein-etal09a} and the CompCert compiler~\cite{Leroy09a}. These instances showcase the efficacy of ITPs in handling complex theories and systems.

Our research strives to rigorously prove fundamental results in the double-pushout approach to graph transformations~\cite{Ehrig-Ehrig-Prange-Taentzer06a}, contributing to our overarching goal of verifying specific programs in the GP\,2 graph programming language~\cite{Plump16b, Campbell-Courtehoute-Plump21a} using the Isabelle proof assistant~\cite{Nipkow2014}. 

This paper not only synthesizes previously published material~\cite{Soeldner-Plump23a} but also introduces new essential technical details, including:
\begin{itemize}
    \item The formalisation of graphs, morphisms, and rules, including a discussion of our design decisions.
    \item Additional details on our gluing construction.
    \item We added additional proofs, such as the fact that pushouts preserve surjectivity.
    \item We adopted Theorem~\ref{thm:church-rosser} to cover the commutativity of sequentially independent direct derivations.
\end{itemize}

While earlier efforts like Strecker's~\cite{Strecker18a} have applied ITPs to graph transformation, they have not fully explored the extensive theoretical results available for the double-pushout approach. To the best of our knowledge, our work is the first to formalise the foundational results in double-pushout graph transformation. Further details on related work can be found in Section~\ref{sec:related-work}.

In the double-pushout approach, our aim is to abstract from specific node and edge identifiers. To achieve this, we introduce independent type variables representing the types of nodes and edges for each graph. This differentiation enables the use of Isabelle's typechecker during development to prevent unintended mix-ups of node or edge identifiers between graphs. However, we encounter a challenge with Isabelle's inability to quantify over new type variables within locale definitions, which complicates the formalisation of the universal properties of pushouts and pullbacks. 

To mitigate this issue, we adopt a strategy in our locales where we define node and edge identifiers as natural numbers. This decision, among others such as potentially using a unified type for both node and edge identifiers, necessitates alternate constructions and the consideration of constraints like ensuring a sufficiently large universe of identifiers to execute the disjoint union. We will discuss these critical design choices in the technical Section~\ref{sec:DPO}.

To establish the uniqueness of direct derivations, our proof strategy is divided into two distinct stages. Initially, we focus on demonstrating the uniqueness of the pushout complement. Following this, we prove the uniqueness of the pushout object itself, assuming an isomorphic pushout complement. In the first stage, our approach is inspired by Lack and Sobocinski's methodology~\cite{Lack-Sobocinski04a}. However, we adapt their method by incorporating the pushout characterisation and the reduced chain condition~\cite{EK79}, along with utilising the composition and decomposition lemmas for pushouts and pullbacks.

For the proof of the Church-Rosser theorem, we base our methodology on Ehrig and Kreowski's original work~\cite{EK76}. Here too, the pushout characterisation and the reduced chain condition is crucial. We are further employing the composition and decomposition lemmas for pushouts and pullbacks. Additionally, we introduce a lemma that facilitates transitioning between pullbacks and pushouts.

While we recognize the potential to generalise our proofs from graph constructions to adhesive categories, this is not our primary objective.
Our decision stems from two considerations: firstly, to avoid the complexity involved in handling an abstract class of categories, such as van Kampen squares; and secondly, because our research involves both abstract concepts like pushouts and pullbacks, as well as set-theoretic constructions for these concepts. Providing the corresponding constructions for all adhesive categories would be impractical.

Our long-term goal is to establish a foundation for verifying graph programs in GP\,2, a language fundamentally based on the double-pushout approach to graph transformation. We aspire to provide both interactive and automatic tool support for formal reasoning in graph transformation languages. An effective proof assistant for GP\,2 will necessitate concrete definitions of graphs, attributes, rules, derivations, and more, which guides our focus towards graph transformation concepts.


The remainder of this paper is organized as follows:
\begin{itemize}
    \item Section~\ref{sec:hol} offers a concise introduction to the Isabelle Proof Assistant, focusing on the constructs employed in our research.
    \item In Section~\ref{sec:DPO}, we delve into the fundamentals of DPO graph transformation and our corresponding formalisation in Isabelle. This section covers the formalisation of key concepts such as graphs, morphisms, and rules.
    \item Section~\ref{sec:uniqueness-direct-derivation} discusses the uniqueness of direct derivations, including a detailed proof of the uniqueness of pushout complements.
    \item The Church-Rosser theorem, which posits that parallel independent direct derivations can be rearranged to conclude in a common graph, is explored in Section~\ref{sec:church-rosser}.
    \item Section~\ref{sec:related-work} presents a brief review of related work in the field, providing context and background to our study.
    \item Finally, Section~\ref{sec:discussion} summarises our key findings and outlines potential avenues for future research.
\end{itemize}
All formalisation artefacts, including the complete Isabelle theories, are available on GitHub at \url{https://github.com/UoYCS-plasma/DPO-Formalisation}.

\section{Isabelle/HOL}\label{sec:hol}
Isabelle is a versatile, interactive theorem prover that operates on the principles of the LCF (Logic for Computable Functions) approach. Central to its design is a compact meta-logical proof kernel responsible for proof checking, a feature that significantly bolsters confidence in the prover's soundness. When referring to Isabelle/HOL, we are discussing the higher-order logic instantiation within Isabelle, widely recognized as the most mature calculus in its suite~\cite{Paulson2019}. This instantiation is characterised by a strong typing system that supports polymorphic, higher-order functions~\cite{Brucker2018}.

In Isabelle/HOL, type variables are distinctly marked by a preceding apostrophe. For instance, a term \texttt{f} of type \texttt{'a} is represented as \texttt{f :: 'a}. Our formalisation leverages \texttt{locales}, a sophisticated mechanism designed for authoring and structuring parametric specifications. A locale encapsulates a set of parameters (\(x_1 \dots x_n\)), assumptions (\(A_1 \dots A_m\)), and a resulting theorem, expressed as \(\bigwedge x_1 \dots x_n. \llbracket A_1; \dots ; A_m \rrbracket \Longrightarrow C\). This approach facilitates the effective combination and enhancement of contexts, yielding a representation that is both clear and maintainable. For further details, Ballarin~\cite{Ballarin} provides an extensive introduction.

Additionally, we employ \emph{intelligible semi-automated reasoning} (Isar), Isabelle's framework for writing structured proofs~\cite{Wenzel1999}. Unlike 'apply-scripts', which linearly execute deduction rules, Isar adopts a structured, organized approach. This methodology significantly improves the readability and maintainability of proofs~\cite{Nipkow2014}. A comprehensive introduction to Isabelle/HOL is available in~\cite{Nipkow2014}.

\section{DPO Graph Transformation in Isabelle}
\label{sec:DPO}
Our formalisation of DPO based graph transformations within the Isabelle/HOL proof assistant is centred around \emph{(finite) directed labelled graphs}.

\subsection{Graphs}
These graphs are defined over an abstract \emph{label alphabet} comprising distinct sets of node and edge labels. We intentionally left labels abstract, as it ensures that the properties we discuss are universally applicable across specific label alphabets. Our definition is inclusive, we accommodate parallel edges and allow loops.

\begin{defi}[Graph]\label{def:graph}
  A \emph{graph} \(G = (V, E, s, t, l, m)\) over the alphabet \(\mathcal{L}\) is a system where
  \(V\) is the finite set of nodes, \(E\) is the finite set of edges,
  \(s,t \colon E \to V\) are functions assigning the source and target to each edge,
  \(l \colon V \to \mathcal{L}_V\) and \(m \colon E \to \mathcal{L}_E\) are functions assigning
  a label to each node and edge.\qed
\end{defi}

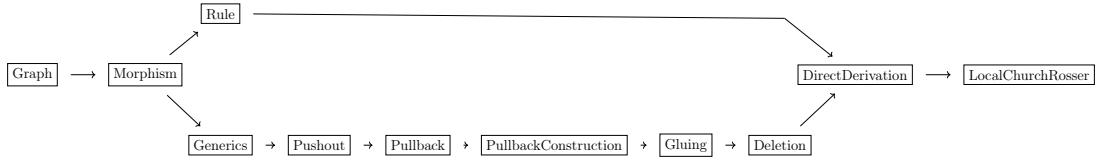
\begin{figure}
    \centering

    \scalebox{0.5}{\begin{tikzpicture}
	\begin{pgfonlayer}{nodelayer}
		\node [style=new style 0] (0) at (-6, 0) {Graph};
		\node [style=new style 0] (1) at (0, 0) {Morphism};
		\node [style=new style 0] (2) at (4, 3.25) {Rule};
		\node [style=new style 0] (3) at (4, -3.75) {Generics};
		\node [style=new style 0] (4) at (9.25, -3.75) {Pushout};
		\node [style=new style 0] (5) at (14.5, -3.75) {Pullback};
		\node [style=new style 0] (7) at (28.75, -3.75) {Gluing};
		\node [style=new style 0] (8) at (33.75, -3.75) {Deletion};
		\node [style=new style 0] (9) at (37.75, 0) {DirectDerivation};
		\node [style=new style 0] (10) at (47, 0) {LocalChurchRosser};
		\node [style=new style 0] (11) at (21.75, -3.75) {PullbackConstruction};
		\node [style=none] (12) at (34, 3) {};
	\end{pgfonlayer}
	\begin{pgfonlayer}{edgelayer}
		\draw [style=right arrow] (9) to (10);
		\draw [style=right arrow] (8) to (9);
		\draw [style=right arrow] (7) to (8);
		\draw [style=right arrow] (11) to (7);
		\draw [style=right arrow] (5) to (11);
		\draw [style=right arrow] (4) to (5);
		\draw [style=right arrow] (3) to (4);
		\draw [style=right arrow] (0) to (1);
		\draw [style=right arrow] (1) to (3);
		\draw [style=right arrow] (1) to (2);
		\draw [style=new edge style 0] (2) to (12.center);
		\draw [style=new edge style 1] (12.center) to (9);
	\end{pgfonlayer}
\end{tikzpicture}}
    \caption{\texttt{locale} hierarchy reflecting our formalisation}
    \label{fig:locales}
\end{figure}
The formalisation process within a Proof Assistant mandates adherence to a specific formalism; in our context, this is classical high-order logic. 
The structure of this work involves a range of design choices. Primarily, we adopt Noschinskis~\cite{Noschinski2015} approach, utilizing record types to achieve a concise representation of components, and employing locales to systematically impose properties on these components. Fig.~\ref{fig:locales} depicts our developed formalisation from the perspective of the locale hierarchy. In this figure, the arrow signifies \emph{used by}. For instance, the \texttt{Graph} locale serves as a foundation for the \texttt{Morphism} locale. Similarly, \texttt{Rule} depends on \texttt{Morphism}, and \texttt{DirectDerivation} relies on \texttt{Rule}.

While portions of our work have been described in preceding work~\cite{Soeldner-Plump22a}, this manuscript will refine and update our methodology.

As initially mentioned in the introduction Section, Isabelle/HOL exclusively deals with total functions. This means that these functions are defined across the entire universe of the underlying type. 
The built-in \texttt{option} data type is a standard way to represent optional values, i.e., the presence or absence of data. This sum type can hold two types of values, represented by their constructor:
\begin{itemize}
    \item \texttt{None} used to represent the absence of a value, and
    \item \texttt{Some a} captures the presence of a value, where \texttt{a} is the actual value held.
\end{itemize}
By using the \texttt{option} data type, we can simulate partial functions as a function \texttt{'a $\Rightarrow$ 'b option}, which is commonly
known as the built-in synonym \texttt{map} with the infix \texttt{$\rightharpoonup$} notation.
Another extension found in the standard library are finite sets and finite maps, which are called \texttt{fset} and \texttt{fmap} respectively. These
data types are built on top of the standard (potentially infinite) sets and maps by using the \texttt{typedef} keyword, which instructs Isabelle/HOL to provide a new datatype with finitely many elements and the additional finite property.
\\

\noindent We consider the following different design options for our formalisation:
\begin{enumerate}
    \item The usage of finite maps and sets (\texttt{fmap} and \texttt{fset}), and
    \item rely on total functions together with the \texttt{finite} predicate.
\end{enumerate}
Although from an engineering perspective, the first option to encode properties in the type system seems appealing, it is not without challenges.
Here, we can represent graphs using a \texttt{pre\_graph} record by using the built-in \texttt{fmap} and \texttt{fset} as follows:
\Snippet{graph-fmap-fset}
The implementation of the \texttt{locale}, which enforces the corresponding properties of the underlying object, can be expressed using higher abstractions in the context of \texttt{fmap} and \texttt{fset}. In this scenario, we utilise the built-in native \texttt{fmdom} and \texttt{fmran} functions to yield a \texttt{fset} of the range and domain, mirroring the defined area. This is possible as the finite map relies on the \texttt{option} datatype. The integrity of the source function is expressed using \texttt{fmdom s$_G$ = E$_G$ $\wedge$ fmran s$_G$ |$\subseteq$| V$_G$}.
Conversely, for the total function scenario \((2)\), it becomes necessary to explicitly quantify over the set of nodes and edges to convey an equivalent statement. We write \texttt{$\forall$ v $\in$ E$_G$. s$_G$ e $\in$ V$_G$} to express the same property.

Drawing from our experience and evaluations, we chose to further pursue the (2) approach to enhance our automation endeavours.
This decision primarily stems from the underdevelopment of \texttt{fset}, \texttt{map}, and \texttt{fmap} theories. 
We found ourselves in the position of having to establish numerous fundamental properties about \texttt{fmap} and \texttt{fset} that were not previously proven. Although the Isabelle distributions \emph{lifting and transfer} package~\cite{huffman2013lifting} offered an infrastructure for applying proven properties from underlying theories, such as the \texttt{set} and \texttt{map} theory, the overall process represented a significant effort.
This aspect is further emphasized by Díaz's submissions~\cite{Finite-Map-Extras-AFP} to the Archive of Formal Proofs (AFP), which enhance the \texttt{fmap} theory with new constructs and proven properties.

The corresponding \texttt{pre\_graph} record is defined by the subsequent Isabelle code.
\Snippet{pre-graph} 

We present the final locale enforcing the graph properties in the following code block. Notably, there is no need for explicit statements about the properties of node and edge labels, since total functions are defined for the entire universe.

\Snippet{graph-locale}

The usage of the built-in \texttt{countable} typeclass for node and edge identifiers in our work addresses the Isabelle/HOL constraint against introducing new type variables within a locales definition. This is a result of the simple type theory used by the HOL family~\cite{harrison-wiedijk12a}. This method parallels the approach found in~\cite{Schirmer-Wenzel09a}, where Isabelle's \texttt{statespace} command was used to create an environment with axioms for converting concrete types to natural numbers. Our approach differs by utilizing the injective \texttt{to\_nat} function provided by the \texttt{countable} typeclass, thus enabling efficient type conversion. This specific limitation is further detailed in \cite{Soeldner-Plump22a}. The injectivity also implies the existence of the inverse, \texttt{from\_nat}. We use this technique to translate arbitrary identifiers for nodes and edges into a generic representation based on natural numbers.
\begin{defi}[Graphs over naturals]
  A graph whose node and edge identifiers are natural numbers is called a \emph{natural graph}.
\end{defi}
In Isabelle/HOL, we establish a type synonym \texttt{ngraph}, which specializes the existing \texttt{pre\_graph} structure to \texttt{nat}, Isabelle’s native type for natural numbers, for both identifiers.
\Snippet{ngraph}
The use of Isabelle/HOL's built-in functions \texttt{to\_nat} and \texttt{from\_nat} allows us to convert
between both representations. We define the conversion function \texttt{to\_ngraph}
from \texttt{ngraph} to the parameterised \texttt{pre\_graph} structure as follows:
\Snippet{to-ngraph}
We convert the set of node and edge identifiers by applying the \texttt{to\_nat} function, which the \texttt{countable} typeclass provides, to each element in these sets using the \texttt{image} function, denoted by the \texttt{`}. The source and target functions initially map from \texttt{nat} to the original identifier, then apply the source and target functions, and subsequently map back into the \texttt{nat} space. For both labelling functions, we simply convert back the \texttt{nat} identifiers before applying the corresponding labelling function.

The reverse process, converting a graph over naturals to a parameterized graph, functions similarly. Its important to note that \texttt{from\_nat (to\_nat x) = x} holds true, but the converse does not apply.
The crucial fact to remember is that the definition of \texttt{to\_nat} depends on Hilbert's choice operator \texttt{SOME}, which selects a fixed yet arbitrary value. This reliance is a key aspect to consider, as it fundamentally influences the behaviour and limitations of the \texttt{to\_nat} function. By utilizing \texttt{SOME}, \texttt{to\_nat} essentially encapsulates a degree of arbitrariness, making its outcomes and their relation to the inverse process non-trivial and significant in our context. Consequently, we often need to include type annotations to avoid the inference of different types, which would result in the selection of a different value. This becomes particularly crucial in scenarios involving nested function applications. In these cases, while the final type aligns correctly, the intermediate type might not, leading to potential discrepancies. By actively managing these type annotations, we ensure consistency and accuracy in our function applications, especially when dealing with complex nested structures where type alignment is critical for the intended outcomes.
\Snippet{from-ngraph}

\subsection{Morphisms}
Following this, we will define mappings between graphs that preserve their structure, known as graph morphisms. 
\begin{defi}[Graph morphism]\label{def:morphism}
\normalfont
A \emph{graph morphism}\/ $f \colon G \to H$ is a pair of mappings $f = (f_V \colon V_G \to V_H,\, f_E \colon E_G \to E_H)$, such that for all $e \in E_G$ and $v \in V_G$:
\begin{enumerate}
    \item \label{def:morph-src-presv}$f_V(s_G(e)) = s_H(f_E(e))$ (sources are preserved)
    \item \label{def:morph-trg-presv}$f_V(t_G(e)) = t_H(f_E(e))$ (targets are preserved)
    \item \label{def:morph-l-presv}$l_G(v) = l_H(f_V(v))$ (node labels are preserved)
    \item \label{def:morph-m-presv}$m_G(e) = m_H(f_E(e))$ (edge labels are preserved)\qed
\end{enumerate}
\end{defi}

We adopt the same design pattern for graph morphisms as we did for graphs: First, we define a record type to bundle the relevant components, which here include the node and edge mappings. Then, we establish a locale to enforce the properties of the morphism. The morphism record is defined as follows.
\Snippet{pre-morph}

It is important to emphasize that a morphism can map between two different types for each node and edge identifier. 
This aspect is vital for our subsequent work on the pushout and pullback construction, which we will discuss later in this section.

The morphism locale utilizes the hierarchical inheritance feature of locales, enabling us to depend on the graph locale for defining the source and target graphs. We enhance its functionality by fixating the \texttt{pre\_morph} record, accomplished using the \texttt{fixes} keyword. This is followed by stating the locale axioms, which assert that graphs represent a structure-preserving mapping. This design choice not only streamlines the locale's structure but also ensures a logical integration of graph properties within the morphism context.
\Snippet{morph-locale}

In this context, the graph \texttt{G} serves as the source, while \texttt{H} represents the target graph. The record encapsulating the node and edge mappings between \texttt{G} and \texttt{H} is assigned the name \texttt{f}. The two assumptions, namely \texttt{morph\_edge\_range} and \texttt{morph\_node\_range}, ensure that for every edge \texttt{e} in the source graph, \texttt{f$_E$} maps to a corresponding edge in the target graph, and similarly for nodes.
The assumptions \texttt{source\_preserve} and \texttt{target\_preserve} are critical for maintaining the structural integrity of the graph; they guarantee that the source and target structures are consistently upheld in the mapping process.
Lastly, the assumptions \texttt{label\_preserve} and \texttt{mark\_preserve} are instrumental in ensuring that the node and edge labels are retained during the mapping, thereby preserving the complete informational content of the graphs.

\begin{defi}[Special morphisms and isomorphic graphs]
\normalfont
A morphism $f$ is \emph{injective}\/ (\emph{surjective}, \emph{bijective}) if $f_V$ and $f_E$ are injective (surjective, bijective). Morphism $f$ is an \emph{inclusion} if for all $v \in V_G$ and $e \in V_E$, $f_V(v) = v$ and $f_E(e) = e$.
A bijective morphism is an \emph{isomorphism}. In this case, $G$ and $H$ are \emph{isomorphic}, which is denoted by $G \cong H$.\qed
\end{defi}

Our characterization also takes advantage of the locale inheritance mechanism. An injective morphism inherits all properties from the morphism locale and additionally asserts that both node and edge mappings are injective over the respective sets of nodes and edges of the source graph.
We achieve this by employing the built-in \texttt{inj\_on} predicate. This predicate plays a crucial role in ensuring the uniqueness of each element's mapping in the source graph to the target graph, thereby preserving the distinctiveness of graph elements in the transformation process.

\Snippet{inj-morph}

It is important to note that the \texttt{f$_V$} (\texttt{f$_E$}) mappings and graph G, with its corresponding components (\texttt{V$_G$} and \texttt{E$_G$}) which are utilized in the body of the \texttt{injective\_morphism} locale, are also supplied by the \texttt{morphism} locale.

A surjective morphism also adheres to this pattern. In this case, we declare that for each node (or edge) in the target graph \texttt{H}, there exists a corresponding node (or edge) in the source graph \texttt{G} that is mapped to it through the respective morphism function (\texttt{f$_V$} and \texttt{f$_E$}). This condition ensures that every element in the target graph has a pre-image in the source graph, signifying that the mapping covers the entire target graph.

\Snippet{surj-morph}

A bijective morphism combines the principles of both injective and surjective morphisms. 
In this framework, we assert that for every node (or edge) in the source graph \texttt{H}, there is a unique corresponding node (or edge) in the target graph \texttt{H}. We choose not to inherit from both the \texttt{injective\_morphism} and \texttt{surjective\_morphism} locales. Instead, we utilize the built-in predicate \texttt{bij\_betw}. A key reason for this decision is to leverage automation: in our experiments, the use of the existing set of proven lemmas significantly facilitated the process of discharging theorems. While it was feasible to demonstrate all these properties if we inherit these locales, we opted to capitalize on the existing infrastructure.

\Snippet{bij-morph}

Moving forward, we introduce the concept of composing two morphisms to form a singular, joint morphism. 

\begin{defi}[Morphism composition]\label{def:morphcomp}
Let $f \colon F \rightarrow G$ and $g \colon G \rightarrow H$ be graph morphisms. The \emph{morphism composition}\/ $g \circ f \colon F \rightarrow H$ is defined by $g \circ f = (g_V \circ f_V, g_E \circ f_E)$.\qed
\end{defi}

In Isabelle/HOL, we define this composition with a definition that typically hides the underlying details and often requires explicit unfolding or an established set of proven properties.
Choosing total functions enables us to use standard function composition techniques for both the node and edge mappings.

\Snippet{morph-comp}

Please note that each morphism can alter the typing of the node and edge identifiers. Therefore, in our definition, we introduced three typing parameters. Additionally, we have introduced the infix syntax \texttt{$\circ_\rightarrow$} for graph morphism composition as a form of syntactic sugar.
\begin{lem}[Well-definedness of morphism composition]\label{lem:morph-comp}
Given two morphisms \(f \colon G \to H\) and \(g \colon H \to K\), the composition \(g \circ f \colon G \to K\) is also a morphism.
\end{lem}
If we have two valid morphisms, say \( f: G \to H \) and \( g: H \to K \), then their composition \( g \circ f: H \to K \) is also a valid morphism. In Isabelle, we capture this property using the \texttt{lemma} keyword along with a specific name (\texttt{wf\_morph\_comp}) for future reference. 
\Snippet{wf-morph-comp}

We start by assuming the existence of two valid morphisms, which we name \( f \) and \( g \), within the \texttt{assumes} blocks. Then, by using the \texttt{shows} keyword, we set our target for proof in Isabelle.
Following this setup, Isabelle requires the user to interactively work towards achieving the stated goal, essentially proving the validity of the morphism composition.

We open the proof using the \texttt{intro\_locales} tactic, which will apply the corresponding introduction rule, resulting in the following proof state.
\begin{verbatim}
proof (state)
goal (3 subgoals):
 1. graph G
 2. graph K
 3. morphism_axioms G K (g ◦→ f)
\end{verbatim}

The output shows that we must discharge three subgoals. The first two, which states that \( G \) and \( K \) are valid graphs, are trivial due to the morphism properties. Since the morphism locale inherits the graph properties for both the source and target graphs, these subgoals follow directly. 

Using Isabelle's \emph{Isar} language, we streamline our proof writing process by focusing on the logical structure and progression of arguments. Isar, emphasizing the \emph{what} over the \emph{how}, allows us to structure complex proofs into understandable segments. We initiate each proof by specifying our goal with the \texttt{show} keyword, clearly setting our target. This approach not only simplifies the proof of our statement, which in this case follows from the morphism axioms of \( f \) and \( g \), but also keeps the proof readable and aligned with traditional mathematical reasoning, enhancing its clarity and verifiability.

To tackle the third subgoal, we focus on the \texttt{morphism\_axioms} predicate, a product of the locales infrastructure based on the properties we stated. To discharge this subgoal, we initiate the proof using the \texttt{standard} tactic, which will execute a single elimination or introduction rule according to the topmost logical connective involved~\cite{isabelle-isar-ref}.

Consequently, our task involves discharging the six morphism axioms stated earlier, which we will briefly outline below. The initial two goals address the fact that for each edge (and node) in \texttt{G}, the image created by the composition of the morphisms falls within the subset of the edges (nodes) of \texttt{K}.
\Snippet{wf-morph-comp-proof-2}
Isabelle's simplifier is able to successfully discharges all goals, provided it is equipped with the necessary lemmas. Specifically, for our case, we need to unfold the \texttt{morph\_comp} definition and, to discharge the first goal, we need to employ the \texttt{morph\_edge\_range} fact for each of the graphs \texttt{G} and \texttt{K}.

Similarly, the property asserting that the composed morphism preserves the structure under the source function is discharged. This property is defined in the morphism locale (\texttt{source\_preserve}). We address it using a method similar to the previous approach, employing relevant tactics and lemmas (in a forward directed style using the \texttt{OF} keyword).
\Snippet{wf-morph-comp-proof-3}
Finally, the preservation of node and edge labels by the morphism follows a similar pattern and is subsequently shown for completeness.
\Snippet{wf-morph-comp-proof-4}

The purpose of the earlier discussion was to offer readers an intuitive understanding of how properties are articulated and to provide a glimpse into the overall setup in Isabelle. Moving forward, the focus will shift to describing the fundamental constructs of our formalisation, highlighting selected content that is crucial to understanding the framework and methodology employed.

Please note, if a morphism is uniquely identified by its source and target, we sometimes omit
the name and write \(F \to G \to H\) to stand for the composition \(g \circ f\).
In our formalisation, we denote morphism composition using the \(\circ_{\rightarrow}\) symbol to prevent a naming clash with Isabelle's built-in function composition.

\subsection{Rules}
In DPO-based graph transformation, graph morphisms are used to define rules as the atomic units of computation.
\begin{defi}[Rule]\label{def:rule}
  A \emph{rule}\/ \((L\leftarrow K \rightarrow R)\) consists of graphs \(L,K\) and \(R\)
  over \(\mathcal{L}\) together with injective morphisms \(K \to L\) and \(K \to R\).\qed
\end{defi}
Following the same pattern as before, we use a \texttt{pre\_rule} record and establish a \texttt{rule} locale to enforce the required properties.

\Snippet{pre-rule}
The polymorphic \texttt{pre\_rule} record already contains eight type parameters, allowing different node and edge identifiers for each graph of the rule (left-hand side, right-hand side, and the interface) while keeping the labelling types consistent. This setup offers significant flexibility by enabling morphisms between different node and edge representations, which is particularly useful in the gluing and deletion construction.
However, this approach comes with a technical drawback: the need to propagate type annotations throughout the codebase. In some cases, this can result in code that is not immediately obvious and may be difficult to read. To address this issue and strike a balance between flexibility and readability, we have reduced the number of locale parameters in our current design by incorporating record types for rules, graphs, morphisms, and other components. This has significantly decreased the number of parameters passed directly into the locales compared to earlier versions of our formalisation.

Despite the challenges, we believe that our current setup offers a sweet spot in terms of conciseness while still maintaining the necessary formal precision. We will explore this trade-off between flexibility and readability in more detail later in the paper, particularly in the context of the gluing and deletion construction.

The locale itself inherits the properties of injective morphisms twice, from the interface to both the left-hand side and the right-hand side.
\Snippet{rule-locale}

We need to populate the \texttt{countable} typeclass restriction to allow conversion of graphs to graphs over naturals.
The main reason for this is a limitation of the locale mechanism, preventing the introduction of new type variables in its definition. Once we covered pushouts (and pullbacks), we will briefly continue this discussion.
\begin{figure}
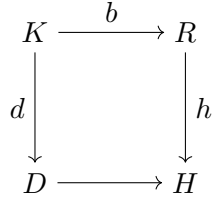

    \ctikzfig{gluing}
    \caption{Gluing diagram}
    \label{fig:gluing}
\end{figure}

The addition of graph components along a common subgraph is called \emph{gluing}. 
\subsection{Gluings}
The gluing construction below employs the disjoint union of sets A and B, typically defined as \[A + B = (A \times \{1\}) \cup (B \times \{2\}).\] It incorporates injective functions \(i_A \colon A \to A + B\) and \(i_B \colon B \to A + B\), ensuring that \(i_A(A) \cup i_B(B) = A+B\) and \(i_A \cap i_B = \emptyset\).

This does not directly translate into a higher-order logic based proof assistant, as the formalism is not expressive enough to capture the change of representation. Here, we would try to mix \texttt{$'a~\text{set}$}, where \texttt{$'a$} is a type variable of the elements, to elements of the type \texttt{$('a,~\text{nat})$}. 

To simplify Lemma~\ref{lemma:gluing}, we assume the sets of \(D\) and \(R - b(K)\) are disjoint. This assumption prevents the contamination of \(i_A\) and \(i_B\) usage throughout the lemma while maintaining respect to the more general statement which does not require the disjointness.

\begin{lem}[Gluing \cite{Ehrig79a}]\label{lemma:gluing}
  Let \(b \colon K \to R\) and \(d \colon K \to D\) be injective graph morphisms.
  Then the following defines a graph \(H\) (see Fig.~\ref{fig:gluing}), the gluing 
  of \(D\) and \(R\) according to \(d\):
  \begin{enumerate}
      \item $V_H = V_D + (V_R - b_V(V_K))$
      \item $E_H = E_D + (E_R - b_E(E_K))$
      \item \(s_H(e) = \begin{cases}
          s_D(e) & \textrm{if } e \in E_D  \\
          d_V(b_V^{-1}(s_R(e))) & \textrm{if } e \in E_R - b_E(E_K) \textrm{ and } s_R(e) \in b_V(V_K)  \\
          s_R(e) & \textrm{otherwise}
          \end{cases}
          \)
      \item $t_H$ analogous to $s_H$
      \item \(l_H = \begin{cases}
          l_D(v) & \textrm{if } v \in V_D \\
          l_R(v) & \textrm{otherwise}
          \end{cases}\)
      \item $m_H$ analogous to $l_H$
  \end{enumerate}
  Moreover, the morphism $D \to H$ is an inclusion and the injective morphism $h$ is defined for 
  all items $x$ in $R$ by $h(x) = \textbf{if } x \in R - b(K) \textbf{ then } x \textbf{ else } d(b^{-1}(x))$.
\end{lem}

Our formalisation of the gluing concept begins with the definition of the corresponding locale, which inherits the two injective morphisms \(b \colon K \to R\) and \(d \colon K \to D\), as follows.
\Snippet{gluing-locale}
Note, the properties of the \texttt{injective\_locale} for \(K \to R\) and \(K \to D\) are bound to \(d\) and \(r\), respectively. These names are chosen arbitrary and are used to refer to facts of the corresponding morphism, while the arguments to the locale correspond to the morphism objects.

Within the locales context, we present the construction described in Lemma~\ref{lemma:gluing}. Initially, we establish abbreviations, which unlike definitions, do not introduce new concepts that have to be unrolled; instead, they simply allow us to refer to their content by an abbreviated name. 
This construction relies on the \texttt{sum} type, featuring two constructors: \texttt{Inl} and \texttt{Inr}. We utilize the built-in function \texttt{Plus}, denoted by the infix notation \texttt{<+>}, for calculating the disjoint sum of sets. This definition relies on applying the \texttt{Inl} constructor to each element of the first set and the \texttt{Inr} constructor to each element of the second. It is important to note that this setup alters the type; we move from the universe, say \texttt{'a} of the first set, and \texttt{'b} of the second, to the sum type \texttt{'a + 'b}.
This shift facilitates the straightforward construction of the disjoint sum, but will result in some complications of our formalisation.
We considered imposing the constraint that the nodes (edges) must be disjoint as an alternative approach to the locale. However, to align with our long-term goal of constructing concrete gluings, we chose not to limit our formalisation and proceeded directly with the construction.

Following the gluing construction (cf. Lemma~\ref{lemma:gluing}), we use \texttt{V} (and \texttt{E}) to represent the elements of the gluing graph.
\Snippet{gluing-sets}
The \texttt{`} notation refers to the built-in \texttt{image} definition, which applies the supplied function to each element of the set. 
The source function now utilizes the two constructors to employ standard pattern matching for identifying the origin graph of an edge. In the first scenario, with \texttt{s (Inl e)}, given the construction of \texttt{V} and \texttt{E}, it becomes evident that this edge remains within graph \texttt{D}. Consequently, we use the source from \texttt{D}, repositioning it into the new type via the application of \texttt{Inl}. This repositioning reflects the application of the previously mentioned injective function \(i_A\), see above.
In the second case, when the source of the edge originates from \texttt{R} without the interface, it is necessary to determine whether the source of the edge projects into the interface. If it does, we must apply the inverse of \texttt{b}, followed by traversal through \texttt{d}. We obtain the element by using the \texttt{inv\_into} built-in function, which uses the Hilbert epsilon operator \texttt{SOME}.
If not, the node remains on the right-hand side, requiring only the application of the \texttt{Inr} constructor.
\Snippet{gluing-source}
The target mapping follows analogous. However, both label mappings are straightforward and only require pattern matching with the respective labelling functions.
\Snippet{gluing-label}
Finally, we have completed defining all the necessary components for the \texttt{pre\_graph} structure of a graph \texttt{H}. We use the \texttt{definition} keyword to define the final \texttt{pre\_graph} structure \texttt{H}.
\Snippet{gluing-H} 
To establish that this record represents a valid graph, we use the \texttt{sublocale} command, which necessitates us to discharge the required graph properties. 
\Snippet{gluing-H-sublocale}
We proceed by defining the graph morphism \(h \colon R \to H\), which is necessary if the node (or edge) resides in \(R\) without \(K\) via \(b\). In the first scenario, it is evident that the node (or edge) is mapped through \texttt{Plus}, using the \texttt{Inr} constructor. In other instances, we utilize the inverse of \(b\) and proceed along \(d\). In all cases, the node (or edge) invariably exists in the left part of \texttt{Plus}, indicated by the \texttt{Inl} constructor.
\Snippet{gluing-morph-h}
The morphism \(c \colon D \to H\) is an injective morphism, which injects the graph \(D\) into \(H\) using the \texttt{Inl} constructor.
\Snippet{gluing-morph-c}
In some cases, it is required to add additional type annotations. If these are skipped, Isabelle tries to infer the corresponding type, which might lead to a surprising outcome. The formalization of the deletion construction is described in~\cite{Soeldner-Plump22a}.

\subsection{Pushouts}
With these components in place, we demonstrate how they define the abstract concept of pushouts.
\begin{defi}[Pushout]\label{def:pushout}
  Given graph morphisms $b \colon A \to B$ and $c \colon A \to C$, a graph $D$ together with
  graph morphisms $f \colon B \to D$ and $g \colon C \to D$ is a \emph{pushout} of
  $A \to B$ and $A \to C$ if the following holds (see Fig.~\ref{fig:pushout}):
  \begin{enumerate}
      \item Commutativity: $f \circ b = g \circ c$, and
      \item Universal property: For all graph morphisms $p \colon B \to D'$ and $t \colon C \to D'$
        such that $p \circ b = t \circ c$, there is a unique morphism $u \colon D \to D'$
        such that $u \circ f = p$ and $u \circ g = t$.
      \end{enumerate}
      We call \(D\) the \emph{pushout object} and \(C\) the \emph{pushout complement}. \qed
\end{defi}

We are not able to express the unique existence, used in the \emph{universal property} directly. The built-in binder for expressing this property presents a challenge. Since total functions are defined over the entire universe and our graphs may only cover a subset of the node (edge) identifiers, the built-in binder would be too strong. Therefore, our quantification needs to focus solely on the nodes (edges) pertinent to the required objects.
As a result, we introduce the abbreviation \texttt{Ex1M}, which will quantify only over the required subset of nodes and edges used in the corresponding graphs.
\Snippet{pushout-ex1m}

We are now prepared to formalize pushouts using a locale, addressing morphisms \(A \to B\), \(A \to C\), \(B \to D\), and \(C \to D\) that fulfil both commutativity and the universal property. The formalisation is given below:
\Snippet{pushout-diagram}

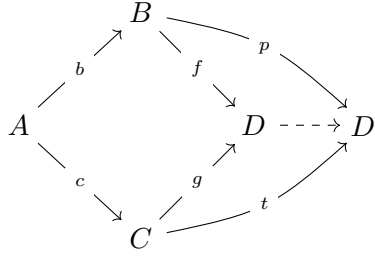
\begin{figure}
    \begin{tikzpicture}
	\begin{pgfonlayer}{nodelayer}
		\node [style=none] (0) at (-3.25, 0) {\(A\)};
		\node [style=none] (1) at (0, 3) {\(B\)};
		\node [style=none] (2) at (0, -3) {\(C\)};
		\node [style=none] (3) at (3, 0) {\(D\)};
		\node [style=none] (4) at (6, 0) {\(D'\)};
	\end{pgfonlayer}
	\begin{pgfonlayer}{edgelayer}
		\draw [style=right arrow] (0.center) to node [fill=white, font={\tiny}] {\(b\)} (1.center);
		\draw [style=right arrow] (0.center) to node [fill=white, font={\tiny}] {\(c\)} (2.center);
		\draw [style=right arrow] (2.center) to node [fill=white, font={\tiny}] {\(g\)} (3.center);
		\draw [style=right arrow] (1.center) to node [fill=white, font={\tiny}] {\(f\)} (3.center);
		\draw [style=right arrow, bend left=15, looseness=1.25] (1.center) to node [fill=white, font={\tiny}] {\(p\)} (4.center);
		\draw [style=right arrow, bend right=15, looseness=1.25] (2.center) to node [fill=white, font={\tiny}] {\(t\)} (4.center);
		\draw [style=right arrow, dashed] (3.center) to (4.center);
	\end{pgfonlayer}
\end{tikzpicture}
    \caption{Pushout \(A \rightarrow B \rightarrow D \leftarrow C \leftarrow A\)}
    \label{fig:pushout}
\end{figure}
The locale incorporates the four morphisms as depicted in Fig,~\ref{fig:pushout}. We express commutativity separately for nodes and edges through \texttt{node\_commutativity} and \texttt{edge\_commutativity}, respectively. Implementing the universal property proves more complex compared to using partial function. When using partial functions, we used equality between the two compositions while without, we explicitly have to quantify over the domain. We address the limitation of Isabelle's locale mechanism, which restricts quantification over fresh types used for the graph \texttt{D'}. 

Our solution involves utilizing the \texttt{to\_ngraph} and \texttt{from\_ngraph} infrastructure to convert our morphism into natural numbers, a setup previously discussed.
The need for this workaround arises from a limitation of simple type theories, which restricts higher-order logic (HOL) systems from allowing explicit quantification over polymorphic type variables~\cite{harrison-wiedijk12a}. Consequently, certain concepts cannot be expressed directly, leading to details that may be hard to read and differ from the presentation in standard mathematical textbooks.

Furthermore, our reliance on total functions means we cannot verify the equality of morphisms directly. Instead, we must quantify over the specific areas they are defined in.

To respond to the limitations of locales, we have introduced an additional lemma within the locale's context. This lemma generalizes the universal property by omitting the \texttt{ngraph} topic. Consequently, we can depend on this lemma for most of our work, eliminating the need to consider the conversion between identifiers in \texttt{graph} and \texttt{ngraph}.

\Snippet{pushout-univ}
The proof relies on the correctness property of the conversion functions.

\begin{lem}[Correctness of \texttt{to\_ngraph}]
    Let \(G\) be a graph, then \texttt{to\_ngraph G} is a natural graph.
\end{lem}
The correctness property is expressed in Isabelle/HOL as an if and only if.
\Snippet{graph_iff_ngraph}
We express a similar property for morphisms and the lifted morphisms over to natural graphs.
\begin{lem}[Correctness of \texttt{to\_nmorph}]
    Let \(m \colon G \to H\) be a graph morphism, then \texttt{to\_nmorph m} is a morphism \texttt{to\_ngraph G} to \texttt{to\_ngraph H}.
\end{lem}
The corresponding formalisation is given by:
\Snippet{morph_iff_nmorph}
We easily discharge both lemmas by unfolding the definitions of \texttt{to\_ngraph} and \texttt{to\_nmorph}, leveraging the injectivity of the corresponding functions and Isabelle's powerful automation capabilities.

An important property is that pushouts are unique up to isomorphism, which we have initially formalised in~\cite{Soeldner-Plump22a} and in this paper enhanced using the previously described techniques.

\begin{thm}[Uniqueness of pushouts \cite{Ehrig-Ehrig-Prange-Taentzer06a}]\label{thm:po-uniqueness}
  Let $b \colon A \to B$ and $c \colon A \to C$ together with $D$ induce a pushout as depicted in
  Fig.~\ref{fig:pushout}. A graph $D'$ together with morphisms $p \colon B \to D'$ and $t \colon C \to H$ is a pushout
  of $b$ and $c$ if and only if there is an isomorphism $u \colon D \to D'$ such that $u \circ f = p$
  and $u \circ g = t$.\qed
\end{thm}

We formalise this property by entering the locale \texttt{pushout\_diagram} context, which incorporates all relevant locale assumptions for use with their specified names. As a result, we are able to assume only the additional graph \(D'\) together with \(B \to D'\) and \(C \to D'\). 
\Snippet{pushout-uniq}

In our particular case, where we have injective pushouts (and rules), we can also state the uniqueness of the pushout complement.
\begin{thm}[Uniqueness of the pushout complement]
  Let $b \colon A \to B$ and $c \colon A \to C$ be injective morphisms, and the graph $D$ induce a pushout as depicted in
  Fig.~\ref{fig:pushout}. Let \(A \rightarrow B \rightarrow D \leftarrow C' \leftarrow A\) be another pushout with injective \(A \to C'\). Then \(C\) and \(C'\) are isomorphic.
  \qed
\end{thm}
We formalise this theorem inside the context of the \texttt{pushout\_diagram} locale. Since the local, in general, relies on morphisms, we need to assume, \(A \to B\) and \(A \to C\) are injective. From our assumption, that \(C' \to D\) is a valid morphism, we could conclude that \(C'\) is a valid graph. Here, we explicitly added the assumption for clarity. 
\Snippet{pushout-comp-uniq}
\noindent We discuss the proof in Section~\ref{sec:uniqueness-direct-derivation}.

\subsection{Direct Derivations}
The transformation of graphs by rules gives rise to \emph{direct derivations}.
The so-called dangling condition is a crucial requirement that ensures the resulting graph remains well-formed after applying a rule. It prevents the creation of \emph{dangling} edges, which are edges that would be left pointing to non-existent nodes after the deletion of nodes. In essence, the dangling condition guarantees that when a node is deleted, all edges connected to it must either be explicitly deleted by the rule or be preserved through the interface graph. We describe the formalisation in~\cite{Soeldner-Plump22a}.

\begin{defi}[Direct derivation]\label{def:direct-derivation}
  Let \(G\) and \(H\) be graphs, \(r =\langle L \leftarrow K \rightarrow R \rangle\) be a rule, and
  \(g \colon L \to G\) an injective morphism satisfying the dangling condition.
  Then \(G\) \emph{directly derives} \(H\) by \(r\) and \(g\), denoted by \(G \Rightarrow_{r,g} H\),
  as depicted in the double pushout diagram in Fig.~\ref{fig:double-pushout}.\qed
\end{defi}
Note that the injectivity of the matching morphism \(g \colon L \to G\) leads to a DPO approach
that is more expressive than in the case of arbitrary matches \cite{Habel-Mueller-Plump98a}.
In our earlier work \cite{Soeldner-Plump22a}, we used the \texttt{direct\_derivation} locale to represent the operational 
view using gluing and deletion. Instead, here we use the categorical definition relying on pushouts.
\begin{figure}
    \begin{tikzpicture}
	\begin{pgfonlayer}{nodelayer}
		\node [style=none] (0) at (-4, 2) {\(L\)};
		\node [style=none] (1) at (0, 2) {\(K\)};
		\node [style=none] (2) at (4, 2) {\(R\)};
		\node [style=none] (3) at (4, -2) {\(H\)};
		\node [style=none] (4) at (0, -2) {\(D\)};
		\node [style=none] (5) at (-4, -2) {\(G\)};
		\node [style=none] (6) at (-2, 0) {\((1)\)};
		\node [style=none] (7) at (2, 0) {\((2)\)};
	\end{pgfonlayer}
	\begin{pgfonlayer}{edgelayer}
		\draw [style=right arrow] (1.center) to (2.center);
		\draw [style=right arrow] (2.center) to (3.center);
		\draw [style=right arrow] (4.center) to (3.center);
		\draw [style=right arrow] (1.center) to (4.center);
		\draw [style=right arrow] (1.center) to (0.center);
		\draw [style=right arrow] (0.center) to node [left] {g} (5.center);
		\draw [style=right arrow] (4.center) to (5.center);
	\end{pgfonlayer}
\end{tikzpicture}
    \caption{Double-pushout diagram}\label{fig:double-pushout}
\end{figure}
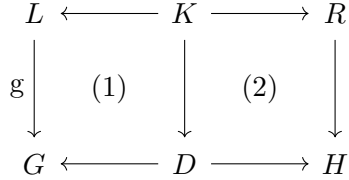

Our formalisation is close to the Def.~\ref{def:direct-derivation}. We inherit the \texttt{rule} locale, together with an injective morphism \(g\) and the two pushouts \((1)\) and \((2)\). Since we introduced the \texttt{pre\_rule} record, we use the accessor functions \texttt{lhs}, \texttt{rhs}, and \texttt{interf} to extract the corresponding objects.
\Snippet{direct-derivation}
The operational definition is available within the \texttt{direct\_derivation\_construction} locale.
\Snippet{direct-derivation-construction}
Here, we follow the same pattern but instead of the \texttt{pushout\_diagram} locale, we rely on \texttt{deletion} and \texttt{gluing}. Inside the \texttt{for} statement, we introduce a fresh \texttt{pre\_graph} record, named \(H\), which represents the final graph, depicted in Fig.~\ref{fig:double-pushout}. By using the \texttt{assumes} statement, we require that the pushout object from \((2)\) is equal to \(H\) (\texttt{H = g.H}). This allows to use the bound name \(H\), instead of \texttt{g.H}, inside the proof.

\begin{figure}
    \begin{tikzpicture}
	\begin{pgfonlayer}{nodelayer}
		\node [style=none] (0) at (-3.25, 0) {\(A\)};
		\node [style=none] (1) at (0, 3) {\(B\)};
		\node [style=none] (2) at (0, -3) {\(C\)};
		\node [style=none] (3) at (3, 0) {\(D\)};
		\node [style=none] (4) at (-6, 0) {\(A'\)};
	\end{pgfonlayer}
	\begin{pgfonlayer}{edgelayer}
		\draw [style=right arrow] (0.center) to node [fill=white, font={\tiny}] {\(b\)} (1.center);
		\draw [style=right arrow] (0.center) to node [fill=white, font={\tiny}] {\(c\)} (2.center);
		\draw [style=right arrow] (2.center) to node [fill=white, font={\tiny}] {\(g\)} (3.center);
		\draw [style=right arrow] (1.center) to node [fill=white, font={\tiny}] {\(f\)} (3.center);
		\draw [style=right arrow, dashed] (4.center) to (0.center);
		\draw [style=right arrow, bend left=15, looseness=1.25] (4.center) to node [fill=white, font={\tiny}] {\(p\)} (1.center);
		\draw [style=right arrow, bend right=15, looseness=1.25] (4.center) to node [fill=white, font={\tiny}] {\(t\)} (2.center);
	\end{pgfonlayer}
\end{tikzpicture}
    \caption{Pullback \(D \leftarrow C \leftarrow A \rightarrow B \rightarrow D\)}
    \label{fig:pullback}
\end{figure}
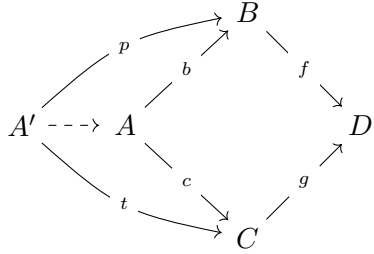

\subsection{Pullbacks}
A pullback is dual to the concept of pushout and Pullbacks generally represent the intersection of objects over a common object.
\begin{defi}[Pullback]\label{def:pullback}
  Given graph morphism \(f \colon B \to D\) and \(g \colon C \to D\), a graph \(A\) together with
  graph morphisms \(b \colon A \to B\) and \(c \colon A \to C\) is a \emph{pullback} of
  \(C \rightarrow D \leftarrow B\) if the following holds (see Fig.~\ref{fig:pullback}):
  \begin{enumerate}
  \item Commutativity: \(f \circ b = g \circ c\), and
  \item Universal property: For all graph morphisms \(p \colon A' \to B\) and
    \(t \colon A' \to C\) such that \(f \circ p = g \circ t\), there
    is a unique morphism \(u \colon H \to A\) such that \(b \circ u = p\)
    and \(c \circ u = t\).\qed
  \end{enumerate}
\end{defi}

The formalisation closely follows the \texttt{pushout\_diagram} locale, and all design considerations also apply here.
\Snippet{pullback}
Many of our proofs use the set-based construction of pullbacks, which is given by the following definition.

\begin{defi}[Pullback construction \cite{Ehrig-Ehrig-Prange-Taentzer06a}]\label{def:pullback-construction}
  Let \(f \colon B \to D\) and \(g \colon C \to D\) be graph morphisms. Then the following defines a graph \(A\)
  (see Fig.~\ref{fig:pullback}), the pullback object of \(f\) \linebreak[5]and \(g\):
  \begin{enumerate}
  \item \(A = \{\langle x, y\rangle \in B \times C\ \vert\ f(x) = g(y)\}\) for nodes and edges, respectively
  \item \(s_A(\langle x, y \rangle) = \langle s_B(x), s_c(y)\rangle\) for \(\langle x, y \rangle \in E_B \times E_C\)
  \item \(t_A(\langle x, y \rangle) = \langle t_B(x), t_c(y)\rangle\) for \(\langle x, y \rangle \in E_B \times E_C\)
  \item \(l_A(\langle x, y\rangle) = l_B(x)\) for \(\langle x, y\rangle \in V_B \times V_C\)
  \item \(m_A(\langle x, y\rangle) = m_B(x)\) for \(\langle x, y\rangle \in E_B \times E_C\)
  \item \(b \colon A \to B\) and \(c \colon A \to C\) are defined by \(b(\langle x, y\rangle) = x\) and
    \(c(\langle x, y\rangle) = y\)\qed
  \end{enumerate}
\end{defi}
We formalise the pullback construction using the \texttt{pullback\_construction} locale,
assuming the graph morphisms \(f\colon B \to D\) and \(g\colon C \to D\). 
\Snippet{pullback-construction}

To construct the pullback object \(C\), we define all the required \texttt{pre\_graph} components 
individually, starting with the node  and edge sets as follows:
\Snippet{pullback-construction-comp}
The corresponding set contains the pairs \((x,y)\) where \(x \in B\) and \(y \in C\) such that they project via \(f\) and \(g\) to the same element in \(D\).
Both source and target functions (\texttt{s} and \texttt{t}), use the corresponding function on each element individually. For the labelling functions, it does not matter and we picked to use the one from \(B\).
\Snippet{pullback-construction-comp2}
With all components available, we can now define the final \texttt{pre\_graph} object:
\Snippet{pullback-construction-obj}
In subsequent code, we used the \texttt{sublocale} keyword to show that our definition of the pullback object \(A\) is a valid graph. We continue by defining the two missing morphisms we need to show the pullback diagram locale.
First, the morphism \(A \to B\) is defined by using the first element of the tuple. We archive this by using the built-in function \texttt{fst}. Note, we explicitly added type annotation to both \texttt{pre\_morph} objects to prevent later problems if the type is inferred to strict. 
\Snippet{pullback-morphism-b-def}
For \(A \to C\), we use the other projection, \texttt{snd}, to extract the second element of the tuple.
\Snippet{pullback-morphism-c-def}

Finally, we prove that both, \texttt{b} and \texttt{c} definitions, are valid morphisms by using the \texttt{sublocale} mechanism. The proof of \(b \colon A \to B\) follows by unfolding the definition of \(A\) and \(b\) while \(c \colon A \to D\) also needs to account for the preservation of labels of \(f\) and \(g\).

The next lemma shows that this construction leads to a valid pullback diagram.
\begin{lem}[Correctness of pullback construction]
  Let \(f \colon B \to D\) and \(g \colon C \to D\) be graph morphisms and let graph \(A\) and graph morphisms $b$ and $c$ be defined as in Def.~\ref{def:pullback-construction}. Then the square in Fig.~\ref{fig:pullback} is a pullback diagram.
\end{lem}
We use the \texttt{sublocale} command, instead of \texttt{interpretation}, to make these facts persistent in the current
context via the \texttt{pb} identifier.
\Snippet{pullback-construction-correctness}
\noindent The proof basically follows from our construction. Similar to pushouts, pullbacks are unique up to isomorphism.
\begin{thm}[Uniqueness of pullbacks]\label{thm:uniqueness-pb}
  Let $f \colon B \to D$ and $g \colon C \to D$ together with $A$ induce a pullback as depicted in
  Fig.~\ref{fig:pullback}. A graph $A'$ together with morphisms $p \colon A' \to B$ and $t \colon A' \to C$ is a pullback
  of $f$ and $g$ if and only if there is an isomorphism $u \colon A' \to A$ such that $b \circ u = p$
  and $c \circ u = t$.
\end{thm}
The theorem is stated as follows in Isabelle/HOL:
\Snippet{pullback-unique}

The proof is dual to the uniqueness of the pushout (cf. Theorem~\ref{thm:po-uniqueness}) and can be found in~\cite{Ehrig-Ehrig-Prange-Taentzer06a}.
\subsection{Composition and Decomposition of Pullbacks and Pushouts}
Essential properties for the forthcoming proofs in Section \ref{sec:uniqueness-direct-derivation} and Section \ref{sec:church-rosser} are the composition and decomposition of pushouts and pullbacks.
\begin{figure}
    \begin{tikzpicture}
	\begin{pgfonlayer}{nodelayer}
		\node [style=none] (0) at (-4, 2) {\(A\)};
		\node [style=none] (1) at (0, 2) {\(B\)};
		\node [style=none] (2) at (4, 2) {\(E\)};
		\node [style=none] (3) at (4, -2) {\(F\)};
            \node [style=none] (x) at (6, -4) {\(X\)};
		\node [style=none] (4) at (0, -2) {\(D\)};
		\node [style=none] (5) at (-4, -2) {\(C\)};
		\node [style=none] (6) at (-2, 0) {\((1)\)};
		\node [style=none] (7) at (2, 0) {\((2)\)};
	\end{pgfonlayer}
	\begin{pgfonlayer}{edgelayer}
		\draw [style=right arrow] (1.center) to (2.center);
		\draw [style=right arrow] (2.center) to (3.center);
		\draw [style=right arrow] (4.center) to (3.center);
		\draw [style=right arrow] (1.center) to (4.center);
		\draw [style=right arrow] (0.center) to (5.center);
		\draw [style=right arrow] (5.center) to (4.center);
		\draw [style=right arrow] (0.center) to (1.center);
            \draw [style=right arrow] (2.center) to (x.center);
            \draw [style=right arrow] (5.center) to (x.center);
            \draw [style=right arrow] (4.center) to (x.center);
	\end{pgfonlayer}
\end{tikzpicture}
    \caption{Composite commutative diagram}\label{fig:composed-commutative-diagram}
  
\end{figure}
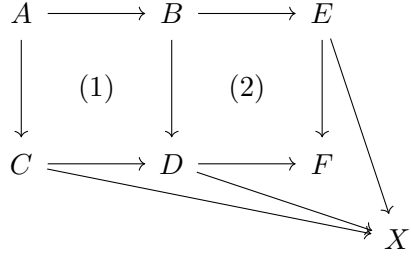

\begin{lem}[Pushout/Pullback composition and decomposition]\label{lem:comp-decomp}
  Given the commutative diagram in Fig.~\ref{fig:composed-commutative-diagram}, then the following statements are true:
  \begin{enumerate}[label=(\alph*)]
     \item If \((1)\) and \((2)\) are pushouts, so is \((1) + (2)\)
     \item If \((1)\) and \((1) + (2)\) are pushouts, so is \((2)\)
     \item If \((1)\) and \((2)\) are pullbacks, so is \((1) + (2)\)\label{lem:comp-decomp-3}
     \item If \((2)\) and \((1) + (2)\) are pullbacks, so is \((1)\)\qed
  \end{enumerate}
\end{lem}

\begin{proof}
The following proofs are based on~\cite{Ehrig-Ehrig-Prange-Taentzer06a}:
\begin{enumerate}[label=(\alph*)]
    \item Assume \((1)\) and \((2)\) are pushouts. By using Lemma~\ref{lem:morph-comp}, the compositions \(A \to B \to E\) and \(C \to D \to F\) are morphisms. Let us first show commutativity: \begin{align*}
        A \to B \to E \to F &= A \to B \to D \to F & \text{(by commutativity of (2))} \\
        &= A \to C \to D \to F & \text{(by commutativity of (1))}
    \end{align*}
    Finally, we show the universal property: Let \(X\) be a graph, and \(E \to X\) and \(C \to X\) be morphisms, such that \(A \to B \to E \to X = A \to C \to X\). By the universal property of \((1)\) using \(B \to E \to X\) and \(C \to X\), we obtain the unique morphism \(D \to X\).
    By the universal property of \((2)\) using \(D \to X\) and \(E \to X\), we obtain the unique morphism \(F \to X\). Now we first show that \(E \to F \to X = E \to X\) and \(C \to D \to F \to X = C \to X\). The first equation is true by the construction of \(F \to X\). The second equation follows by:
    \begin{align*}
        C \to D \to F \to X &= C \to D \to X & \text{(by construction of \(F \to X\))}\\
        &= C \to X & \text{(by construction of \(D \to X\))}
    \end{align*}
    By the universal property of \((2)\), we know \(F \to X\) is unique such that \(E \to F \to X = E \to X\) and \(D \to F \to X = D \to X\). We compose the second equation with the morphism \(C \to D\) on both sides, so we get \(C \to D \to F \to X = C \to D \to X = C \to X\) by construction of \(D \to X\). Hence \(F \to X\) is unique such that \(E \to F \to X = E \to X\) and \(C \to D \to F \to X = C \to X\).
    \item Assume \((1)\) and \((1) + (2)\) are pushouts and let \(X\) be a graph together with morphisms \(D \to X\) and \(E \to X\) such that \(B \to E \to X = B \to D \to X\). Let \(F \to X\) be the unique morphism resulting from the universal property of \((1) + (2)\) such that \(E \to F \to X = E \to X\) and \[C \to D \to F \to X = C \to X\text{.}\] But we also know that from the universal property of \((1)\), that \(D \to X\) is unique such that \[C \to D \to X = C \to X\text{.}\]
    Hence \(D \to F \to X = D \to X\) because of the uniqueness of \(D \to X\).
    \item Proof is analogous via dualisation. 
    \item Proof is analogous via dualisation. \qedhere
\end{enumerate}
\end{proof}
Our Isabelle formalisation is an adaptation of a given proof, incorporating necessary technical modifications while maintaining the essence of the original argument. The pushout composition is a standalone lemma defined in the theory \texttt{Pushout}, cf. Fig.~\ref{fig:locales}, following (a) of Lemma~\ref{lem:comp-decomp}.
\Snippet{pushout-composition}
Here, \(f \colon A \to B\), \(g \colon A \to C\), \(g' \colon B \to D\) and \(f' \colon C \to D\) together with
\(e \colon B \to E\), \(e' \colon D \to F\), and \(e'' \colon E \to F\) forms the composed pushout diagram.

For the pushout decomposition, we additionally assume that \((2)\) commutes, i.e., \(e'' \circ e = e' \circ g'\).
In Isabelle, we express this independently for the nodes and edge of \(B\).
\Snippet{pushout-decomposition}

In our formalisation, the parts (c) and (d) of the pullback are developed analogously to the pushout ones. 
\Snippet{pullback-composition}
\Snippet{pullback-decomposition}
The proof is analogous to \cite[Fact 2.27]{Ehrig-Ehrig-Prange-Taentzer06a}.
In the next section, we show that direct derivations have a unique (up to isomorphism) result by making use of the infrastructure we have employed so far.

\section{Uniqueness of Direct Derivations}\label{sec:uniqueness-direct-derivation}

The uniqueness of direct derivations is an important property when reasoning about rule applications. This section does not rely on the adhesiveness of the category of graphs, instead we base our proof on the characterisation of graph pushouts in \cite{EK79}.
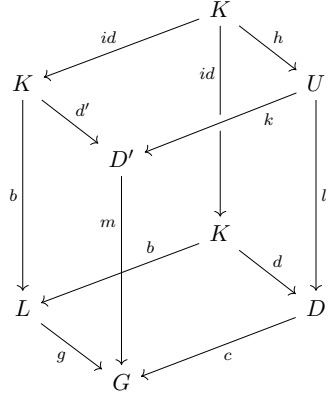
\begin{figure}
  \begin{center}
    \adjustbox{scale=0.8,center}{%
      \begin{tikzcd}
        & & |[alias=CT]|K \arrow[lld, "id", swap] \arrow[rd, "h"] \arrow[ddd, near start,"id", swap] & \\
        |[alias=C]|K \arrow[rd, "d'"] \arrow[ddd,"b", swap] & & & |[alias=U]|U \arrow[lld, crossing over, near start,"k"] \arrow[ddd, "l"] \\
        & |[alias=BP]|D' \arrow[ddd, crossing over, near start, "m", swap] & &\\
        & & |[alias=CB]|K \arrow[lld, near start,"b",swap] \arrow[rd, "d"] &\\
        |[alias=A]|L \arrow[rd, "g",swap] & & & |[alias=B]|D \arrow[lld,"c"]\\
        & |[alias=D]|G & &
    \end{tikzcd}}
  \end{center}
  \caption{Commutative cube based on direct derivation~\cite{Lack-Sobocinski04a}}
  \label{fig:commutative-cube}
\end{figure}
Before stating the theorem, we introduce additional facts, mainly about pushouts and pullbacks, which are used within
the proof of Theorem~\ref{thm:uniqueness-dd}. 

In general, pushouts along injective morphisms are also pullbacks.
\begin{lem}[Injective pushouts are pullbacks~\cite{Ehrig-Ehrig-Prange-Taentzer06a}]\label{lem:injective-po-is-pb}
  A pushout diagram as depicted in Fig.~\ref{fig:pushout} is also a pullback if \(A \to B\) and \(A \to C\) are injective.
\end{lem}
The proof relies on the pullback construction (cf. Def.~\ref{def:pullback-construction}) and the fact that pullbacks
are unique (cf. Theorem~\ref{thm:uniqueness-pb}). We show this property in our formalisation within the \texttt{Gluing} theory (cf. Fig.~\ref{fig:locales}). 
\Snippet{inj-po-is-pb}

Furthermore, pushouts and pullbacks preserve injectivity (surjectivity) in the sense that the
opposite morphism of the corresponding diagram (see Fig.~\ref{fig:pushout} and Fig.~\ref{fig:pullback}) is also injective (surjective). 
\begin{lem}[Preservation of injective and surjective morphisms~\cite{Ehrig-Ehrig-Prange-Taentzer06a}]\label{lem:preservation-inj-surj}
  Given a pushout diagram in Fig.~\ref{fig:pushout}, if \(A \to B\) is injective (surjective), so is \(C \to D\).
  Given a pullback diagram in Fig.~\ref{fig:pullback}, if \(C \to D\) is injective (surjective), so is \(A \to B\).
\end{lem}
We formalise these properties independent of the pushout, pullback and injectivity or surjectivity using our infrastructure.
Inside the \texttt{pushout} locale context, if \(b \colon A \to B\) is injective, so is \(g \colon C \to D\) as depicted in Fig.~\ref{fig:pushout}.
\Snippet{pushout-b-inj-g}
In case, \(b \colon A \to B\) is surjective, so is \(g \colon C \to D\).
\Snippet{po-surj-b-is-g}
Consequently, in case \(A \to B\) is bijective, so is \(C \to D\).
\Snippet{po-bij-b-is-g}
The pullback statements are done similar in the corresponding locale context.

Certain forms of commutative diagrams give rise to pullbacks. This property is used in the proof of
the uniqueness of the pushout complement (cf. Theorem~\ref{thm:uniqueness-dd}).
\begin{lem}[Special pullbacks~\cite{Ehrig-Ehrig-Prange-Taentzer06a}]\label{lem:special-pb}
  The commutative diagram in Fig.~\ref{fig:pullback-specialdiag} is a pullback
  if \(m\) is injective.
\end{lem}

In Isabelle, we describe this lemma as follows.
\Snippet{special-diag-is-pb}
Isabelle is able to discharge this goal using supplied facts automatically. 

\begin{defi}[Reduced chain-condition \cite{EK79}]\label{def:reduced-chain}
  The commutative diagram in Fig.~\ref{fig:commutative-diagram} satisfies the
  \emph{reduced chain-condition}, if for all \(b' \in B\) and \(c' \in C\)
  with \(f(b') = g(c')\) there is \(a \in A\) such that \(b(a) = b'\)
  and \(c(a) = c'\).
\end{defi}
We show that pullbacks satisfy the reduced chain-condition.
\begin{lem}[Pullbacks satisfy the reduced chain-condition]
  \label{lem:pb-redchain}
  Each pullback diagram as depicted in Fig.~\ref{fig:pullback} satisfies the
  reduced chain-condition.  
\end{lem}
We first prove this lemma, separately for nodes and edges, state this lemma in Isabelle as follows:
\Snippet{pbconstr-rchain}
Isabelle's simplifier is able to discharge all goals if supplied with the corresponding facts. Especially, that the construction is also a pullback (commutativity) and the definition of the graph \(A\), together with the two morphisms.
Subsequently, we can state a more generalised lemma inside the \texttt{pullback\_diagram} locale:
\Snippet{pullback-reduced-chain}
%
Our proof relies on the pullback construction (cf. Def.~\ref{def:pullback-construction}) and the fact that
pullbacks are unique (cf. Theorem~\ref{thm:uniqueness-pb}).

\begin{defi}[Jointly surjective]
  Given injective graph morphisms \(f \colon B \to D\) and \(g \colon C \to D\).
  \(f\) and \(g\) are \emph{jointly surjective}, if each item in \(D\)
  has a preimage in \(B\) or \(C\).
\end{defi}

\begin{lem}[Pushouts are jointly surjective]
Given the pushout diagram in Fig.~\ref{fig:pushout}. The pair \((f,g)\) is \emph{jointly surjective}.
\end{lem}

We express this property for edges in Isabelle as follows:
\Snippet{pushout-joint-surj}
The proof is by contradiction. We construct the graph \(D'\) using the disjoint union with an edge \(e \in D\), such that it has no preimage in \(B\) or \(C\) via \(f\) and \(g\). We can now construct two different morphisms from \(D'\) to \(D\), which contradicts the uniqueness property of the pushout.
\begin{thm}[Pushout characterization \cite{EK79}] \label{thm:pushout-characterization}
  The commutative diagram in Fig.~\ref{fig:commutative-diagram} is a pushout,
  if the following conditions are true:
  \begin{enumerate}
  \item The morphisms \(b,c,f,g\) are injective.
  \item The diagram satisfies the reduced chain-condition.
  \item The morphisms \(g,f\) are jointly surjective.
  \end{enumerate}
\end{thm}

\Snippet{pushout-characterization}

\begin{figure}
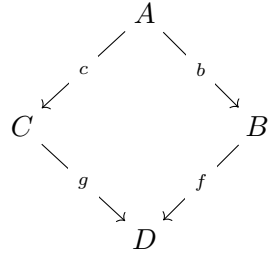

  \ctikzfig{commutative-diagram}
  \caption{Commutative diagram}
  \label{fig:commutative-diagram}
\end{figure}

\begin{figure}
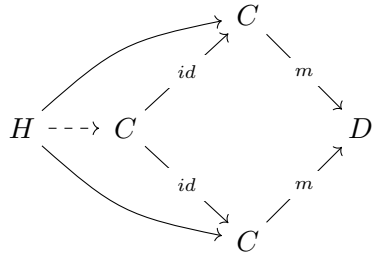

   \ctikzfig{pullback-specialdiag}
  \caption{Special pullback diagram}
  \label{fig:pullback-specialdiag}
\end{figure}

The following theorem implies the uniqueness of the pushout complement, which is known to hold if the morphism $K \to L$\/ in the applied rule is injective, even if the matching morphism $L \to G$ is non-injective \cite{Ros75a}.  In our case, both morphisms are injective.

\begin{thm}[Uniqueness of direct derivations]\label{thm:uniqueness-dd}
  Let \((1)+(2)\) and \((3)+(4)\) be direct derivations as depicted in Fig.~\ref{fig:uniqueness-direct-derivation}.
  Then \(D \cong D'\) and \(H \cong H'\).
\end{thm}
This theorem is stated in Isabelle/HOL within the locale \texttt{direct\_derivation}. Within the \texttt{assumes} part,
the second direct derivation, which we call \texttt{dd2}, is introduced.
\Snippet{direct-derivation-unique}
The uniqueness proof of direct derivations (see Fig.\ref{fig:uniqueness-direct-derivation}) is performed in two phases.
Firstly, we show the uniqueness of the pushout complement, which was first shown by Rosen~\cite{Ros75a}.
Subsequently, we show that given a bijection between \(D\) and \(D'\), the pushout object is also unique up to isomorphism.
%
%
\begin{proof}[Theorem~\ref{thm:uniqueness-dd}]
  The first phase of our proof closely follows Lack and Sobocinski~\cite{Lack-Sobocinski04a}, except for
  the final step, where the authors rely on adhesiveness. We finish the proof by relying on the
  pushout characterization (cf. Theorem \ref{thm:pushout-characterization}).
  Given the two pushout diagrams \((1)\) and \((3)\) in Fig.~\ref{fig:uniqueness-direct-derivation}
  with injective \(K \to D\) and \(K \to D'\).
  To show the existence of a bijection between \(D\) and \(D'\), we construct the commutative cube
  in Fig.~\ref{fig:commutative-cube} with \((1)\) as the bottom face and \((3)\) as the front-left face,
  and show that \(l\) and \(k\) are bijections. For the latter, we show that the back-right and top faces are pushouts.
  (In \cite{Lack-Sobocinski04a}, this is shown by adhesiveness, while we argue with the pushout characterization of Theorem~\ref{thm:pushout-characterization}.)
  The front-right face is a pullback construction (cf. Def.~\ref{def:pullback-construction}) which we
  tell Isabelle by interpretation of the \texttt{pullback\_construction} locale.
  \Snippet{udd-fr}
  We use Isabelle's shorthand notation \texttt{..} for the \texttt{standard} tactic, to discharge the
  proof obligations which follow from the assumptions. Note that the pullback object together with the two morphisms is specified within the locale. Subsequent code will reference the pullback object by \texttt{fr.A}, the morphism \(l\) by \texttt{fr.b} and \(k\) by \texttt{fr.c} (see Fig.~\ref{fig:commutative-cube}). (The identifiers within locales are given by the definition. As a result, the pullback object of the front-right face is referred to as \texttt{A} rather than the interpretation parameter \texttt{K}.)
  From Lemma~\ref{lem:special-pb}, in our formalisation referenced by \texttt{fun\_algrtr\_4\_7\_2}, we know that the back-left face is a pullback.
  \Snippet{udd-bl}
  To show that the back-right face is a pullback, we start with the front-left face.
  As the front left face is a pushout and \(m\) is injective, \(n'\) is too (cf. Lemma~\ref{lem:preservation-inj-surj}).
  Since pushouts along injective morphisms are also pullbacks (cf. Lemma~\ref{lem:injective-po-is-pb}), the
  front-left face is also a pullback.
  Using the pullback composition (cf. Lemma~\ref{lem:comp-decomp}), the back face is pullback.
 \Snippet{udd-b}
  We define \(h \colon K \to U\) using both, the \(d\) and \(d'\) morphisms as \(h~x = (d~x, d'~x)\) and
  subsequently prove the morphism properties.
 \Snippet{udd-defh}
  We follow by showing that the top and bottom face commutes, i.e., \(d' \circ id = k \circ h\) and
  \(g \circ b = c \circ d\), respectively. This establishes the fact, that the right-side of the cube is
  a pullback. Using the pullback decomposition (cf. Lemma~\ref{lem:comp-decomp}), the back-right face is a pullback.
  To approach the top-face, we start by showing it is a pullback, and subsequently it is also a pushout.
  Since \(m\) is injective, from Lemma~\ref{lem:injective-po-is-pb} we know the bottom-face is also a pullback.
  Using the pullback composition (cf. Lemma~\ref{lem:comp-decomp}), the bottom and back-left face is a pullback.
  By commutativity of the bottom face \(g \circ b = c \circ d\) and back-right face \(l \circ h = d \circ id\),
  the front-right and top face is a pullback.
  By the pullback decomposition (cf. Lemma~\ref{lem:comp-decomp}) we can show that the top face is a pullback.
  We show this pullback is also a pushout by using the pushout characterization (cf. Theorem~\ref{thm:pushout-characterization}).
  Therefore, we need to show that \(h\) is injective, which follows from the construction above of \(h\).
  \Snippet{udd-injh}
  Joint surjectivity of \(k\) and \(d'\) follows from the pullback construction and the
  reduced-chain condition (cf. Lemma~\ref{def:reduced-chain}) of the front-left and top face.
  Note, that the reduced-chain condition holds for all pullbacks.
  Finally, we need to show that \(k\) and \(l\) are bijections.
  Since the top face is a pushout and the \(C \to C\) morphism is a bijection, by Lemma~\ref{lem:preservation-inj-surj}, so is \(k\).
  \Snippet{udd-bijk}
  To show \(l\) is a bijection, we show that the back-right face is a pushout by using the pushout characterization.
  The bijectivity of \(l\) follows from the fact pushouts preserve bijections.
  We follow by defining the morphism \(u \colon D \to D'\) as the composition of the \(l^{-1}\) and \(k\).
  The inverse of \(l\) is obtained by using a lemma within our formalisation of bijective morphisms, stating the existence
  of the inverse (\texttt{ex\_inv}) using the \texttt{obtain} keyword:
  \Snippet{udd-linv}
  We finish the first phase by defining the morphism \(u \colon D \to D'\) as \(k~\circ~l^{-1}\) and 
  subsequently prove that morphism composition preserves bijections (using the already
  proven \texttt{bij\_comp\_is\_bij} lemma).
  \Snippet{udd-defu}
  \Snippet{udd-biju}
  
  The second phase is to show the existence of an isomorphism \(H \to H'\). We start by obtaining
  \(u' \colon H \to H'\) and \(u'' \colon H' \to H\) and show they are inverses.
  \begin{figure}
    \begin{tikzpicture}
	\begin{pgfonlayer}{nodelayer}
		\node [style=none] (0) at (0.5, 0) {\(K\)};
		\node [style=none] (1) at (3.75, 3) {\(R\)};
		\node [style=none] (2) at (3.75, -3) {\(D\)};
		\node [style=none] (3) at (6.75, 0) {\(H\)};
		\node [style=none] (4) at (9.75, 0) {\(H'\)};
	\end{pgfonlayer}
	\begin{pgfonlayer}{edgelayer}
		\draw [style=right arrow, in=225, out=45] (0.center) to node [fill=white, font={\tiny}] {\(b'\)} (1.center);
		\draw [style=right arrow] (0.center) to node [fill=white, font={\tiny}] {\(d\)} (2.center);
		\draw [style=right arrow] (2.center) to node [fill=white, font={\tiny}] {\(c'\)} (3.center);
		\draw [style=right arrow] (1.center) to node [fill=white, font={\tiny}] {\(f\)} (3.center);
		\draw [style=right arrow, bend left=15, looseness=1.25] (1.center) to node [fill=white, font={\tiny}] {\(f'\)} (4.center);
		\draw [style=right arrow, bend right=15, looseness=1.25] (2.center) to node [fill=white, font={\tiny}] {\(m' \circ u\)} (4.center);
		\draw [style=right arrow, dashed] (3.center) to node [fill=white, font={\tiny}] {\(u'\)} (4.center);
	\end{pgfonlayer}
\end{tikzpicture}
  \caption{Construction of \(u'\)}
  \label{fig:uniqueness-h-1}
\end{figure}
\begin{figure}
  \center
   \begin{tikzpicture}
	\begin{pgfonlayer}{nodelayer}
		\node [style=none] (0) at (0.5, 0) {\(K\)};
		\node [style=none] (1) at (3.75, 3) {\(R\)};
		\node [style=none] (2) at (3.75, -3) {\(D'\)};
		\node [style=none] (3) at (6.75, 0) {\(H'\)};
		\node [style=none] (4) at (9.75, 0) {\(H\)};
	\end{pgfonlayer}
	\begin{pgfonlayer}{edgelayer}
		\draw [style=right arrow, in=225, out=45] (0.center) to node [fill=white, font={\tiny}] {\(b'\)} (1.center);
		\draw [style=right arrow] (0.center) to node [fill=white, font={\tiny}] {\(d'\)} (2.center);
		\draw [style=right arrow] (2.center) to node [fill=white, font={\tiny}] {\(m'\)} (3.center);
		\draw [style=right arrow] (1.center) to node [fill=white, font={\tiny}] {\(f'\)} (3.center);
		\draw [style=right arrow, bend left=15, looseness=1.25] (1.center) to node [fill=white, font={\tiny}] {\(f\)} (4.center);
		\draw [style=right arrow, bend right=15, looseness=1.25] (2.center) to node [fill=white, font={\tiny}] {\(c' \circ u^{-1}\)} (4.center);
		\draw [style=right arrow, dashed] (3.center) to node [fill=white, font={\tiny}] {\(u''\)} (4.center);
	\end{pgfonlayer}
\end{tikzpicture}
  \caption{Construction of \(u''\)}
  \label{fig:uniqueness-h-2}
\end{figure}
  We use the universal property of the pushout depicted in Fig.~\ref{fig:uniqueness-h-1}, which requires us to show commutativity: \(f' \circ b' = m' \circ u \circ d\). So we substitute \(u \circ d = d'\) into the commutativity equation of pushout \((4)\) (\(f' \circ b' = m' \circ d'\)) in Fig.~\ref{fig:uniqueness-direct-derivation}. We get \(u \circ d = d'\) as follows:
  \( u \circ d \overset{(1)}{=} k \circ l^{-1} \circ d \overset{(2)}{=} k \circ l^{-1} \circ l \circ h \overset{(3)}{=}
    k \circ h \overset{(4)}{=} d'\).
  Here, \((1)\) is justified by the definition of \(u\), \((2)\) by the definition of \(l\) and \(h\) (which makes the back-right face in Fig.~\ref{fig:commutative-cube} commute), \((3)\) by inverse cancellation, and finally \((4)\) by the definitions of \(k\) and \(h\) (similarly to step \((2)\)).
  We obtain \(u'' \colon H' \to H\) by using the universal property of the pushout depicted in Fig.~\ref{fig:uniqueness-h-2}. We show the commutativity \(f \circ b = c' \circ u^{-1} \circ d'\) by substituting \(u^{-1} \circ d' = d\) into the commutativity equation of pushout \((2)\) (\(f \circ b' = c' \circ d\)) in Fig.~\ref{fig:uniqueness-direct-derivation}:
  \(u^{-1} \circ d' \overset{(5)}{=} u^{-1} \circ u \circ d \overset{(6)}{=} d\).
  Here, \((5)\) is justified by the above proven equation \(u \circ d = d'\), and \((6)\) follows from cancellation of inverses.
    The final steps are to show \(u' \circ u'' = id\) and \(u'' \circ u' = id\). To show the first equation, we start
  with \(f' = u' \circ u'' \circ f'\) and \(m' = u' \circ u'' \circ m'\), which we get from the definitions of \(u'\) and \(u''\).
  Using the universal property of pushout \((4)\) in Fig.~\ref{fig:uniqueness-direct-derivation} together with \(H', f', m'\), we conclude that the identity is the unique morphism \(H' \to H'\) that makes the triangles commute. If \(u' \circ u''\) makes the triangles commute as well, it is equal to the identity morphism.
  The first triangle commutes because \(u' \circ u'' \circ f' \overset{(7)}{=} u' \circ f \overset{(8)}{=} f'\).
  Here, \((7)\) and \((8\)) are justified by the corresponding construction of \(u'\) and \(u''\) (see the triangles in Fig.~\ref{fig:uniqueness-h-1} and Fig.~\ref{fig:uniqueness-h-2}).
  For the second triangle, we start by using the commutativity of the bottom triangle in Fig.~\ref{fig:uniqueness-h-1} and composing to the right with \(u^{-1}\): \(u' \circ c' \circ u^{-1} = m' \circ u \circ u^{-1}\). By cancellation of inverses we get \(u' \circ c' \circ u^{-1} = m'\) and by substituting \(c' \circ u^{-1}\) using the commutativity of the bottom triangle in Fig.~\ref{fig:uniqueness-h-2}, we prove that \(u' \circ u'' \circ m' = m'\).
  Showing that \(u'' \circ u' = id\) follows analogously and is omitted to save space.
\end{proof}
\begin{figure}
    \center
    \begin{tikzpicture}
	\begin{pgfonlayer}{nodelayer}
		\node [style=none] (0) at (-4, 2) {\(L\)};
		\node [style=none] (1) at (0, 2) {\(K\)};
		\node [style=none] (2) at (4, 2) {\(R\)};
		\node [style=none] (3) at (4, -2) {\(H\)};
		\node [style=none] (4) at (0, -2) {\(D\)};
		\node [style=none] (5) at (-4, -2) {\(G\)};
		\node [style=none] (6) at (-1.75, 0) {\((1)\)};
		\node [style=none] (7) at (2, 0) {\((2)\)};
		\node [style=none] (8) at (0, -4.75) {\(D'\)};
		\node [style=none] (9) at (0, 2) {};
		\node [style=none] (10) at (4, 2) {};
		\node [style=none] (11) at (4, -4.75) {\(H'\)};
		\node [style=none] (12) at (-4, -4.75) {\(G\)};
		\node [style=none] (13) at (-1.75, -3.25) {\((3)\)};
		\node [style=none] (14) at (2, -3.25) {\((4)\)};
		\node [style=none] (15) at (-4, 2) {};
	\end{pgfonlayer}
	\begin{pgfonlayer}{edgelayer}
		\draw [style=right arrow] (1.center) to (2.center);
		\draw [style=right arrow] (2.center) to (3.center);
		\draw [style=right arrow] (4.center) to (3.center);
		\draw [style=right arrow] (1.center) to (4.center);
		\draw [style=right arrow] (1.center) to (0.center);
		\draw [style=right arrow] (0.center) to node [left] {g} (5.center);
		\draw [style=right arrow] (4.center) to (5.center);
		\draw [style=right arrow, bend left=15] (9.center) to (8.center);
		\draw [style=right arrow, bend left=15] (10.center) to (11.center);
		\draw [style=right arrow] (8.center) to (11.center);
		\draw [style=right arrow] (8.center) to (12.center);
		\draw [style=right arrow, bend left=15] (15.center) to node [xshift=5, yshift=5] {g} (12.center);
	\end{pgfonlayer}
\end{tikzpicture}
  \caption{Uniqueness of direct derivations}
  \label{fig:uniqueness-direct-derivation}
\end{figure}
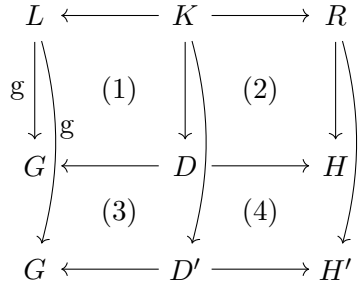

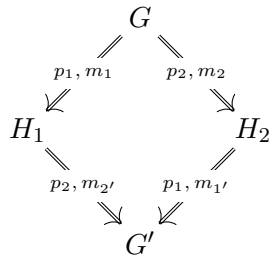
\begin{figure}
  \center
  \begin{tikzpicture}
	\begin{pgfonlayer}{nodelayer}
		\node [style=none] (0) at (0, 3) {\(G\)};
		\node [style=none] (1) at (3, 0) {\(H_2\)};
		\node [style=none] (2) at (-3, 0) {\(H_1\)};
		\node [style=none] (3) at (0, -3) {\(G'\)};
	\end{pgfonlayer}
	\begin{pgfonlayer}{edgelayer}
		\draw [style=right arrow, double] (0.center) to node [fill=white, font={\tiny}] {\(p_1,m_1\)} (2.center);
		\draw [style=right arrow, double] (0.center) to node [fill=white, font={\tiny}] {\(p_2,m_2\)} (1.center);
		  \draw [style=right arrow, double] (1.center) to node [fill=white, font={\tiny}] {\(p_1,m_{1'}\)} (3.center);
		\draw [style=right arrow, double] (2.center) to node [fill=white, font={\tiny}] {\(p_2,m_{2'}\)} (3.center);
	\end{pgfonlayer}
\end{tikzpicture}
  \caption{Church-Rosser Theorem}
  \label{fig:church-rosser}
\end{figure}
With the uniqueness of direct derivations (cf. Theorem~\ref{thm:uniqueness-dd}), we get the uniqueness of
the pushout complement.
\begin{cor}[Uniqueness of pushout complements]
  Given a pushout as depicted in Fig.~\ref{fig:pushout} where \(A \to B\) is injective.
  Then the graph \(D\) is unique up to isomorphism.
\end{cor}
We omit the proof to conserve space. The upcoming section introduces the so-called Church-Rosser Theorem,
which states that parallell independent direct derivations have the diamond property.
\section{Church-Rosser Theorem}\label{sec:church-rosser}
The Church-Rosser Theorem refers to the idea that two graph transformation rules can be applied independently of each other,
either sequentially or in parallel, without changing the final result. We follow the independence characterization of
direct derivations given in~\cite{Ehrig-Ehrig-Prange-Taentzer06a}.
\begin{defi}[Parallel independence~\cite{Ehrig-Ehrig-Prange-Taentzer06a}]
\label{def:parallel-independence}
  The two direct derivations \(G\Rightarrow_{p_1, m_1} H_1\) and \(G \Rightarrow_{p_2, m_2} H_2\)
  in Fig.~\ref{fig:parallel-independence} are \emph{parallel independent} if there exists morphisms \(L_1 \to D_2\) and
  \(L_2 \to D_1\) such that \(L_1 \to D_2 \to G  = L_1 \to G\) and \(L_2 \to D_1 \to G = L_2 \to G\). \qed
\end{defi}
\begin{figure}
  \center
  \begin{tikzpicture}
	\begin{pgfonlayer}{nodelayer}
		\node [style=none] (0) at (-10, 2) {\(R_1\)};
		\node [style=none] (1) at (-6, 2) {\(K_1\)};
		\node [style=none] (2) at (-2, 2) {\(L_1\)};
		\node [style=none] (3) at (0, -1.5) {\(G\)};
		\node [style=none] (4) at (2, 2) {\(L_2\)};
		\node [style=none] (5) at (6, 2) {\(K_2\)};
		\node [style=none] (6) at (10, 2) {\(R_2\)};
		\node [style=none] (7) at (10, -1.5) {\(H_2\)};
		\node [style=none] (8) at (6, -1.5) {\(D_2\)};
		\node [style=none] (9) at (-10, -1.5) {\(H_1\)};
		\node [style=none] (10) at (-6, -1.5) {\(D_1\)};
		\node [style=none] (11) at (-8.25, 0.25) {\small\((2)\)};
		\node [style=none] (12) at (-4, 0.25) {\small\((1)\)};
		\node [style=none] (13) at (4, 0.25) {\small\((3)\)};
		\node [style=none] (14) at (8, 0.25) {\small\((4)\)};
	\end{pgfonlayer}
	\begin{pgfonlayer}{edgelayer}
		\draw [style=right arrow] (0.center) to (9.center);
		\draw [style=right arrow] (1.center) to (0.center);
		\draw [style=right arrow] (1.center) to (2.center);
		\draw [style=right arrow] (1.center) to (10.center);
		\draw [style=right arrow] (10.center) to (9.center);
		\draw [style=right arrow] (10.center) to (3.center);
		\draw [style=right arrow] (2.center) to node [fill=white, font={\tiny}, yshift=-5pt, xshift=2pt] {\(m_1\)} (3.center);
		\draw [style=right arrow] (4.center) to node [fill=white, font={\tiny}, yshift=-5pt, xshift=-2pt] {\(m_2\)} (3.center);
		\draw [style=right arrow] (5.center) to (4.center);
		\draw [style=right arrow] (5.center) to (6.center);
		\draw [style=right arrow] (6.center) to (7.center);
		\draw [style=right arrow] (8.center) to (7.center);
		\draw [style=right arrow] (5.center) to (8.center);
		\draw [style=right arrow] (8.center) to (3.center);
		\draw [style=right arrow, dashed] (2.center) to (8.center);
		\draw [style=right arrow, dashed] (4.center) to (10.center);
	\end{pgfonlayer}
\end{tikzpicture}
  \caption{Parallel independence}
  \label{fig:parallel-independence}
\end{figure}
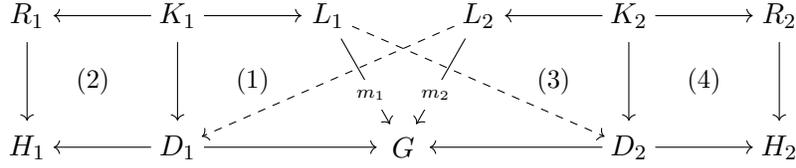

\Snippet{parallelindependence}

\begin{defi}[Sequential independence~\cite{Ehrig-Ehrig-Prange-Taentzer06a}]
\label{def:sequential-independence}
  The two direct derivations \(G\Rightarrow_{p_1, m_1} H_1\) and \(G \Rightarrow_{p_2, m_2} H_2\)
  in Fig.~\ref{fig:sequential-independence} are \emph{sequential independent} if there exists morphisms \(R_1 \to D_2\) and
  \(L_2 \to D_1\) such that \(R_1 \to D_2 \to H  = R_1 \to H\) and \(L_2 \to D_1 \to H = L_2 \to H\).\qed
\end{defi}

\begin{figure}
  \center
  \begin{tikzpicture}
	\begin{pgfonlayer}{nodelayer}
		\node [style=none] (0) at (-10, 2) {\(L_1\)};
		\node [style=none] (1) at (-6, 2) {\(K_1\)};
		\node [style=none] (2) at (-2, 2) {\(R_1\)};
		\node [style=none] (3) at (0, -1.5) {\(H\)};
		\node [style=none] (4) at (2, 2) {\(L_2\)};
		\node [style=none] (5) at (6, 2) {\(K_2\)};
		\node [style=none] (6) at (10, 2) {\(R_2\)};
		\node [style=none] (7) at (10, -1.5) {\(G\prime\)};
		\node [style=none] (8) at (6, -1.5) {\(D_2\)};
		\node [style=none] (9) at (-10, -1.5) {\(G\)};
		\node [style=none] (10) at (-6, -1.5) {\(D_1\)};
            \node [style=none] (11) at (-8.25, 0.25) {\small\((1)\)};
		\node [style=none] (12) at (-4, 0.25) {\small\((2)\)};
		\node [style=none] (13) at (4, 0.25) {\small\((3)\)};
		\node [style=none] (14) at (8, 0.25) {\small\((4)\)};
	\end{pgfonlayer}
	\begin{pgfonlayer}{edgelayer}
		\draw [style=right arrow] (0.center) to (9.center);
		\draw [style=right arrow] (1.center) to (0.center);
		\draw [style=right arrow] (1.center) to (2.center);
		\draw [style=right arrow] (1.center) to (10.center);
		\draw [style=right arrow] (10.center) to (9.center);
		\draw [style=right arrow] (10.center) to (3.center);
		\draw [style=right arrow] (2.center) to node [fill=white, font={\tiny}] {\(n_1\)} (3.center);
		\draw [style=right arrow] (4.center) to node [fill=white, font={\tiny}] {\(m_2\)} (3.center);
		\draw [style=right arrow] (5.center) to (4.center);
		\draw [style=right arrow] (5.center) to (6.center);
		\draw [style=right arrow] (6.center) to (7.center);
		\draw [style=right arrow] (8.center) to (7.center);
		\draw [style=right arrow] (5.center) to (8.center);
		\draw [style=right arrow] (8.center) to (3.center);
		\draw [style=right arrow, dashed] (2.center) to (8.center);
		\draw [style=right arrow, dashed] (4.center) to (10.center);
	\end{pgfonlayer}
\end{tikzpicture}
  \caption{Sequential independence}
  \label{fig:sequential-independence}
\end{figure}
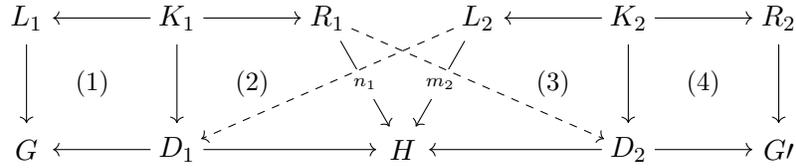

\Snippet{sequentialindependence}

\begin{thm}[Church-Rosser Theorem~\cite{EK76}] \label{thm:church-rosser}
  Given two parallel independent direct derivations \(G \Rightarrow_{p_1,m_1} H_1\) and \(G \Rightarrow_{p_2, m_2} H_2\),
  there is a graph \(G'\) together with sequential independent direct derivations \(H_1 \Rightarrow_{p_2, m'_2} G'\) and
  \(H_2 \Rightarrow_{p_1, m'_1} G'\).
\end{thm}
Actually, we have shown more, namely that \(G \Rightarrow_{p_1, m_1} H_1 \Rightarrow_{p_2, m'_2} G'\) and \(G \Rightarrow_{p_2, m_2} H_2 \Rightarrow_{p_1, m'_1} G'\) are sequentially independent.
We express this theorem in Isabelle/HOL within the \texttt{parallel\_independence} locale as follows:
\Snippet{churchrosser}
\begin{proof}[Theorem~\ref{thm:church-rosser}]
  We closely follow the original proof by Ehrig and Kreowksi~\cite{EK76} where in a first stage,
  the pushouts (1) - (4) of Fig.~\ref{fig:parallel-independence} are vertically decomposed into pushouts
  \((11)+(11)\), \((21)+(22)\), \((31)+(32)\), and \((41)+(42)\), as depicted in Fig.~\ref{fig:church-rosser-decomposition}.
  In a second stage, these pushouts are rearranged as in Fig.~\ref{fig:church-rosser-rearrange} and the new pushout (5) is constructed.
  Subsequently, we prove the two vertical pushouts \((11)\) and \((12)\).
  The pushouts \((31)\) and \((32)\) follow analogously and are not shown to conserve space.

  We start by constructing the pullback \((12\)) which we bind to the symbol \texttt{c12}, allowing later references, using
  our \texttt{pullback\_construction} locale.
  \Snippet{churchrosser-pb}
  The existence of \(K_1 \to D\) follows from the universal property, and \(D \to D_2\) from the construction of the pullback \((12)\):
  \Snippet{churchrosser-j1}
  From the fact that \((1) = (11)+(12)\), we know (\(11)+(12)\) is a pushout and since \(K_1 \to L_1\) is injective,
  it is also a pullback (cf. Lemma~\ref{lem:injective-po-is-pb}).
  By pullback decomposition (cf. Lemma~\ref{lem:comp-decomp}), \((11)\) is a pullback.
  We use the pushout characterization (cf. Theorem~\ref{thm:pushout-characterization}) to show it is also a pushout, which requires us to show
  injectivity of all morphisms, reduced-chain condition, and joint surjectivity of \(D \to D_2\) and \(L_1 \to D_2\).
  The injectivity of \(K_1 \to L_1\) is given, \(D \to D_2\) follows from pushout \((1)\) and the injectivity of \(K_1 \to L_1\) (cf. Lemma~\ref{lem:preservation-inj-surj}). To show the injectivity of \(L_1 \to D_2\), we use the parallel independence (cf. Def.~\ref{def:parallel-independence}) \( L_1 \to D_2 \to G = L_1 \to G\) and the injectivity of \(L_1 \to G\).To show injectivity of \(K_1 \to D\), we use the triangle \(L_1 \to D_2 \to G = L_1 \to G\) obtained by the universal property of pullback \((12)\) and the injectivity of both, \(L_1 \to G\) and \(D_2 \to G\).
  The reduced-chain condition follows by Lemma~\ref{lem:pb-redchain}.
  To show the joint surjectivity of \(D \to D_2\) and \(L_1 \to D_2\) (that is each \(x\) in \(D_2\) has a preimage in either \(D\) or \(L_1\)).
  Let \(y\) be the image of \(x\) in \(G\). We apply the joint surjectivity of pushout \((11) + (12)\) to \(y\), that is \(y\) has a preimage in either \(D_1\) or \(L_1\). In the former case (\(y\) has a preimage \(z\) in \(D_1\)): from the pullback construction (cf. Def.~\ref{def:pullback-construction}), we get the common preimage of \(z\) and \(x\) in \(D \) which shows the former case.
  In the latter case, \(y\) has a preimage in \(L_1\) via \(D_2\). Since \(D_2 \to G\) is injective, that preimage is mapped via \(x\) which means, \(x\) has a preimage in \(L_1\). This shows the latter case.

  \begin{figure}[ht]
    \center
    \scalebox{0.8}{\begin{tikzpicture}
	\begin{pgfonlayer}{nodelayer}
		\node [style=none] (0) at (-10, 2) {\(R_1\)};
		\node [style=none] (1) at (-6, 2) {\(K_1\)};
		\node [style=none] (2) at (-2, 2) {\(L_1\)};
		\node [style=none] (3) at (0, -5) {\(G\)};
		\node [style=none] (4) at (2, 2) {\(L_2\)};
		\node [style=none] (5) at (6, 2) {\(K_2\)};
		\node [style=none] (6) at (10, 2) {\(R_2\)};
		\node [style=none] (7) at (10, -5) {\(H_2\)};
		\node [style=none] (8) at (6, -5) {\(D_2\)};
		\node [style=none] (9) at (-10, -1.5) {\(\overline{D}_2\)};
		\node [style=none] (10) at (-6, -1.5) {\(D\)};
		\node [style=none] (11) at (-8.25, 0.25) {\small\((21)\)};
		\node [style=none] (12) at (-4, 0.25) {\small\((11)\)};
		\node [style=none] (13) at (4, 0.25) {\small\((31)\)};
		\node [style=none] (14) at (8, 0.25) {\small\((41)\)};
		\node [style=none] (15) at (6, -1.5) {\small\(D\)};
		\node [style=none] (16) at (10, -1.5) {\small\(\overline{D}_1\)};
		\node [style=none] (17) at (2, -1.5) {\small\(D_1\)};
		\node [style=none] (18) at (-6, -5) {\small\(D_1\)};
		\node [style=none] (19) at (-10, -5) {\small\(H_1\)};
		\node [style=none] (20) at (-2, -1.5) {\small\(D_2\)};
		\node [style=none] (21) at (-8, -3.25) {\small\((22)\)};
		\node [style=none] (22) at (-4, -3.25) {\small\((12)\)};
		\node [style=none] (23) at (4, -3.25) {\small\((32)\)};
		\node [style=none] (24) at (8, -3.25) {\small\((42)\)};
	\end{pgfonlayer}
	\begin{pgfonlayer}{edgelayer}
		\draw [style=right arrow] (0.center) to (9.center);
		\draw [style=right arrow] (1.center) to (0.center);
		\draw [style=right arrow] (1.center) to (2.center);
		\draw [style=right arrow] (1.center) to (10.center);
		\draw [style=right arrow] (10.center) to (9.center);
		\draw [style=right arrow] (2.center) to node [fill=white, font={\tiny}] {\(m_1\)} (3.center);
		\draw [style=right arrow] (4.center) to node [fill=white, font={\tiny}] {\(m_2\)} (3.center);
		\draw [style=right arrow] (5.center) to (4.center);
		\draw [style=right arrow] (5.center) to (6.center);
		\draw [style=right arrow, bend left] (6.center) to (7.center);
		\draw [style=right arrow] (8.center) to (7.center);
		\draw [style=right arrow] (8.center) to (3.center);
		\draw [style=right arrow] (15.center) to (16.center);
		\draw [style=right arrow] (15.center) to (17.center);
		\draw [style=right arrow] (4.center) to (17.center);
		\draw [style=right arrow] (17.center) to (3.center);
		\draw [style=right arrow] (6.center) to (16.center);
		\draw [style=right arrow] (16.center) to (7.center);
		\draw [style=right arrow] (5.center) to (15.center);
		\draw [style=right arrow] (15.center) to (8.center);
		\draw [style=right arrow] (9.center) to (19.center);
		\draw [style=right arrow] (10.center) to (18.center);
		\draw [style=right arrow] (18.center) to (3.center);
		\draw [style=right arrow, in=180, out=0] (10.center) to (20.center);
		\draw [style=right arrow] (2.center) to (20.center);
		\draw [style=right arrow] (20.center) to (3.center);
		\draw [style=right arrow, bend right] (0.center) to (19.center);
		\draw [style=right arrow] (18.center) to (19.center);
	\end{pgfonlayer}
\end{tikzpicture}}
  \caption{Vertical pushout decomposition of Fig.~\ref{fig:parallel-independence}}
  \label{fig:church-rosser-decomposition}
\end{figure}
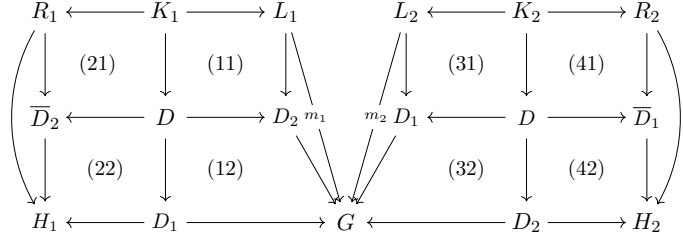
The pushouts \((21)\), \((41)\) are constructed using the \texttt{gluing} locale (see~\cite{Soeldner-Plump22a} for a detailed description). \Snippet{churchrosser-c21-c41}
The existence of \(\overline{D}_2 \to H_1\) and \(\overline{D}_1 \to H_2\) follows from the universal property of pushout \((21)\) and \((41)\), respectively.
Pushouts \((22)\) and \((42)\) are obtained using the pushout decomposition (see Lemma~\ref{lem:comp-decomp}).
This finishes the first stage of the proof.
\begin{figure}
  \center
  \begin{tabular}{c}
    \scalebox{0.8}{\begin{tikzpicture}
	\begin{pgfonlayer}{nodelayer}
		\node [style=none] (0) at (-10, 2) {\small\(L_1\)};
		\node [style=none] (1) at (-6, 2) {\small\(K_1\)};
		\node [style=none] (2) at (-2, 2) {\small\(R_1\)};
		\node [style=none] (3) at (0, -5) {\small\(H_1\)};
		\node [style=none] (4) at (2, 2) {\small\(L_2\)};
		\node [style=none] (5) at (6, 2) {\small\(K_2\)};
		\node [style=none] (6) at (10, 2) {\small\(R_2\)};
		\node [style=none] (7) at (10, -5) {\small\(G'\)};
		\node [style=none] (8) at (6, -5) {\small\(\overline{D}_2\)};
		\node [style=none] (9) at (-10, -1.5) {\small\(D_2\)};
		\node [style=none] (10) at (-6, -1.5) {\small\(D\)};
		\node [style=none] (11) at (-8.25, 0.25) {\small\((11)\)};
		\node [style=none] (12) at (-4, 0.25) {\small\((21)\)};
		\node [style=none] (13) at (4, 0.25) {\small\((31)\)};
		\node [style=none] (14) at (8, 0.25) {\small\((41)\)};
		\node [style=none] (15) at (6, -1.5) {\small\(D\)};
		\node [style=none] (16) at (10, -1.5) {\small\(\overline{D}_1\)};
		\node [style=none] (17) at (2, -1.5) {\small\(D_1\)};
		\node [style=none] (18) at (-6, -5) {\small\(D_1\)};
		\node [style=none] (19) at (-10, -5) {\small\(G\)};
		\node [style=none] (20) at (-2, -1.5) {\small\(\overline{D}_2\)};
		\node [style=none] (21) at (-8, -3.25) {\small\((12)\)};
		\node [style=none] (22) at (-4, -3.25) {\small\((22)\)};
		\node [style=none] (23) at (4, -3.25) {\small\((22)\)};
		\node [style=none] (24) at (8, -3.25) {\small\((5)\)};
	\end{pgfonlayer}
	\begin{pgfonlayer}{edgelayer}
		\draw [style=right arrow] (0.center) to (9.center);
		\draw [style=right arrow] (1.center) to (0.center);
		\draw [style=right arrow] (1.center) to (2.center);
		\draw [style=right arrow] (1.center) to (10.center);
		\draw [style=right arrow] (10.center) to (9.center);
		\draw [style=right arrow] (2.center) to (3.center);
		\draw [style=right arrow] (4.center) to node [fill=white, font={\tiny}] {\(m'_2\)} (3.center);
		\draw [style=right arrow] (5.center) to (4.center);
		\draw [style=right arrow] (5.center) to (6.center);
		\draw [style=right arrow, bend left] (6.center) to (7.center);
		\draw [style=right arrow] (8.center) to (7.center);
		\draw [style=right arrow] (8.center) to (3.center);
		\draw [style=right arrow] (15.center) to (16.center);
		\draw [style=right arrow] (15.center) to (17.center);
		\draw [style=right arrow] (4.center) to (17.center);
		\draw [style=right arrow] (17.center) to (3.center);
		\draw [style=right arrow] (6.center) to (16.center);
		\draw [style=right arrow] (16.center) to (7.center);
		\draw [style=right arrow] (5.center) to (15.center);
		\draw [style=right arrow] (15.center) to (8.center);
		\draw [style=right arrow] (9.center) to (19.center);
		\draw [style=right arrow] (10.center) to (18.center);
		\draw [style=right arrow] (18.center) to (3.center);
		\draw [style=right arrow] (10.center) to (20.center);
		\draw [style=right arrow] (2.center) to (20.center);
		\draw [style=right arrow] (20.center) to (3.center);
		\draw [style=right arrow, bend right] (0.center) to node [fill=white, font={\tiny}] {\(m_1\)} (19.center);
		\draw [style=right arrow] (18.center) to (19.center);
	\end{pgfonlayer}
\end{tikzpicture}} \\
    \\
    \scalebox{0.8}{\begin{tikzpicture}
	\begin{pgfonlayer}{nodelayer}
		\node [style=none] (0) at (-10, 2) {\small\(L_2\)};
		\node [style=none] (1) at (-6, 2) {\small\(K_2\)};
		\node [style=none] (2) at (-2, 2) {\small\(R_2\)};
		\node [style=none] (3) at (0, -5) {\small\(H_2\)};
		\node [style=none] (4) at (2, 2) {\small\(L_1\)};
		\node [style=none] (5) at (6, 2) {\small\(K_1\)};
		\node [style=none] (6) at (10, 2) {\small\(R_1\)};
		\node [style=none] (7) at (10, -5) {\small\(G'\)};
		\node [style=none] (8) at (6, -5) {\small\(\overline{D}_1\)};
		\node [style=none] (9) at (-10, -1.5) {\small\(D_1\)};
		\node [style=none] (10) at (-6, -1.5) {\small\(D\)};
		\node [style=none] (11) at (-8.25, 0.25) {\small\((31)\)};
		\node [style=none] (12) at (-4, 0.25) {\small\((41)\)};
		\node [style=none] (13) at (4, 0.25) {\small\((11)\)};
		\node [style=none] (14) at (8, 0.25) {\small\((21)\)};
		\node [style=none] (15) at (6, -1.5) {\small\(D\)};
		\node [style=none] (16) at (10, -1.5) {\small\(\overline{D}_2\)};
		\node [style=none] (17) at (2, -1.5) {\small\(D_2\)};
		\node [style=none] (18) at (-6, -5) {\small\(D_2\)};
		\node [style=none] (19) at (-10, -5) {\small\(G\)};
		\node [style=none] (20) at (-2, -1.5) {\small\(\overline{D}_1\)};
		\node [style=none] (21) at (-8, -3.25) {\small\((32)\)};
		\node [style=none] (22) at (-4, -3.25) {\small\((42)\)};
		\node [style=none] (23) at (4, -3.25) {\small\((42)\)};
		\node [style=none] (24) at (8, -3.25) {\small\((5)\)};
	\end{pgfonlayer}
	\begin{pgfonlayer}{edgelayer}
		\draw [style=right arrow] (0.center) to (9.center);
		\draw [style=right arrow] (1.center) to (0.center);
		\draw [style=right arrow] (1.center) to (2.center);
		\draw [style=right arrow] (1.center) to (10.center);
		\draw [style=right arrow] (10.center) to (9.center);
		\draw [style=right arrow] (2.center) to (3.center);
		\draw [style=right arrow] (4.center) to node [fill=white, font={\tiny}] {\(m'_1\)} (3.center);
		\draw [style=right arrow] (5.center) to (4.center);
		\draw [style=right arrow] (5.center) to (6.center);
		\draw [style=right arrow, bend left] (6.center) to (7.center);
		\draw [style=right arrow] (8.center) to (7.center);
		\draw [style=right arrow] (8.center) to (3.center);
		\draw [style=right arrow] (15.center) to (16.center);
		\draw [style=right arrow] (15.center) to (17.center);
		\draw [style=right arrow] (4.center) to (17.center);
		\draw [style=right arrow] (17.center) to (3.center);
		\draw [style=right arrow] (6.center) to (16.center);
		\draw [style=right arrow] (16.center) to (7.center);
		\draw [style=right arrow] (5.center) to (15.center);
		\draw [style=right arrow] (15.center) to (8.center);
		\draw [style=right arrow] (9.center) to (19.center);
		\draw [style=right arrow] (10.center) to (18.center);
		\draw [style=right arrow] (18.center) to (3.center);
		\draw [style=right arrow] (10.center) to (20.center);
		\draw [style=right arrow, in=90, out=-90] (2.center) to (20.center);
		\draw [style=right arrow] (20.center) to (3.center);
		\draw [style=right arrow, bend right] (0.center) to node [fill=white, font={\tiny}] {\(m_2\)} (19.center);
		\draw [style=right arrow] (18.center) to (19.center);
	\end{pgfonlayer}
\end{tikzpicture}} \\
  \end{tabular}
  \caption{Rearranged pushouts of Fig.~\ref{fig:church-rosser-decomposition}}
  \label{fig:church-rosser-rearrange}
\end{figure}
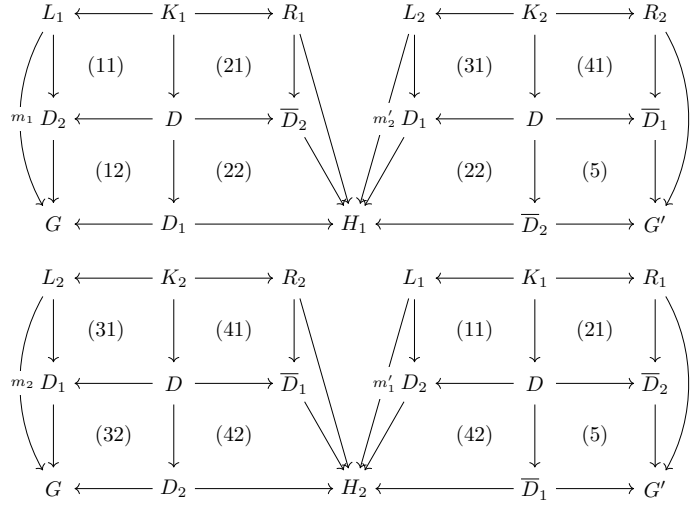
The second stage rearranges the pushouts, as depicted in Fig.~\ref{fig:church-rosser-rearrange}, such that we obtain two
direct derivations \(H_1 \Rightarrow_{p_2,m'_2} G'\) and \(H_2 \Rightarrow_{p_1,m'_1} G'\).
Here, we compose the pushouts from stage 1 (see Lemma~\ref{lem:comp-decomp}). Exemplary, the pushouts \((31)\) and \((22)\)
are composed in Isabelle/HOL.
\Snippet{churchrosser-31-22}
The final pushout \((5)\) is constructed and the pushouts are rearranged and vertically composed as depicted in Fig.~\ref{fig:church-rosser-rearrange}.
Isabelle is able to discharge the goal at this point automatically as we instantiated all required locales.
The sequential independence follows from the construction. This finishes the second stage of the proof.
\end{proof}
In developing the proofs, we manually laid out the proof structure by carefully instantiating the locales with the correct parameters in the right order. This process was challenging due to the large number of type parameters involved. However, Isabelle's typechecker provided some support by catching errors when we attempted to instantiate the parameters in the wrong order, helping to guide the proof development process.

Despite this assistance, the overall proof structure was developed manually, with tactics being employed only for specific subgoals. The manual approach allowed us to maintain control over the proof flow and ensured that the reasoning remained clear and understandable. While the type parameters added complexity to the proof, they also provided a valuable safeguard against potential inconsistencies that could arise from incorrectly ordered instantiations which could lead to long debugging.

\section{Related Work}\label{sec:related-work}

Isabelle/HOL was used by Strecker for interactive reasoning about graph transformations \cite{Strecker18a}. A major difference to our work is that he introduces a form of graph transformation that does not fit with any of the established approaches, such as the double-pushout approach. As a consequence, his framework cannot draw on existing theory. Another difference is that \cite{Strecker18a} focuses on verifying first-order properties of some form of graph programs while the current paper is concerned with formalising and proving fundamental results of the DPO theory.
Strecker's formalisation fixes node and edge identifiers as natural numbers, while we keep them abstract. Similar to our development, Isabelle's locale mechanism is employed. 

Our formalisation of graphs follows the work of Noschinski \cite{Noschinski2015}, where records are used to group components and locales to enforce properties such as the well-formedness of graphs or morphisms. The main objective of \cite{Noschinski2015} is to formalise and prove fundamental results of classical graph theory, such as Kuratowski's theorem.

Stark~\cite{stark-afp16} developed an "object-free" definition of categories in Isabelle/HOL, which simplifies the specification of a category and allows functors and natural transformations to be defined as functions satisfying certain axioms. Here, the primary focus is on efficiently representing abstract algebra in higher-order logic. A limitation of this approach is that a special zero element had to be defined. As a result, the formalization prevents the construction of a discrete category over the entire universe, and instead requires a large enough base type. We followed a more traditional approach in formalizing the category theory-related concepts but had to work around the single universe and explicit quantification limitations discussed earlier.

da Costa Cavalheiro et al. \cite{Cavalheiro17a} use the Event-B specification method and its associated theorem prover to reason about double-pushout and single-pushout graph transformations, where rules can have attributes and negative application conditions. Event-B is based on first-order logic and typed set theory. Different from our approach, \cite{Cavalheiro17a} gives only a non-formalised proof for the equivalence between the abstract definition of pushouts and the set-theoretic construction. In contrast, we formalise both the abstract and the operational view and prove their correspondence using Isabelle/HOL. 
As Event-B is based on first-order logic, the properties that can be expressed and verified are quite limited. For it is known that non-local properties of finite graphs cannot be specified in first-order logic \cite{Libkin04a}. This restriction does not apply to our formalisation, as we can make full use of higher-order logic.

\section{Conclusion}\label{sec:discussion}
In this paper, we have significantly advanced our formalisation of basic DPO graph transformation theory using the Isabelle proof assistant, relying on its higher-order logic instantiation. We have thoroughly elaborated on two pivotal results of the DPO theory: the uniqueness of direct derivations and the Church-Rosser theorem. Our discussion includes a series of critical lemmata that build towards the final proof.

We also delve into the technical aspects of our formalisation, including the application of total and partial functions, finite sets, and maps (\texttt{fset} and \texttt{fmap}). Our approach strikes a balance between expressiveness and automation support. While the finite sets (maps) offer advantages in certain proofs, their lack of a comprehensive lemma set needs significant additional implementation and proof work on our part. Moreover, the idea of using a unified identifier for both nodes and edges, though initially appealing, led to increased complexity in our construction, as detailed in Section~\ref{sec:DPO}. This complexity extended to the entire process of gluing and deletion, affecting the properties we proved using this construction.

Looking forward, our aim is to broaden our methodology to include attributed DPO graph transformation, as outlined in~\cite{Hristakiev-Plump16a}. Our long-term goal is to develop a GP\,2 proof assistant, enabling the interactive verification of individual graph programs. This tool may leverage the proof calculus presented in~\cite{Wulandari-Plump21a}, offering a significant advancement in the field of graph program verification.

\section*{Acknowledgment}
  \noindent We are grateful to Brian Courthoute, Annegret Habel, and Thomas T\"urk for discussions on the topics of this paper.

\bibliographystyle{alphaurl}
\bibliography{main}

\end{document}